\newif\ifarxiv
\arxivtrue

\ifarxiv
  \documentclass[sigconf, nonacm]{acmart}
\else
  \documentclass[sigconf]{acmart}
\fi

\usepackage{balance}
\usepackage{enumerate}
\usepackage{thm-restate}

\usepackage{diego-macro} 

\usepackage[backgroundcolor=orange!50, textsize=small]{todonotes}
\definecolor{green}{RGB}{0,120,0}
\definecolor{hlyellow}{RGB}{250, 250, 190}
\definecolor{diegoeditcolor}{RGB}{210,210,255}
\definecolor{migueleditcolor}{RGB}{210,255,210}

\definecolor{light-gray}{gray}{0.9}
\newcommand{\proofcase}[1]{\noindent\colorbox{light-gray}{#1}~~}

\newcommand{\arxivurl}{https://arxiv.org/abs/xxxxxx}
\newcommand{\fullversion}{\href{\arxivurl}{full version}}

\renewcommand{\int}{\mathbb{Z}}
\newcommand{\pitwo}{\ensuremath{\Pi^p_2}\xspace}
\newcommand{\sigmatwo}{\ensuremath{\Sigma^p_2}\xspace}

\newcommand{\vars}{\textit{vars}}

\newcommand{\CQ}{\ensuremath{\textsf{CQ}}\xspace}
\newcommand{\CRPQ}{\ensuremath{\textsf{CRPQ}}\xspace}
\newcommand{\CRPQfin}{\ensuremath{\textsf{CRPQ}^{\textsf{fin}}}\xspace}
\newcommand{\injto}{\xrightarrow{\textit{inj}}}
\newcommand{\ainjto}{\xrightarrow{\textit{a-inj}}}

\newcommand{\collapse}{\equiv}
\newcommand{\Exp}[1]{\textup{\textsc{Exp}}^{\mathit{#1}}}
\newcommand{\ani}{\textit{a-inj}\xspace}

\newcommand{\qni}{\textit{q-inj}\xspace}

\newcommand{\decpb}[3]{%
\begin{center}
\begin{tabular}{ |r p{6.5cm}|  }
 \hline
 \textbf{\sc Problem} & #1\\
 \hline
 \textbf{\sc Given}   & #2\\
 \textbf{\sc Question}& #3\\
 \hline
\end{tabular}
\end{center}
}

\copyrightyear{2023}
\acmYear{2023}
\setcopyright{licensedothergov}\acmConference[PODS '23]{Proceedings of the 42nd ACM SIGMOD-SIGACT-SIGAI Symposium on Principles of Database Systems}{June 18--23, 2023}{Seattle, WA, USA}
\acmBooktitle{Proceedings of the 42nd ACM SIGMOD-SIGACT-SIGAI Symposium on Principles of Database Systems (PODS '23), June 18--23, 2023, Seattle, WA, USA}
\acmPrice{15.00}
\acmDOI{10.1145/3584372.3588664}
\acmISBN{979-8-4007-0127-6/23/06}

\newtheorem{claim}{Claim}[section]
\newtheorem{remark}{Remark}[section]

\AtBeginDocument{%
  \providecommand\BibTeX{{%
    \normalfont B\kern-0.5em{\scshape i\kern-0.25em b}\kern-0.8em\TeX}}}

\begin{document}
\fancyhead{}

\title{Conjunctive Regular Path Queries under Injective Semantics}
\ifarxiv\else\titlenote{Full version available at \url{\arxivurl}}\fi

\author{Diego Figueira}
\email{diego.figueira@cnrs.fr}
\orcid{0000-0003-0114-2257}
\affiliation{%
  \institution{Univ. Bordeaux, CNRS,  \\Bordeaux INP, LaBRI, UMR 5800}
  \city{F-33400, Talence}
  \country{France}
  \postcode{F-33400}
}
\author{Miguel Romero}
\email{miguel.romero.o@uai.cl}
\orcid{0000-0002-2615-6455}
\affiliation{%
  \institution{Faculty of Engineering and Science, \\Universidad Adolfo Ib\'a\~nez}
  \country{Chile}
  \postcode{}
}
\thanks{Diego Figueira is partially supported by ANR QUID, grant ANR-18-CE400031.
Miguel Romero is funded by Fondecyt grant 11200956, the Data Observatory Foundation, and the National Center for Artificial Intelligence CENIA FB210017, Basal ANID}

\begin{abstract}
	We introduce injective semantics for Conjunctive Regular Path Queries (CRPQs), and study their fundamental properties.
We identify two such semantics: atom-injective and query-injective semantics, both defined in terms of injective homomorphisms.
These semantics are natural generalizations of the well-studied class of RPQs under simple-path semantics to the class of CRPQs.
We study their evaluation and containment problems, providing useful characterizations for them, and we pinpoint the complexities of these problems. 
Perhaps surprisingly, we show that containment for CRPQs becomes undecidable for atom-injective semantics, and \pspace-complete for query-injective semantics, in contrast to the known \expspace-completeness result for the standard semantics.
The techniques used differ significantly from the ones known for the standard semantics, and new tools tailored to injective semantics are needed. 
We complete the picture of complexity by investigating, for each semantics, the containment problem for the main subclasses of CRPQs, namely Conjunctive Queries and CRPQs with finite languages.

\end{abstract}


\maketitle

\section{Introduction}
\label{sec:introduction}

Graph databases are important for many applications nowadays \cite{biggraphs2021,survey-graphdbs2017}. 
In a nutshell, graph databases are abstracted as edge-labeled directed graphs, 
where nodes represent entities and labeled edges represent relations between these entities.
A fundamental way to query graph databases is by finding \emph{patterns} on the interrelation between entities. 
In this respect, a central querying mechanism for modern graph query languages is that of \emph{regular path queries} (RPQs). 
RPQs provide a simple form of recursion tailored to discovering entities linked by paths with certain properties.
These are queries of the form $x \xrightarrow{L} y$, where $L$ is a regular expression on the alphabet of database edge labels. Such a query returns all pairs of nodes $(u,v)$ in the database such that there is a (directed) path from $u$ to $v$ whose label matches $L$.

The closure under conjunction and existential quantification of RPQ yields what is known as \emph{Conjunctive RPQs} (CRPQs).
Indeed, CRPQs can be understood as the generalization of conjunctive queries with this simple form of recursion. 
CRPQs are part of SPARQL~\cite{sparql1.1}, the W3C standard for querying RDF data (a widespread format for graph databases). Some examples of RDF databases are well-known knowledge bases such as DBpedia and Wikidata. 
For example, (C)RPQs are popular for querying Wikidata \cite{MalyshevKGGB18,BonifatiMT-www19}. 
More generally, CRPQs are basic building blocks for querying graph-structured databases \cite{DBLP:conf/pods/Baeza13,survey-graphdbs2017}. They are part of G-core \cite{DBLP:conf/sigmod/AnglesABBFGLPPS18} and Cypher \cite{DBLP:conf/sigmod/FrancisGGLLMPRS18}, the latter being the query language of Neo4j, which is currently one of the most popular commercial graph databases. 
They are also part of the ongoing standardization effort GQL for graph query languages \cite{GQL-sigmod2022,GQLsite}. 

\paragraph{Alternative semantics.}

The semantics for the evaluation of an RPQ $x \xrightarrow{L} y$ as above often assumes that any, arbitrary, (directed) path from $u$ to $v$ is allowed as long as it satisfies the regular property $L$. However, there are alternative semantics in which one may restrict the path to have no repeated vertices (\aka~\emph{simple path}), or no repeated edges (\aka~\emph{trail}). 
In this way, only a \emph{finite number} of paths need to be considered.
As a matter of fact, these alternative semantics have received considerable attention both in practice and from the database theory community.
Indeed, they are part of Neo4j's Cypher query language and are included in GQL as possible ways to evaluate RPQs. 
Furthermore, these alternative semantics have been extensively studied in the literature from the late 80s onwards \cite{DBLP:conf/sigmod/CruzMW87,DBLP:journals/siamcomp/MendelzonW95,DBLP:journals/jcss/BaganBG20,DBLP:journals/tods/MartensT19,DBLP:journals/tods/LosemannM13,MartensPopp-pods22}. 
However, rather surprisingly, this body of research has focused mainly on RPQs, 
leaving the case of CRPQs under alternative semantics essentially unexplored.

\paragraph{Contribution.}
We introduce \emph{injective} semantics for CRPQs, which generalize the simple-path semantics for RPQs, and we investigate the fundamental properties of CRPQs under these semantics. 
Concretely, we identify two possible natural semantics:

(1) \textit{Atom-injective.}
The first semantics is to require that each CRPQ \emph{atom} is interpreted with the simple-path semantics of RPQs, that is, atoms of the form $x \xrightarrow{L} y$ must be mapped to \textit{simple paths} and atoms of the form $x \xrightarrow{L} x$ must be mapped to \textit{simple cycles}.
In particular, there is no requirement that paths from different atoms be disjoint. 
Under this semantics, a Boolean CRPQ like $Q=\exists x,y,z\, (x \xrightarrow{(a+b)^+} y \land x \xrightarrow{(b+c)^+} z)$
holds true in the graph database consisting of a directed path of $b$'s by mapping $x$ to the first node of the path and $y$ and $z$ to the last node.

(2) \textit{Query-injective.}
The alternative is to consider a more restrictive semantics in which paths corresponding to different atoms must be mapped to different nodes, and hence there cannot be repeated nodes neither in paths nor between paths. 
This semantics generalizes both RPQ under simple path semantics and Conjunctive Queries under injective semantics.
In this case, the query $Q$ above may only be true if the database contains two simple paths starting in the same node with the corresponding language which are disjoint (except for the origin).

We call the former \defstyle{atom-injective semantics} since the ``no repeated nodes'' condition is required separately for each atom, and the latter \defstyle{query-injective semantics} since injectivity is required for the query as a whole. 
The three semantics (standard, atom-injective, query-injective) form a hierarchy, where query-injective is the most restrictive and standard semantics is the least.

We first consider the \emph{evaluation problem} for CRPQs, that is, checking whether a tuple $\bar{v}$ belongs to the results of a query $Q$ over a particular graph database $G$. 
For standard semantics, this problem is \np-complete in \emph{combined complexity}, that is, when both query and database are part of the input,  
and \nl-complete in \emph{data complexity}, that is, when the query is considered to be fixed. 
Under both injective semantics, the evaluation problem remains \np-complete in combined complexity, 
and becomes \np-complete in data complexity. 
This follows straight from the fact that RPQ evaluation under simple-path semantics is \np-complete, even for very simple RPQs~\cite{DBLP:journals/siamcomp/MendelzonW95,DBLP:conf/pods/Baeza13}. 

We therefore turn our attention to the \emph{containment problem}, which is the main focus of this paper. 
This problem asks whether every result of a query $Q_1$ is also returned by a query $Q_2$, independently of the underlying database. Checking containment is one of the most basic static analysis tasks, and it can be a means for query optimization.
The containment problem for CRPQs is known to be \expspace-complete \cite{CGLV00,Florescu:CRPQ} under standard semantics, and the study of containment has been extended to queries with restricted shapes \cite{Figueira20} or restricted regular expressions \cite{FigueiraGKMNT20}.

Rather surprisingly, as we show, the hierarchy of the three semantics (standard, atom-injective, query-injective) is not reflected in the complexity of the containment problem: the most restrictive semantics (query-injective) is \pspace-complete, the least restrictive one (standard) is \expspace-complete, and the middle one (atom-injective) is undecidable.
These results for atom- and query-injective semantics are the main technical contributions of this paper. We complete the picture with a thorough study on the containment problem for the two main subclasses of CRPQs: Conjunctive Queries, and CRPQs with no Kleene star operator. We provide complexity completeness results for all possible combinations (\cf~Figure~\ref{fig:summary}).
\begin{figure*}
                \begin{tabular}{|r|ccccc|}
                    \hline
                        &\CQ/\CQ 
                        &\CQ/\CRPQ  
                        & \CRPQ/\CQ 
                        & \CQ/\CRPQfin
                        &\CRPQfin/\CQ
                        \\
        \hline
                    standard
                        &\np-c \cite{DBLP:conf/stoc/ChandraM77}  
                        &\np-c ($\dag$)
                        &\pitwo-c ($\ddag$)
                        &\np-c ($\dag$) 
                        &\pitwo-c ($\mathsection$,$\ddag$)
                        \\
                    query-injective
                        &\np-c \ifarxiv(\ref{prop:cq-crpq-cq-inj})\else(\ref{thm:other} (1))\fi
                        & \np-c \ifarxiv(\ref{prop:cq-crpq-cq-inj})\else(\ref{thm:other} (1))\fi
                        &\pitwo-c \ifarxiv(\ref{thm:crpq-cq-pitwo-hard-qni},\ref{prop:crpq-cq-all-upper})\else(\ref{thm:crpq-cq-pitwo-hard-qni},\ref{thm:other} (4))\fi
                        &\np-c \ifarxiv(\ref{prop:cq-crpq-cq-inj})\else(\ref{thm:other} (1))\fi
                        &\pitwo-c \ifarxiv(\ref{thm:crpq-cq-pitwo-hard-qni},\ref{prop:crpq-cq-all-upper})\else(\ref{thm:crpq-cq-pitwo-hard-qni},\ref{thm:other} (4))\fi
                        \\
                    atom-injective
                        & \np-c \ifarxiv(\ref{cor:cq-cq-sp})\else(\ref{thm:other} (2))\fi
                        & \pitwo-c \ifarxiv(\ref{thm:crpqfin-crpqfin-lower},\ref{prop:crpqfin-crpq-upper})\else(\ref{thm:crpqfin-crpqfin-lower},\ref{thm:other} (7))\fi
                        &\pitwo-c \ifarxiv(\ref{prop:crpq-cq-ani-hard},\ref{prop:crpq-cq-all-upper})\else(\ref{thm:other} (3),\ref{thm:other} (4))\fi
                        &\pitwo-c \ifarxiv(\ref{thm:crpqfin-crpqfin-lower},\ref{prop:crpqfin-crpq-upper})\else(\ref{thm:crpqfin-crpqfin-lower},\ref{thm:other} (7))\fi
                        &\pitwo-c \ifarxiv(\ref{prop:crpq-cq-ani-hard},\ref{prop:crpq-cq-all-upper})\else(\ref{thm:other} (3),\ref{thm:other} (4))\fi
                        \\
                        \hline
                \end{tabular}
        
                \smallskip
        
                \hspace{.440cm}
                \begin{tabular}{|r|cccc|}
                    \hline
                        &\CRPQ/\CRPQfin
                        &\CRPQfin/\CRPQ
                        & \CRPQfin/\CRPQfin
                        & \CRPQ/\CRPQ
                        \\
        \hline
                    standard       
                        &\pspace-c \ifarxiv(\ref{prop:crpq-crpqfin-pspace-h},\ref{prop:crpq-crpqfin-pspace-st})\else(\ref{thm:other} (5),\ref{thm:other} (6))\fi
                        &\pitwo-c \ifarxiv($\mathsection$,\ref{prop:crpqfin-crpq-upper})\else($\mathsection$,\ref{thm:other} (7))\fi
                        &\pitwo-c \ifarxiv($\mathsection$,\ref{prop:crpqfin-crpq-upper})\else($\mathsection$,\ref{thm:other} (7))\fi
                        &\expspace-c ($\mathdollar$,\textcent) 
                        \\
                    query-injective
                        &\pspace-c \ifarxiv(\ref{prop:crpq-crpqfin-pspace-h},\ref{thm:crpq-crpq-nodeinj-pspace})\else(\ref{thm:other} (5),\ref{thm:crpq-crpq-nodeinj-pspace})\fi
                        &\pitwo-c \ifarxiv(\ref{thm:crpq-cq-pitwo-hard-qni},\ref{prop:crpqfin-crpq-upper})\else(\ref{thm:crpq-cq-pitwo-hard-qni},\ref{thm:other} (7))\fi
                        & \pitwo-c \ifarxiv(\ref{thm:crpq-cq-pitwo-hard-qni},\ref{prop:crpqfin-crpq-upper})\else(\ref{thm:crpq-cq-pitwo-hard-qni},\ref{thm:other} (7))\fi
                        &\pspace-c \ifarxiv(\ref{prop:crpq-crpqfin-pspace-h},\ref{thm:crpq-crpq-nodeinj-pspace})\else(\ref{thm:other} (5),\ref{thm:crpq-crpq-nodeinj-pspace})\fi \\
                    atom-injective
                        &undec.\ (\ref{theo:undec-atom-ni})
                        &\pitwo-c \ifarxiv(\ref{thm:crpqfin-crpqfin-lower},\ref{prop:crpqfin-crpq-upper})\else(\ref{thm:crpqfin-crpqfin-lower},\ref{thm:other} (7))\fi
                        &\pitwo-c \ifarxiv(\ref{thm:crpqfin-crpqfin-lower},\ref{prop:crpqfin-crpq-upper})\else(\ref{thm:crpqfin-crpqfin-lower},\ref{thm:other} (7))\fi
                        & undec.\ (\ref{theo:undec-atom-ni})
                        \\
                        \hline
                \end{tabular}
                \hspace{.5cm}
        
                \vspace{.1cm}
        
                {\small 
                    $\dag$: \cite[Thm~4.2]{FigueiraGKMNT20} 	
                        ~~
                    $\mathsection$: \cite[Thm~4.3]{FigueiraGKMNT20} 
                        ~~
                    $\ddag$: \cite[Thm~4.4]{FigueiraGKMNT20} 
                        ~~
                    $\mathdollar$: \cite[Thm.~6]{CGLV00} 
                        ~~
                    \textcent: \cite[Thm.~4.8]{Florescu:CRPQ}
                }
    \vspace{-.4cm}
            \caption{Complexity of the containment problem under standard, query-injective, and atom-injective semantics. Numbers in brackets reference proposition/theorem numbers\ifarxiv ~(some of them in Appendix~\ref{app:other-results}).\else. All missing proofs can be found in the \fullversion.\fi}
            \label{fig:summary}
    \end{figure*}
Our main results require the development of novel techniques, which yield insights on the 
subtle difficulties for handling static analysis under these semantics.

\paragraph{Organization.}
After a preliminary section \S \ref{sec:preliminaries}, we define and characterize the evaluation and containment problems in \S \ref{sec:eval} and \S \ref{sec:containment}, respectively. 
We study the containment problem for arbitrary CRPQs in \S \ref{sec:unrestricted}, and for subclasses of CRPQs in \S \ref{sec:subclasses}. We conclude with \S \ref{sec:discussion}.
\ifarxiv Omitted or sketched proofs can be found in the Appendix.\fi

\section{Preliminaries}
\label{sec:preliminaries}

We assume familiarity with regular languages, regular expressions and non-deterministic finite automata (NFA). 
We often blur the distinction between a regular expression and the language it defines; similarly for NFAs. 

\paragraph{Graph databases and paths.} 
A \defstyle{graph database} over a finite alphabet $\A$ is a finite edge-labeled graph $G = (V, E)$ over $\A$, where $V$ is a finite set of vertices and $E \subseteq V \times \A \times V$ is the set of labeled edges (or simply \defstyle{edges}). We write $u \xrightarrow{a} v$ to denote an edge $(u,a,v) \in E$. 
A \defstyle{path} from $u$ to $v$ in a graph database $G=(V,E)$ over alphabet $\A$ is a (possibly empty) sequence 
$\pi = v_0 \xrightarrow{a_1} v_1,\, v_1 \xrightarrow{a_2} v_2,\, \dots\, ,v_{k-1}  \xrightarrow{a_k} v_k$ 
of edges of $G$, where $k \geq 0$,  $u=v_0$ and $v=v_k$. 
An \defstyle{internal node} of such a path is any node $v_i$ with $0 < i < k$.
The {\em label} of $\pi$ is the word $a_1 \dots a_k \in \A^*$. When $k = 0$ the label of $\pi$ is the empty word $\epsilon$. 
We say that $\pi$ is a \defstyle{simple path} if all the nodes $v_i$ are pairwise distinct,
and a \defstyle{simple cycle} if $v_0=v_k$ and all the nodes $v_i$ (for $i<k$) are pairwise distinct.

\paragraph{Conjunctive queries and homomorphisms.} In the setting of graph databases, a \defstyle{conjunctive query} (CQ) $Q$ over a finite alphabet $\A$ is an expression 
$Q(x_1,\dots, x_n) = A_1\land \dots\land A_m$, for $m\geq 0$, where $(x_1,\dots, x_n)$ is a tuple of variables, and each $A_i$ is an \defstyle{atom} of the form $x \xrightarrow{a} y$, for variables $x$ and $y$, and $a\in \A$. 
We denote by $\vars(Q)$ the set of variables appearing in $Q$. We often write $Q(\bar{x})$ instead of $Q$ to emphasize 
the tuple $\bar{x} = (x_1,\dots, x_n)$ of \defstyle{free variables} of $Q$. We assume that the free variables $x_i$ are not necessarily distinct. 
The variables of $Q$ which are not in $\set{x_1, \dots, x_n}$ are (implicitly) existentially quantified.
As usual, if $\bar{x}$ is empty, we say that the CQ $Q$ is \defstyle{Boolean}.
Note that every CQ can be seen as a graph database (each atom is an edge), hence, by slightly abusing notation, 
we sometimes use graph database terminology for CQs.

A \defstyle{homomorphism} $h$ from a CQ $Q(\bar{x})$ to a graph database $G=(V,E)$ is a mapping from $\vars(Q)$ to $V$ such that $h(x) \xrightarrow{a} h(y)$ belongs to $E$ for each atom $x \xrightarrow{a} y$ of $Q$. 
We say that $h$ is \defstyle{injective} if additionally we have $h(x)\neq h(y)$ for all pairs of distinct variables $x$ and $y$.
We write $Q\to G$ if there is a homomorphism from $Q$ to $G$ and $h:Q\to G$ if $h$ is such a homomorphism. 
Similarly, for a tuple $\bar{v}$, we write $Q\to (G, \bar{v})$ if there is a homomorphism $h$ from $Q$ to $G$  such that $h(\bar{x}) = \bar{v}$ and 
$h: Q\to (G, \bar{v})$ to make such $h$ explicit. 
We use similar notation for injective homomorphisms, replacing $\to$ by $\injto$.
Homomorphisms between CQs are essentially defined as before with the  difference that free variables are mapped to free variables.
That is, given two CQs $Q_1(\bar{x}_1)$, $Q_2(\bar{x}_2)$, we have $h: Q_1 \to Q_2$  if $h:Q_1 \to (G,\bar{x}_2)$, and $h: Q_1 \injto Q_2$ if $h:Q_1 \injto (G,\bar{x}_2)$, where $G$ is the graph database denoted by $Q_2$.

We also work with CQs with \defstyle{equality atoms}, which are queries of the form $Q(\bar{x}) = P \land I$, 
where $P$ is a CQ (without equality atoms) and $I$ is a conjunction of equality atoms of the form 
$x=y$ (the variables $x$ and $y$ may not belong to $\vars(P)$). 
Again, we denote by $\vars(Q)$ the set of variables appearing in $Q$.
We define the binary relation $=_Q$ over $\vars(Q)$ to be the reflexive-symmetric-transitive closure of the binary relation $\{(x, y): \text{$x=y$ is an equality atom in $Q$}\}$. 
In other words, we have $x=_Q y$ if the equality $x=y$ is forced by the equality atoms of $Q$. 
Note that every CQ with equality atoms $Q(\bar{x}) = P \land I$ is equivalent to a CQ without equality atoms  $Q^{\collapse}$, 
which is obtained from $Q$ by collapsing each equivalence class of the relation $=_Q$ into a single variable. 
This transformation gives us a \emph{canonical} renaming, which we always denote by $\Phi$, from $\vars(Q)$ to $\vars(Q^{\collapse})$, defined by $\Phi(x) = C$, where $C$ is the equivalence class containing $x$.  
In particular, the tuple of free variables of $Q^{\collapse}$ is $\Phi(\bar{x})$.

\paragraph{Conjunctive regular path queries.} A \defstyle{conjunctive regular path query} (CRPQ) $Q$ over a finite alphabet $\A$ is an expression 
$Q(x_1,\dots, x_n) = A_1\land \dots\land A_m$, for $m\geq 0$, where each $A_i$ is an atom of the form $x \xrightarrow{L} y$, for variables $x$ and $y$, and a regular expression $L$ over $\A$.  
As before, we denote by $\vars(Q)$ the set of variables of $Q$ and often write $Q(\bar{x})$ instead of $Q$ where $\bar{x} = (x_1,\dots, x_n)$ is the tuple of (not necessarily distinct) free variables of $Q$. 
If the tuple $\bar{x}$ is empty, we say that  $Q$ is \defstyle{Boolean}. 
The class of CRPQs extends the class of CQs and the well-studied class of \defstyle{regular path queries} (RPQs). 
Indeed, each CQ can be seen as a CRPQ  where the regular expressions are single labels from $\A$. 
On the other hand, an RPQ corresponds to a CRPQ of the form $Q(x,y)= x \xrightarrow{L} y$.

In this paper we shall consider three basic classes: the class \CQ of all Conjunctive Queries, the class \CRPQ of all CRPQs, and the class \CRPQfin of CRPQs using regular expressions with no Kleene-star (denoting finite languages). 
Observe that the latter corresponds to the subclass of \CRPQ without recursion.

\subsection{Standard, atom-injective, and query-injective semantics}

We now define the \emph{standard semantics} for CRPQs (\ie, the usual semantics from the database theory literature) and we introduce the two new sorts of injective semantics.

For simplicity of exposition, we first give the semantics for CRPQ's without $\epsilon$, and we then show how to expand the semantics to languages that include $\epsilon$.
Let $Q$ be a CRPQ of the form $Q(\bar z) =  x_1 \xrightarrow{L_1} y_1 \land \dotsb \land x_n \xrightarrow{L_n} y_n$ and assume that no language $L_i$ contains $\epsilon$ (the empty word).
Given a graph database $G$, the \defstyle{evaluation} of $Q$ over $G=(V,E)$ \defstyle{under standard semantics} ($st$-semantics for short), denoted by $Q(G)^{st}$, is the set of tuples 
$\bar v$ of nodes for which there is a mapping $\mu : \vars(Q) \to V$ such that $\mu(\bar z) = \bar v$ and for each $i$ there is a path $\pi_i$ from $\mu(x_i)$ to $\mu(y_i)$ in $G$ whose label is in $L_i$.
The evaluation under \defstyle{atom-injective semantics} ($\ani$-semantics for short), denoted by $Q(G)^{\ani}$, is defined similarly, but we further require that each $\pi_i$ is a simple path (if $x_i\neq y_i$) or a simple cycle (if $x_i = y_i$).
Finally, the evaluation under \defstyle{query-injective semantics} ($\qni$-semantics for short), denoted by $Q(G)^{\qni}$, is similar to the atom-injective semantics (\ie, each $\pi_i$ must be simple), 
but we additionally require that $\mu$ is injective and that for every $i \neq j$ there are no internal nodes shared by $\pi_i$ and $\pi_j$.

The semantics for a CRPQ $Q$ with $\epsilon$-words is defined as expected: the query $Q$ is equivalent to a \emph{union} of 
$\epsilon$-free CRPQs and hence its evaluation is the union of the evaluation of these $\epsilon$-free queries. 
More formally, 
for $\star \in \set{st, \ani, \qni}$,
the semantics under $\star$-semantics of $Q(\bar z) = x \xrightarrow L y \land Q'$, where $L$ contains $\epsilon$ and $Q'$ is a CRPQ, is the union of the set of tuples given by the $\star$-evaluation of $Q(\bar z) = x \xrightarrow {L \setminus \set \epsilon} y \land Q'$ and by the $\star$-evaluation of
$Q(\bar z[x / y]) = Q'[x / y]$, where $X[x / y]$ is the result of replacing every occurrence of variable $x$ with variable $y$ in $X$.

\begin{remark}
    \label{remark:sem-comparison}
    The three semantics form a hierarchy. In particular, for every CRPQ $Q$ and every graph database $G$, 
    we have $Q(G)^{\qni}\subseteq Q(G)^{\ani}\subseteq Q(G)^{st}$. The converse inclusions do not hold in general. 
    \end{remark}
    
            \begin{figure}
                \begin{center}
                    \includegraphics[width=.47\textwidth]{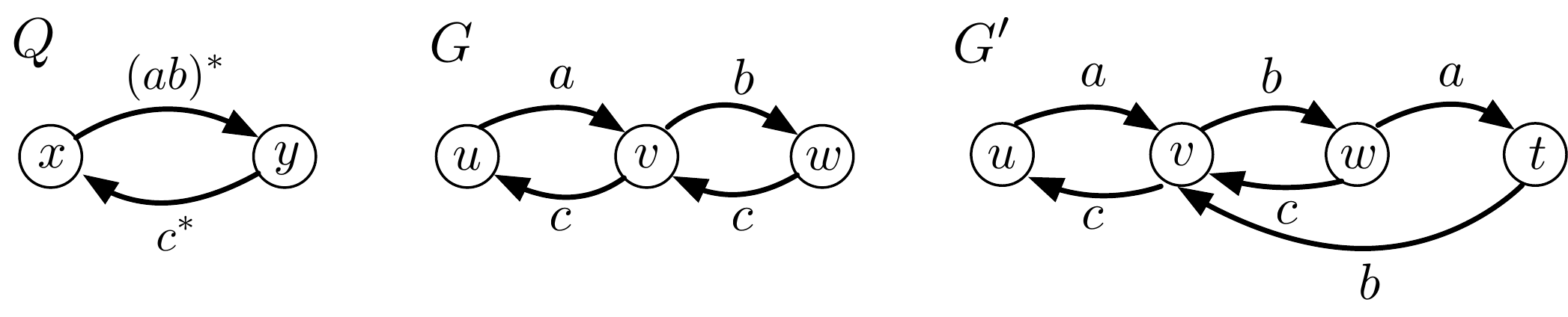}
                \end{center}
                \caption{The CRPQ $Q(x,y)$ and graph databases $G$ and $G'$ from Example~\ref{ex:1}.}
                \label{fig:example1}
            \end{figure}
    
    \begin{example}
    \label{ex:1}
    Consider the CRPQ $Q(x,y) = x \xrightarrow{(ab)^*} y \land y \xrightarrow{c^*} x$ and the graph database $G$ from Figure~\ref{fig:example1}. 
    Observe that $(u,w)\in Q(G)^{\ani}$ but $(u,w)\notin Q(G)^{\qni}$. On the other hand, it is easy to check that $Q(G)^{st}=Q(G)^{\ani}$. 
    The graph database $G'$ from Figure~\ref{fig:example1} provides a separation of the three semantics. Indeed, as before $Q(G')^{\ani}\not\subseteq Q(G')^{\qni}$, but additionally 
    we have $(u,v)\in Q(G')^{st}$ and $(u,v)\not\in Q(G')^{\ani}$. 
\end{example}

\subsection{Characterizing evaluation}
We state the semantics defined above in terms of restrictive notions of homomorphisms based on injectivity. This is based on the key notion of \emph{expansion} of a CRPQ (\aka~canonical database), which will become useful for the technical developments of the next sections.

For any atom $x \xrightarrow{L} y$ of a CRPQ $Q$ and $w \in L$, the \defstyle{$w$-expansion} of $x \xrightarrow{L} y$ is the Boolean CQ with equality atoms 
of the form 
    (i) $P = x \xrightarrow{a_1} z_1 \land z_1 \xrightarrow{a_2} z_2 \land \cdots \land z_{k-1} \xrightarrow{a_k} y$ if $w \neq \epsilon$, such that the $z_i$ are fresh new variables, or of the form
    (ii) $P = (y=z)$ if $w = \epsilon$.
We usually write $x \xrightarrow{w} y$ to denote such a $w$-expansion, with $w=a_1\cdots a_k$. An \defstyle{expansion} of $x \xrightarrow{L} y$ is a \defstyle{$w$-expansion} for some $w \in L$.
An \defstyle{expansion profile} of the CRPQ $Q$ is any function $\varphi$ mapping each atom of $Q$ to an expansion thereof. 
An \defstyle{expansion} of the CRPQ $Q(\bar{x})=A_1\land \dots\land A_m$ is a CQ $E(\bar{y})$ for which there is an expansion profile $\varphi$ of $Q$ such that $E=\widetilde{E}^{\collapse}$, where 
$\widetilde{E}$ is the CQ with equality atoms defined by $\widetilde{E}(\bar{x}) = \varphi(A_1)\land \cdots\land\varphi(A_m)$.  
We denote by $\Exp{}(Q)$ the set of all expansions of the CRPQ $Q$.
Intuitively, an expansion of $Q$ is obtained by expanding each atom of $Q$ and then collapsing equivalent variables. 
For example, one possible expansion of the query $Q(x,y) = x \xrightarrow{(ab)^*} y \land y \xrightarrow{c^*} x$ from Figure~\ref{fig:example1} is $E_1(x,x) = x \xrightarrow{a} z \land z \xrightarrow{b} x$ through the expansion profile mapping $x \xrightarrow{(ab)^*} y$ to $ab$ and $y \xrightarrow{c^*} x$ to $\epsilon$, and another expansion 
$E_2(x,y) = x \xrightarrow{a} z \land z \xrightarrow{b} y \land y \xrightarrow{c} x$ mapping the atoms to $ab$ and $c$ respectively.

\paragraph{Standard and query-injective semantics.}
The standard and $\qni$ semantics can be rephrased as follows. 

\begin{proposition}
\label{prop:semantics-st-qinj}
Let $Q$ be a CRPQ and $G$ be a graph database. 
Then $Q(G)^{st}$ \resp{$Q(G)^{\qni}$} is the set of tuples 
    $\bar{v}$ of nodes for which there is $E\in \Exp{}(Q)$ such that $E\to (G, \bar{v})$ \resp{$E\injto (G, \bar{v})$}. 
\end{proposition}

\paragraph{Atom-injective semantics.}
\label{sec:atom-level}
The atom-injective semantics corresponds to the less restrictive alternative that an arbitrary homomorphism can be allowed as long as it is injective when restricted to the expansions of each atom.
Let $E(\bar{y})$ be an expansion of a CRPQ $Q(\bar{x})$, $\varphi$ be an expansion profile producing $E$, and 
$\widetilde{E}(\bar{x}) = \varphi(A_1)\land \cdots\land\varphi(A_m)$ be the associated CQ with equality atoms (in particular $E=\widetilde{E}^{\collapse}$).
Let $\Phi:\vars(\widetilde{E}) \to \vars(E) $ be the canonical renaming.  
We say that two variables $x,y \in \vars(E)$ are \defstyle{$\varphi$-atom-related} if there is some atom expansion $\varphi(A_i)$ containing some $x',y'$
such that $\Phi(x') = x$ and $\Phi(y')=y$. 

We now define a notion of injective homomorphism tailored to CRPQ expansions. 
We say that $h$ is an \defstyle{atom-injective homomorphism} from an expansion $E$ of a CRPQ $Q$ to a graph database $G$ mapping free variables to $\bar v$ 
if $h:E\to (G,\bar v)$ and there is an expansion profile $\varphi$ producing $E$ such that $h(x)\neq h(y)$ for every pair of distinct $\varphi$-atom-related variables $x$ and $y$. We write $E \ainjto (G,\bar v)$ if such an $h$ exists. 
Atom-injective homomorphisms from $E$ to CQs are defined in the obvious way.

\begin{proposition}
\label{prop:semantics-ainj}
Let $Q$ be a CRPQ and $G$ be a graph database. 
Then $Q(G)^{\ani}$ is the set of tuples 
    $\bar{v}$ of nodes for which there is $E\in \Exp{}(Q)$ such that $E\ainjto (G, \bar{v})$. 
\end{proposition}
\section{The evaluation problem}
\label{sec:eval}
    
    The decision problem associated to evaluation is the \defstyle{evaluation problem} for a class $\Cc$ of CRPQs and a semantics $\star \in \set{st,\ani,\qni}$.
    \decpb{Evaluation problem for $\Cc$ under $\star$-semantics}
          {A graph database $G$, a query $Q(\bar x) \in \Cc$, and a tuple $\bar v$ of nodes.}
          {Is $\bar v \in Q(G)^\star$?}

    The evaluation problem is \np-complete under injective semantics, as it is the case 
    under standard semantics.

\begin{proposition}\label{prop:eval-np-c}
    The evaluation problem for \CRPQ and \CQ is \np-complete in combined complexity, for all semantics.
\end{proposition}
\begin{proof}
    The lower bound follows by an easy reduction from the injective-homomorphism testing problem, also known as the subgraph isomorphism problem, which is a well-known \np-complete problem. \cite{DBLP:conf/stoc/Cook71,DBLP:books/fm/GareyJ79}. 
    Indeed, a Boolean CQ $Q$ maps injectively to $G$ if{f} $Q(G)^{\qni} \neq \emptyset$ if{f} $Q^+(G^+)^{\ani}\neq\emptyset$, where $G^+$ \resp{$Q^+$} is the result of adding, for a fresh symbol $R$, an $R$-edge between every pair of vertices \resp{an $R$-atom between every pair of variables}.
    
    The upper bound is a consequence of the polynomial-sized witness property. That is, if $Q \in \CRPQ$, and $\bar{v} \in Q(G)^\star$, then there exists an expansion $E$ of $Q$ such that 
    $E\injto (G,\bar{v})$ if $\star=\qni$ and $E\ainjto (G,\bar{v})$ if $\star=\ani$. In either case, $E$ is linear in $G$ and $Q$. 
    One can then guess such an expansion and check the existence of the corresponding homomorphism.
\end{proof}

The data complexity for the alternative semantics (\ie, when the query is considered to be of constant size) is also \np-complete, since evaluation of RPQs under simple path semantics is \np-complete, even for very simple regular expressions~\cite{DBLP:journals/siamcomp/MendelzonW95}:
\begin{proposition}
    \label{prop:eval-np-c-datacomplexity}
    The evaluation problem for \CRPQ is \np-complete in data complexity, for atom-injective and query-injective semantics.
\end{proposition}
The RPQs which can be evaluated efficiently in data complexity have been characterized via a trichotomy result: they can be either \np-complete, \nl-complete, or in AC${}^0$ \cite[Theorem~2]{DBLP:journals/jcss/BaganBG20}. The generalization of this result to CRPQs under injective semantics seems highly non-trivial, and in particular it would necessitate a comprehensive understanding of the query equivalence problem, which is the focus of the next sections.

\section{The containment problem}
\label{sec:containment}

A CRPQ $Q_1$ is \emph{contained} in a CRPQ $Q_2$ under $\star$-semantics, denoted by $Q_1\subseteq_{\star} Q_2$, if $Q_1(G)^\star\subseteq Q_2(G)^\star$ for every graph database $G$. 
We define the \defstyle{containment problem}, 
which is parameterized by classes $\Cc_1$ and $\Cc_2$ of CRPQs as well as the semantics used (standard, query-injective, or atom-injective). 

\decpb{$\Cc_1$/$\Cc_2$ containment problem under $\star$-semantics}
	{CRPQs $Q_1\in \Cc_1$ and $Q_2\in \Cc_2$.}
	{Does $Q_1\subseteq_\star Q_2$ hold?}

Under standard semantics, all combinations among \CQ, \CRPQ and \CRPQfin have been studied and are decidable. In particular:
\begin{theorem}
\label{thm:cont-st} \cite{CGLV00,Florescu:CRPQ}
\label{thm:crpq-expspace-standard}
The \CRPQ/\CRPQ  containment problem under standard semantics is \expspace-complete. 
\end{theorem}

We will dedicate the rest of the paper to study the situation for injective semantics. We show that one injective semantics becomes undecidable while the other becomes better behaved computationally under standard complexity theoretic assumptions (\cf~the \CRPQ/\CRPQ column of Figure~\ref{fig:summary}).

\subsection{Characterizing containment}
\label{subsec:characterize-cont}

For the standard semantics, it is well-known that containment of CRPQs can be characterized in terms of expansions:

\begin{proposition}
\label{prop:cont-char-exp-st} \cite{CGLV00}
Let $Q_1$ and $Q_2$ be CRPQs. Then $Q_1\subseteq_{st} Q_2$ if{f} for every $E_1\in \Exp{}(Q_1)$ there is $E_2\in \Exp{}(Q_2)$ 
such that $E_2\to E_1$. 
 \end{proposition}
 
A similar characterization holds for query-injective semantics:
 
\begin{proposition}
\label{prop:cont-char-qinj} 
Let $Q_1$ and $Q_2$ be CRPQs. Then $Q_1\subseteq_{\qni} Q_2$ if{f} for every $E_1\in \Exp{}(Q_1)$ there is $E_2\in \Exp{}(Q_2)$ 
such that $E_2\injto E_1$. 
\end{proposition}

As it turns out, the previous characterization does not work for atom-injective semantics (replacing $\injto$ by $\ainjto$). 
In this case, the space of expansions of $Q_1$ is not enough and we need to check $Q_2$ against a \emph{larger} space of expansions of $Q_1$ we define below.

An \defstyle{atom-injective-expansion} ($\ani$-expansion for short)
of a CRPQ $Q(\bar{x})$ is a CQ $F(\bar{y})$ for which there is a CQ with equality atoms $\widetilde{F}(\bar{z})=E(\bar{z})\land J$ such that (i) $F=\widetilde{F}^{\collapse}$, 
(ii) $E(\bar{z})$ is an expansion of $Q$ produced by some expansion profile $\varphi$, 
and (iii) $J$ is a conjunction of equality atoms $x'=y'$, for variables $x', y'\in \vars(E)$, 
such that for every pair of distinct $\varphi$-atom-related variables $x,y$ in $E$, we have $x\neq_{\widetilde{F}} y$.
We denote by $\Exp{\ani}(Q)$ 
the set of all $\ani$-expansions of $Q$. 
The intuition is that these types of expansions are obtained from an ordinary expansion of $Q$ by identifying some pairs of variables which are not atom-related (the identifications are given by $J$).
We use this for the following useful result:

\begin{lemma}
\label{lemma:ainj-expansions}
Let $Q$ be a CRPQ, $E'$ be a CQ, $G$ be a graph database, and $\bar{v}$ be a tuple of nodes. The following are equivalent:
\begin{enumerate}
\item There is $E\in \Exp{}(Q)$ s.t.\ $E\ainjto (G, \bar{v})$ \resp{$E\ainjto E'$}.
\item There is $F\in \Exp{\ani}(Q)$ s.t.\ $F\injto (G, \bar{v})$ \resp{$F\injto E'$}.
\end{enumerate}
\end{lemma}

As a corollary of Lemma~\ref{lemma:ainj-expansions} we obtain an alternative definition of atom-injective semantics:

\begin{corollary}
\label{coro:a-inj-sem}
Let $Q$ be a CRPQ, $G$ be a graph database and $\bar{v}$ be a tuple of nodes. Then $\bar{v}\in Q(G)^{\ani}$ if and only if 
there is $F\in \Exp{\ani}(Q)$ such that $F\injto (G, \bar{v})$. 
\end{corollary}

We now give our characterization of atom-injective containment:

\begin{proposition}
\label{prop:cont-char-ainj} 
For every pair $Q_1$, $Q_2$ of CRPQs, the following are equivalent:
\begin{enumerate}
	\item \label{it:prop:cont-char-ainj:1} $Q_1\subseteq_{\ani} Q_2$.
	\item \label{it:prop:cont-char-ainj:2} For every $F_1\in \Exp{\ani}(Q_1)$ there is $E_2\in \Exp{}(Q_2)$ 
	such that $E_2\ainjto F_1$.
	\item For every $F_1\in \Exp{\ani}(Q_1)$ there is $F_2\in \Exp{\ani}(Q_2)$ 
	such that $F_2\injto F_1$.
\end{enumerate}
\end{proposition}

From the characterizations above, for every pair of CRPQs $Q_1,Q_2$, we have that $Q_1 \subseteq_{\qni} Q_2$ implies $Q_1 \subseteq_{st} Q_2$ and that $Q_1 \subseteq_{\ani} Q_2$ implies $Q_1 \subseteq_{st} Q_2$ (while the converse implications do not hold in general).
In contrast, and in spite of the hierarchy between the semantics, there is no such implication between query-injective and atom-injective containment, as the following example shows.

\begin{example}
Consider the Boolean CRPQs $Q_1 = x\xrightarrow{a} y\land y\xrightarrow{b} z$, $Q_2 = x\xrightarrow{ab} y$, $Q_1'=x\xrightarrow{a} y\land x\xrightarrow{b} y$ and $Q_2'=x\xrightarrow{a} y\land x'\xrightarrow{b} y'$. 
We have $Q_1'\subseteq_{\ani} Q_2'$ (and $Q_1'\subseteq_{st} Q_2'$) but $Q_1'\not\subseteq_{\qni} Q_2'$ as there cannot be an injective homomorphism from the unique expansion of $Q'_2$ to the unique expansion of $Q'_1$. 
On the other hand, we have $Q_1\subseteq_{\qni} Q_2$ (and $Q_1\subseteq_{st} Q_2$) but $Q_1\not\subseteq_{\ani} Q_2$. Indeed, we can take the $\ani$-expansion $F$ of $Q_1$ obtained from $x\xrightarrow{a} y\land y\xrightarrow{b} z$ by identifying $x$ and $z$. Then, there cannot be an atom-injective homomorphism from the unique expansion of $Q_2$ to $F$. 
\end{example}

In view of the characterizations above, in the sequel we will sometimes write \defstyle{st-expansions} or \defstyle{\qni-expansions} to denote a (normal) expansion. 
For $\star\in\{st, \qni, \ani\}$, we say that $E_1(\bar y)$ is a \defstyle{counter-example} for $\star$-semantics if $E_1$ is a $\star$-expansion of $Q_1$ such that  $\bar y \not\in Q_2(E_1)^\star$ (recall that any CQ, in particular $E_1$, can be seen as a graph database). Note that the latter condition $\bar y \not\in Q_2(E_1)^\star$ is equivalent to the non-existence of a (normal) expansion $E_2$ of $Q_2$ such that either (a) $E_2\to E_1$ if $\star=st$; (b) $E_2\injto E_1$ if $\star=\qni$; or (c) $E_2\ainjto E_1$ if $\star=\ani$. 
Hence, $Q_1 \not\subseteq_{\star} Q_2$ if, and only if, there exists a counter-example for $\star$-semantics.

\section{Containment for unrestricted CRPQs}
\label{sec:unrestricted}
In this section we study the \CRPQ/\CRPQ containment problem under query-injective and atom-injective semantics. We show that the former is in \pspace while the latter is undecidable. Both proofs are non-trivial and provide novel insights on how the semantics can be exploited for static analysis problems: In the first case by reducing the space needed from exponential to polynomial, and in the second case by enforcing counter-examples to witness solutions of the PCP problem, through an intricate encoding.

The \CRPQ/\CRPQ containment problem for standard semantics is in \expspace \cite{Florescu:CRPQ} and all proofs \cite{Florescu:CRPQ,CGLV00,Figueira20} of which we are aware reduce the problem to test containment or universality on exponentially-sized NFA's, encoding the set of expansions or the set of (non-)counter-examples. The \pspace bound for query-injective semantics is, however, quite different in nature, and uses some exponential number of polynomial-sized certificates to ensure that the containment holds.

\begin{theorem}\label{thm:crpq-crpq-nodeinj-pspace}
        The \textup{\CRPQ/\CRPQ} containment problem under query-injective semantics is in \pspace.
\end{theorem}

\begin{proof}[Proof sketch\ifarxiv (full proof in Appendix~\ref{sec:pr:lem:crpq-crpq-nodeinj-pspace})\fi]
	Let $Q_1, Q_2$ be CRPQs; we want to test $Q_1 \subseteq_{\qni} Q_2$.
	We only give a high-level description of the proof due to space constraints.
	We will work with polynomial-sized `abstractions' of expansions of $Q_1$. These abstractions contain, for each atom $A$ of $Q_1$, all the information on how the languages of $Q_2$ can be mapped into it. For example, it includes the information ``\textsl{there is a partial run from state $q$ to state $q'$ of the NFA $\+A_L$ of language $L$ from $Q_2$ reading the expansion word of $A$}'', or ``\textsl{there is a partial run from the initial state of $\+A_L$ to $q$ reading some suffix of the expansion word of $A$}''. Such an abstraction contains all the necessary information needed to retain from an expansion to check whether it is a counter-example. Indeed, any expansion having the same abstraction as a counter-example will be a counter-example. 

	In order to test whether an abstraction $\alpha$ abstracts a counter-example, we need to consider all possible ways of injectively mapping an expansion of $Q_2$ to an expansion of $Q_1$. We call this a \emph{morphism type}, which contains the information of where each atom expansion of $Q_2$ is mapped to. For example, we can have the information that the path to which the expansion of atom $A$ of $Q_2$ is mapped starts at some internal node of the expansion of atom $A_1= x \xrightarrow{L_1} y$ of $Q_1$ then arrives to variable $y$ with state $q$ and continues reading the full expansion of atom $A_2 = y \xrightarrow{L_2} z$ arriving to variable $z$ with state $q'$, and it ends its journey by reading a prefix of the expansion of atom $A_3 = z \xrightarrow{L_3} t$ arriving to a final state at some internal node. For each morphism type, we can check if it is \emph{compatible} with an abstraction by checking, for example, that $\alpha$ indeed contains the information of having a partial run from $q'$ to a final state reading a prefix of the expansion of $A_3$. 

	More concretely,
	\newcommand{\lGraph}{\mathbf{G}}
	consider the directed graph $\lGraph$ consisting of replacing each atom $A = x \xrightarrow{L} y$ of $Q_1$ with a path $\pi^\lGraph_A$ of length 3 (\ie, adding two new internal vertices). 
	A \defstyle{morphism type} from $Q_2$ to $Q_1$ is a pair $(H,h)$ such that $h: H \injto \lGraph$ and $H$ is a graph resulting from replacing each atom $A =  x \xrightarrow{L} y$ of $Q_2$ with a path $\pi^H_A$ from $x$ to $y$. 
	Figure~\ref{fig:lem:crpq-crpq-nodeinj-pspace-sketch-body} has an example of a morphism type.
	\begin{figure}
		\includegraphics[width=.47\textwidth]{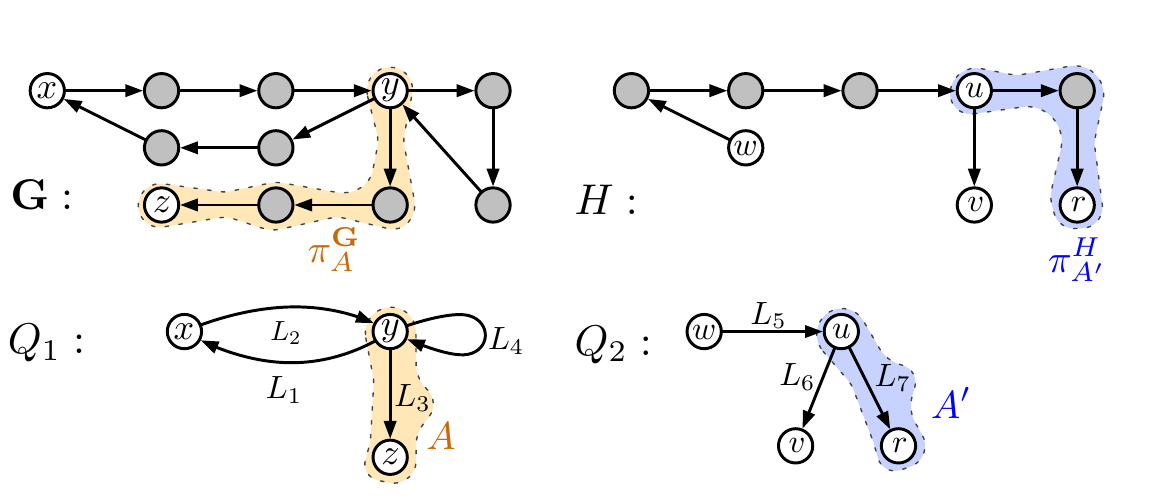}
		\caption{\textnormal{Example of definition of $\lGraph$ and morphism type $(H,h)$ from $Q_1,Q_2$. In this case, the injective morphism $h$ from $H$ to $\lGraph$ maps each node of $H$ to the node in the same position on $\lGraph$ (\eg, the lower-right node $r$ of $H$ maps to the lower-right node of $\lGraph$).}}
		\label{fig:lem:crpq-crpq-nodeinj-pspace-sketch-body}
	\end{figure}
	By injectivity, the size of $H$ in any morphism type is linearly bounded on $Q_1$.

	A morphism type $(H,h)$ is \defstyle{compatible} with an abstraction $\alpha$ if there is a mapping $\lambda$ from the internal nodes of paths $\pi^H_A$ to states of $A$, such that for every atom $A=x_1 \xrightarrow{L} x_2$ of $Q_1$, all the expected properties must hold. 
	For example, if there is an atom $A'$ of $Q_2$ and an infix $\pi$ of the path $\pi_{A'}^H$ with $h(\pi) = \pi^\lGraph_A$, then the abstraction $\alpha$ ensures that there is a run of $\+A_{A'}$ from $\lambda(src(\pi))$ to $\lambda(tgt(\pi))$, where $\+A_{A'}$ is the NFA of the language of $A'$, and $src(\pi)$ and $tgt(\pi)$ are the first and last nodes of $\pi$, respectively. 
	Or, as another example, one must also check that if there is an atom $A'$ of $Q_2$ and a suffix $\pi$ of $\pi^H_{A'}$ with $h(\pi)$ being a prefix of $\pi^\lGraph_A$, then $\alpha$ ensures that there is a run from $\lambda(src(\pi))$ to some final state in $\+A_{A'}$ on the prefix of the expansion of $A$. There are actually many other possible cases (17 in total), but each of these cases can be easily checked with the information compiled in an abstraction.

	The key property of compatibility is that it captures whether an abstraction contains a $Q_1$-expansion  which is a counter-example:
	\begin{claim}
		The following are equivalent:
		\begin{enumerate}
			\item There is a morphism type compatible with an abstraction $\alpha$;
			\item for every expansion $E_1 \in \Exp{}(Q_1)$ with abstraction $\alpha$ there exists some expansion $E_2 \in \Exp{}(Q_2)$ such that $E_2 \injto E_1$;
			\item there is an expansion $E_1 \in \Exp{}(Q_1)$ with abstraction $\alpha$ and an expansion $E_2 \in \Exp{}(Q_2)$ such that $E_2 \injto E_1$.
		\end{enumerate}
	\end{claim}
  Finally, the \pspace algorithm guesses a mapping $\alpha$ from the atoms of $Q_1$ to subsets of $P$, checks that $\alpha$ is an abstraction of $Q_1$, and checks that there is no morphism type $(H,h)$ which is compatible with $\alpha$. Due to the Claim above, if the algorithm succeeds, then any expansion $E_1$ of $Q_1$ is a counter-example, and thus $Q_1 \not\subseteq_{\qni} Q_2$; otherwise, for every expansion $E_1$ with abstraction $\alpha$ there is a compatible morphism type, which means that $E_1$ is not a counter-example and hence $Q_1 \subseteq_{\qni} Q_2$. 
\end{proof}

On the other hand, the \CRPQ/\CRPQ containment problem for atom-injective semantics becomes undecidable. 
Remarkably, the bound holds even when the right-hand side query has no infinite languages, and both queries are of very simple shape (\cf~Figure~\ref{fig:query-q1-q2-simple}).
\begin{theorem}
    \label{theo:undec-atom-ni}
    The \textup{\CRPQ/\CRPQ} and \textup{\CRPQ/\CRPQfin} containment problems under atom-injective semantics are undecidable.
\end{theorem}

\begin{proof}[Proof sketch\ifarxiv (full proof in Appendix~\ref{sec:app-full-undec})\fi]
    We reduce from the \emph{Post Correspondence Problem} (PCP), a well-known undecidable problem. 
    An instance of the PCP is a sequence of pairs $(u_1,v_1),\dots, (u_\ell, v_\ell)$, where $u_i$ and $v_i$ are non-empty words over an alphabet $\Sigma$. 
    The goal is to decide whether there is a \emph{solution}, that is, a sequence $i_1,\dots,i_k$ of indices from $\{1,\dots, \ell\}$, with $k\geq 1$, such that  
    the words $u_{i_1}\cdots u_{i_k}$ and $v_{i_1}\cdots v_{i_k}$ coincide. 
    
    We provide a high-level description of the reduction. 
    The idea is to construct Boolean CRPQs $Q_1$ and $Q_2$ such that the PCP instance 
    $(u_1,v_1),\dots, (u_\ell, v_\ell)$ has a solution if and only if $Q_1\not\subseteq_{\ani}Q_2$. 
    In particular, the PCP instance has a solution if and only if there exists a counterexample for $\ani$-semantics, \ie, an $\ani$-expansion $F$ of $Q_1$ such that there is no expansion $E\in \Exp{}(Q_2)$ with $E\ainjto F$.  
    The general structure of $Q_1$ is shown in Figure~\ref{fig:query-q1-q2-simple}. We have a ``middle'' variable $x$, 
    two ``incoming'' atoms and two ``outgoing'' atoms:
    $$  y_1\xrightarrow{L_{I}} x \land y_2\xrightarrow{\widehat{L}_{a}} x \land x\xrightarrow{\widehat{L}_{I}} z_1 \land x \xrightarrow{L_{a}} z_2 $$
    Words in the languages $L_{I}$ and $\widehat{L}_{I}$ encode sequences of indices from $\{1,\dots,\ell\}$, 
    using special symbols from $\Ic:=\{I_1,\dots, I_\ell\}$ and  $\widehat{\Ic}:=\{\widehat{I}_1,\dots, \widehat{I}_\ell\}$, respectively.  On the other hand, 
    words from $L_{a}$ and $\widehat{L}_{a}$ encode sequences of words from $\{u_1,\dots, u_\ell\}$ and $\{v_1,\dots, v_\ell\}$,  using symbols from the PCP alphabet $\Sigma$ and $\widehat{\Sigma}:=\{\widehat{a}: a\in\Sigma\}$, respectively. In the four languages, we have some extra symbols to make the reduction work. 
We stress that the finite alphabet used for the CRPQ $Q_1$ (and also for $Q_2$) depends on the PCP instance.
    
        \begin{figure}
            \begin{center}
                \includegraphics[width=.37\textwidth]{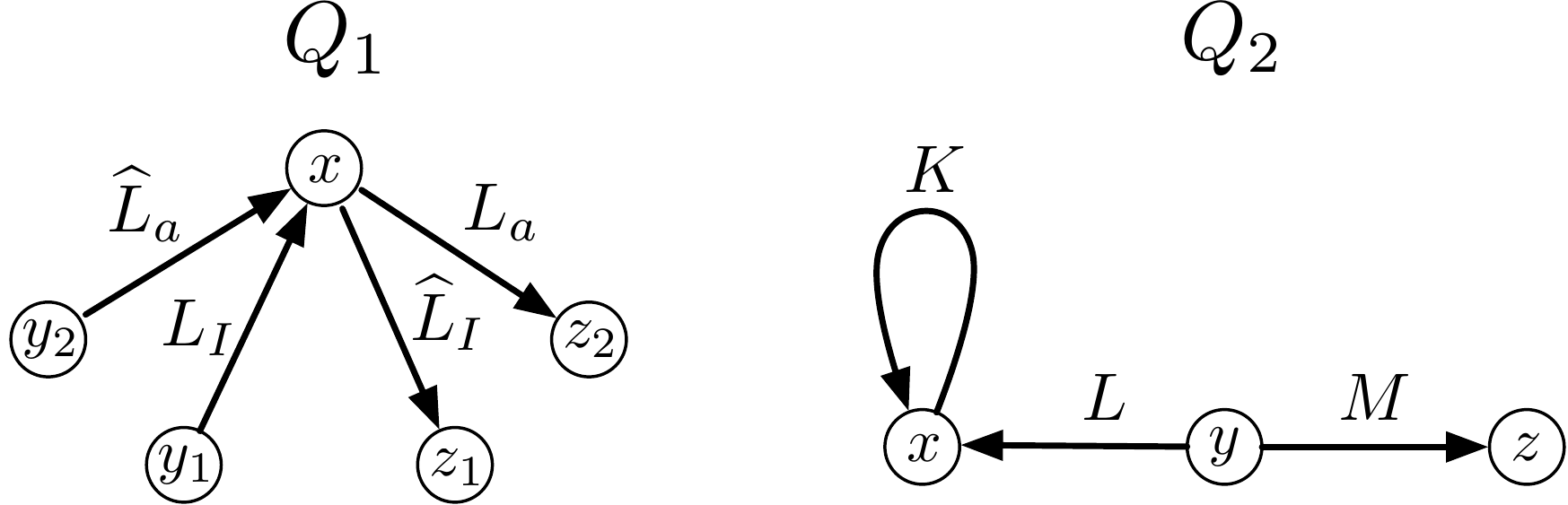}
            \end{center}
            \caption{\textnormal{The general structure of Boolean CRPQs $Q_1$ and $Q_2$ from the reduction.}}
            \label{fig:query-q1-q2-simple}
        \end{figure}

        We are interested in a particular type of $\ani$-expansions of $Q_1$ that we call \emph{well-formed}. 
        The idea is that well-formed $\ani$-expansions correspond to solutions of the PCP instance. In particular, 
        if there is a well-formed $\ani$-expansion of $Q_1$ then there is a solution to the PCP instance and vice versa. 
       We then construct $Q_2$ such that an $\ani$-expansion of $Q_1$ is well-formed if and only if 
       it is a counterexample for $Q_1\subseteq_{\ani}Q_2$. 
       
       Let $F$ be an $\ani$-expansion of $Q_1$ such that  $F=\widetilde{F}^{\collapse}$ for $\widetilde{F} = E\land J$ (here $E\in \Exp{}(Q_1)$ and $J$ are the equality atoms). The  $\ani$-expansion $F$ 
       is \emph{well-formed} if it satisfies four structural conditions we call $I$-$\widehat{I}$-, 
       $I$-$a$-, $\widehat{a}$-$\widehat{I}$-, and $\widehat{a}$-${a}$-condition. 
       Intuitively, the $I$-$\widehat{I}$-condition requires that the words $w_I\in L_{I}$ and $\widehat{I}\in \widehat{L}_I$ chosen in the expansion $E$ encode the \emph{same} sequence of indices from $\{1,\dots, \ell\}$. 
       On the other hand, the $I$-$a$-condition ensures that the word $w_a \in L_a$, chosen in the expansion $E$,
       encodes a sequence of words from $\{u_1,\dots,u_\ell\}$ \emph{according} to the sequence encoded in $w_I\in L_{I}$.  Similarly, the $\widehat{a}$-$\widehat{I}$-condition requires that the chosen word $\widehat{w}_a \in \widehat{L}_a$
       encodes a sequence of words from $\{v_1,\dots,v_\ell\}$ according to $\widehat{w}_I\in \widehat{L}_{I}$. 
       Finally, the $\widehat{a}$-${a}$-condition requires that the chosen words $w_a \in L_a$ and 
       $\widehat{w}_a \in \widehat{L}_a$ ``coincide'' after removing the $\,\widehat{\cdot}\,$ superscripts from $\widehat{w}_a$ 
       and focusing on the symbols from $\Sigma$. In other words, the $\widehat{a}$-${a}$-condition ensures that 
       the sequence of indices chosen by $w_I$ (and $\widehat{w}_I$) is actually a solution to the PCP. 
       In the four cases, we additionally need to require some conditions on the equality atoms $J$ to make the reduction work (see Figure~\ref{fig:well-formed} for an example).
       
             \begin{figure}
            \begin{center}
                \includegraphics[width=.37\textwidth]{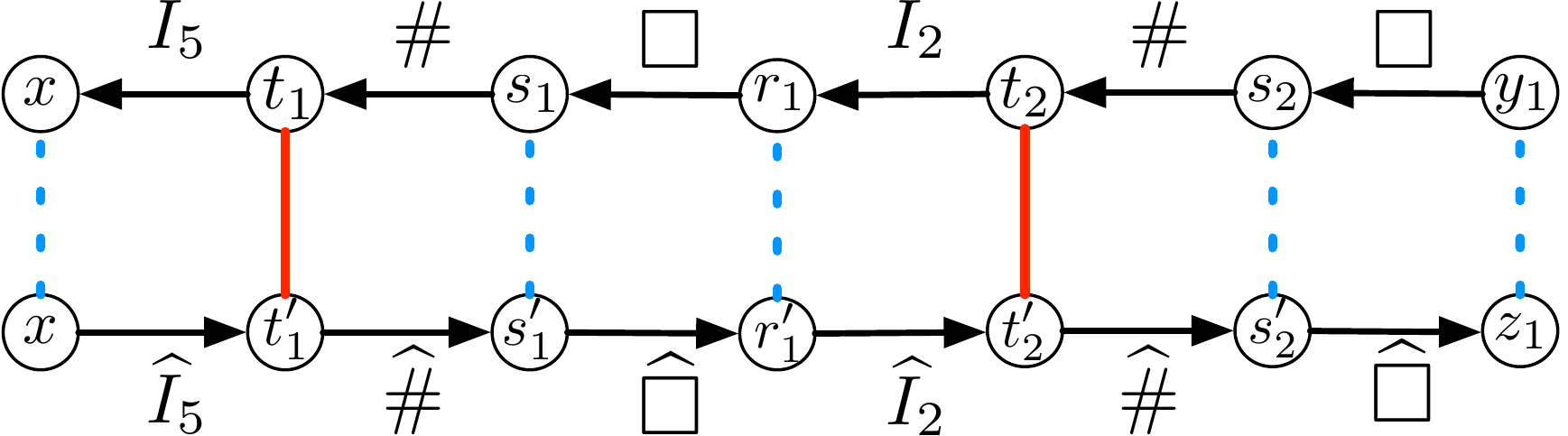}
            \end{center}
            \caption{\textnormal{Example of the $I$-$\widehat{I}$-condition of well-formed expansions of $Q_1$. 
            We show the expansions $y_1\xrightarrow{w_{I}} x$ and 
             $x\xrightarrow{\widehat{w}_{I}} z_1$ of the atoms $y_1\xrightarrow{L_{I}} x$ and 
             $x\xrightarrow{\widehat{L}_{I}} z_1$, respectively. The words $w_{I} = \square\, \#\, I_2 \, \square\, \#\, I_5$ and 
             $\widehat{w}_{I} =  \widehat{I}_5\, \widehat{\#}\, \widehat{\square} \, \widehat{I}_2 \, \widehat{\#}\, \widehat{\square}$ encode the sequence of indices $5,2$. 
             Dotted blue lines indicate pairs of equal variables while red lines indicate distinct variables. We have some extra symbols $\#, \widehat{\#}, \square, \widehat{\square}$.  
             }}
            \label{fig:well-formed}
        \end{figure}
       
The key property of well-formedness is that it can be characterized in terms of the non-existence of a finite number of simple cycles and simple paths having certain labels. Let us illustrate this for the case of the 
$I$-$\widehat{I}$-condition and the expansion $F$ of Figure~\ref{fig:well-formed}. 
The forbidden labels for simple cycles are given by the finite language $K_{I \widehat{I}} = \Ic\, \widehat{\Ic}$. 
In the case of simple paths, these are given by $M_{I \widehat{I}} = \sum_{i\neq j} I_i \widehat{I}_j \, + \widehat{\Ic} \, \# \, + \widehat{\#}\, \Ic \,+ \#\, \Ic \, \widehat{\Ic} \, \widehat{\#} \,+ \square\, \widehat{\square}$.
We have that an $\ani$-expansion $F$ of $Q_1$ satisfies the $I$-$\widehat{I}$-condition if and only if 
$F$ does not contain simple cycles with labels in $K_{I \widehat{I}}$ nor simple paths with labels
in $M_{I \widehat{I}}$.  

To see the backward direction, note that $t_1$ and $t_1'$ cannot be identified, otherwise 
we have a simple cycle from $t_1$ to itself
with label in $\Ic\, \widehat{\Ic}\subseteq K_{I \widehat{I}}$. 
Now, the symbols $I_{5}$ and $\widehat{I}_{5}$ need to correspond to the same index from $\{1,\dots, \ell\}$; 
otherwise we find a simple path from $t_1$ to $t_1'$ with label in $\sum_{i\neq j} I_i \widehat{I}_j \subseteq M_{I \widehat{I}}$. To see the identification between $s_1$ and $s_1'$, note first that $t_1$ and $s_1'$ cannot be identified, 
as this would imply a simple path from $t_1'$ to $x$ with label in $\widehat{\#}\,  \Ic \subseteq M_{I \widehat{I}}$. 
Analogously, we have that $t_1'$ and $s_1$ cannot be identified. This implies that $s_1$ and $s_1'$ are actually 
identified, otherwise we have a simple path from $s_1$ to $s_1'$ with label in $\#\, \Ic \, \widehat{\Ic} \, \widehat{\#} \subseteq M_{I \widehat{I}}$. Finally, $r_1$ and $r_1'$ are identified, 
otherwise there would be a simple path from $r_1$ to $r_1'$ with label in $\square\, \widehat{\square} \subseteq M_{I \widehat{I}}$. Note that we can repeat this argument from ``left-to-right'' starting from $r_i=r_i'$ instead of $x$, 
and obtain the $I$-$\widehat{I}$-condition. In order to ensure that the words $w_I$ and $\widetilde{w}_I$ have the
same length, we need to slightly modify the construction of $Q_1$, $K_{I \widehat{I}}$ and $M_{I \widehat{I}}$\ifarxiv (see Appendix~\ref{sec:app-full-undec} for details)\fi. The forward direction follows directly from the definition of 
the $I$-$\widehat{I}$-condition.

Since all the four conditions can be characterized via forbidden finite sets of simple cycles and paths, it is possible 
to write two CRPQs from $\CRPQfin$  of the form $Q_2^{\circlearrowright} = x \xrightarrow{K^{\circlearrowright}} x$ 
and $Q_2^{\rightarrow} = y \xrightarrow{M^{\rightarrow}} z$ such that 
for every $\ani$-expansion $F$ of $Q_1$, $F$ is well-formed if and only if $Q_2^{\circlearrowright} \lor Q_2^{\rightarrow}(F)^{\ani} = \emptyset$, where $Q_2^{\circlearrowright} \lor Q_2^{\rightarrow}$ is the \emph{union} of both CRPQs. 
In particular, there is a solution to the PCP instance if and only if $Q_1\not\subseteq_{\ani}Q_2^{\circlearrowright} \lor Q_2^{\rightarrow}$. 
We finally show how to simulate the union $Q_2^{\circlearrowright} \lor Q_2^{\rightarrow}$ with a single query $Q_2\in \CRPQfin$ as in Figure~\ref{fig:query-q1-q2-simple}.
\end{proof}

\section{Containment for CRPQ subclasses}
\label{sec:subclasses}

With the two previous results in place for the containment of unconstrained CRPQs, we now explore 
the $\Cc_1$/$\Cc_2$ containment problem under all the possible semantics, where $\Cc_1$ and/or $\Cc_2$ belong to simpler classes of queries, namely either Conjunctive Queries or CPRQs with no Kleene star (and hence with finite languages).

In many cases one can apply or adapt previously established techniques or reductions. There are, however, two noticeable exceptions: the lower bounds for \CRPQfin/\CQ under query-injective semantics and for \CQ/\CRPQfin under atom-injective semantics. We highlight only these two results. The remaining proofs can be found in \ifarxiv   Appendix~\ref{app:other-results}. \else the {\fullversion} of the paper. \fi In particular, as mentioned in Section~\ref{sec:containment}, almost all the results for the standard semantics follow from previous work (in particular~\cite{CGLV00,Florescu:CRPQ,FigueiraGKMNT20}).

\begin{theorem}\label{thm:crpq-cq-pitwo-hard-qni}
    The \textup{\CRPQfin/\CQ} containment problem under query-injective semantics is \pitwo-hard.
\end{theorem}
\begin{proof}
	We show that even when the languages of the left-hand query are unions of alphabet symbols, \pitwo hardness still holds.
	We show a reduction from the following problem on graphs, which is known to be \sigmatwo-complete \cite[Theorem~5]{DBLP:conf/mfcs/Rutenburg86}, to non-containment. For a graph $G$ let $V(G)$ denote its sets of vertices and let $G|_{V'}$ denote the subgraph induced by $V' \subseteq V(G)$.
	\newcommand{\GCPb}{\text{GCP${}_{\text{2}}$}\xspace}
	\decpb{Generalized Two-Coloring Problem (\GCPb)}
		{An undirected 	graph $G$, a number $n \in \N$ (in unary).}
		{Is there a partition $V_1 \dcup V_2 = V(G)$ s.t.\ neither $G|_{V_1}$ nor $G|_{V_2}$  
		contains an $n$-vertex clique as subgraph?}
	We will produce two Boolean queries $Q_1, Q_2$ over the alphabet $\A = \set{E,1,2,\#}$ such that:
		(1) $Q_1 \not\subseteq_\qni Q_2$ if{f} the \GCPb instance is positive; and
		(2) $Q_2$ is a CQ, and every language of $Q_1$ is a set of one-letter words.
	Consider the input graph $G$, and the associated \CQ $Q_G$ on over the alphabet $\set{E}$, where for each edge $\set{u,v}$ in $G$ we have atoms $u \xrightarrow{E} v \land v \xrightarrow{E} u$ in $Q_G$. Similarly, let $K_n$ be the \CQ associated to the $n$-vertex clique.
	\newcommand{\ext}[1]{#1\textit{-ext}}%
	For a \CQ $Q$ and $i \in \set{1,2}$, let us define $\ext{i}(Q)$ \resp{$\ext{(1+2)}(Q)$; $\ext{(12)}(Q)$} as the extension of $Q$ by adding one atom $x \xrightarrow{i} x$ \resp{one atom $x \xrightarrow{1 + 2} x$; two atoms $x \xrightarrow{1} x \land x \xrightarrow{2} x$} for every variable $x \in \vars(Q)$.
	We define $Q_1,Q_2$ as in Figure~\ref{fig:lem:crpq-cq-pitwo-hard-qni}.
	\begin{figure}
		\begin{center}
			\includegraphics[width=.47\textwidth]{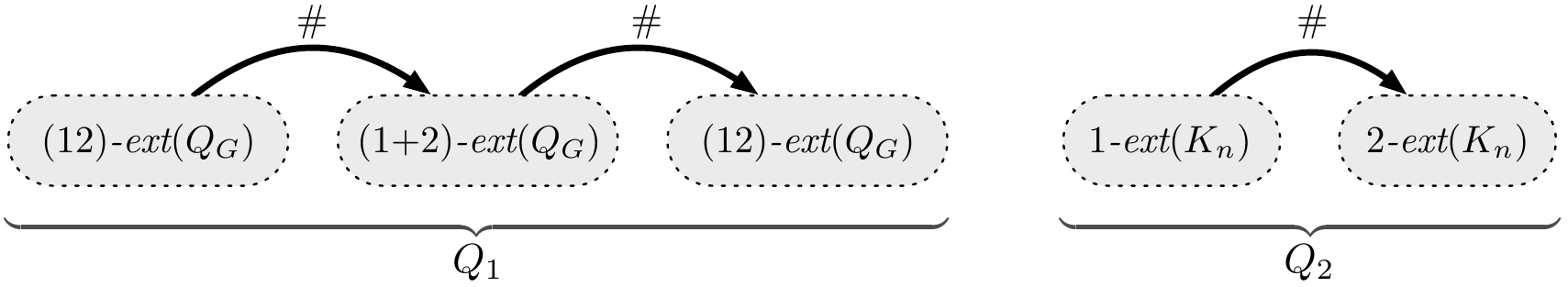}
		\end{center}
		\caption{\textnormal{Definition of $Q_1$ and $Q_2$ in terms of $G$ and $n$ for the reduction of Theorem~\ref{thm:crpq-cq-pitwo-hard-qni}. The $\#$-labeled thick arrows denote that there is an atom $x \xrightarrow{\#} y$ for each variable $x$ from the source query to each variable $y$ of the target query.}}
		\label{fig:lem:crpq-cq-pitwo-hard-qni}
	\end{figure}
	On the one hand, if $Q_1 \not\subseteq_{\qni} Q_2$ there must be some expansion $E$ of $Q_1$ which is a counter-example. From $E$ we can derive the partitioning $V_1 \dcup V_2$ of $V(G)$ where $V_i$ is the set of vertices labeled with an $i$-loop in the middle gadget of $Q_1$. Now observe that, for every $i$, $K_n$ is not injectively mapped to $G|_{V_i}$, as otherwise we would have 
	that  $i\textit{-ext}(K_n) \injto E$, implying $Q_2 \injto E$ and contradicting that $E$ is a counter-example. This means that the \GCPb instance is positive. 
	On the other hand, if there is a partitioning $V_1 \dcup V_2$ of $V(G)$ avoiding the $n$-clique as a subgraph, then the corresponding expansion $E$ of $Q_1$ (by choosing to have an $i$-loop for each node $x \in V_i$) is such that $Q_2$ cannot be injectively mapped to $E$; in other words showing that $E$ is a counter-example and thus $Q_1 \not\subseteq_{\qni} Q_2$.
\end{proof}

\begin{theorem}\label{thm:crpqfin-crpqfin-lower}
	The \textup{\CQ/\CRPQfin} containment problem under atom-injective semantics is \pitwo-hard.
\end{theorem}
\begin{proof}[Proof sketch\ifarxiv (full proof in Appendix~\ref{sec:app-cq-crpq-sf-atom-full})\fi]
	We show that even when all languages on the right-hand side are of the form $\set{ w }$ with $|w| \leq 2$ we have \pitwo-hardness for containment. For this, we show how to adapt the proof of \pitwo-hardness of \cite[Theorem~4.3]{FigueiraGKMNT20}, which shows \pitwo-hardness for \CRPQfin/\CQ containment for the standard semantics.
	We reduce from $\forall \exists$-QBF (\ie, $\Pi_2$-Quantified Boolean Formulas).
	Let
	$
	\Phi
	=
	 \forall x_1, \ldots, x_n\; \exists y_1, \ldots, y_\ell\;
	\varphi(x_1, \ldots, x_n,y_1, \ldots, y_\ell)
	$
	be an instance of
	$\forall \exists$-QBF such that $\varphi$ is quantifier-free and in
	3-CNF. We construct boolean queries $Q_1$ and $Q_2$ such that
	$Q_1 \subseteq_\ani Q_2$ if, and only if, $\Phi$ is satisfiable.

	The query $Q_1$ is defined in Figure~\ref{fig:pi-p-3-hard}, over the alphabet of labels $\{a,x_1,\dots,\allowbreak x_n, \allowbreak y_1,\dots, \allowbreak y_\ell, \allowbreak t,f,r\}$. We now explain how we define $Q_2$, over the same alphabet. Every clause of $\Phi$ is represented
  by a subquery in $Q_2$, as depicted in Figure~\ref{fig:pi-p-3-hard}. All nodes
  with identical label ($y_{1,t}$ and $y_{1,f}$ in gadgets $D,
  E$) in the figures are the same
  node. Note that for every clause and every existentially quantified literal $y_i$ therein we have one node named $y_{i,tf}$ in $Q_2$. The $E$-gadget is designed
  such that every represented literal can be homomorphically embedded, while 
  exactly one literal has to be embedded in the $D$-gadget.

  The intuitive idea is that the valuation of the $x$-variables is given by the
  \ani-expansion $E_1$ of $Q_1$, whether the two nodes under $x$ incident to $t$ are equal or not: if they are equal this corresponds to a \emph{false} valuation, otherwise a \emph{true} valuation. 
  On the other hand, the valuation of the $y$-variables is
  given by the homomorphism of an expansion of $Q_2$ into $E_1$ (i.e., whether the corresponding node
  is mapped to the node $y_{\mbox{\textvisiblespace},t}$ or
  $y_{\mbox{\textvisiblespace},f}$). The homomorphism of $y$-variables across several
  clauses has to be consistent, as all clauses share the same nodes
  $y_{\mbox{\textvisiblespace},tf}$, which uniquely get mapped either into 
  $y_{\mbox{\textvisiblespace},t}$ or $y_{\mbox{\textvisiblespace},f}$. Hence, when the
  formula $\Phi$ is satisfiable, for any assignment to the variables $\{x_i\}$
  (given by the choice of $t$/$f$ edges in $D$), there is a mapping from
  $y_{\mbox{\textvisiblespace},tf}$ to one of $y_{\mbox{\textvisiblespace},f}$ or
  $y_{\mbox{\textvisiblespace},t}$. This gives $Q_1 \subseteq_\ani Q_2$. Conversely, if
  an expansion of
  $Q_2$ can be mapped into $K$, then, for a choice of $t$/$f$ edges in $D$, we
  have an embedding of each clause gadget of $Q_2$ in $K$. In particular, we can
  always map a literal in each clause of $Q_2$ to $D$, ensuring that $\varphi$ is
  satisfied. As this is true for any expansion $K$ obtained by any $t$/$f$ assignment to $\{x_i\}$, we obtain that $\Phi$ is satisfiable.
  \newcommand{\tileN}{\lozenge}
  \newcommand{\tileY}{\blacklozenge}
  \newcommand{\tilen}{\mbox{$\vartriangle$}}
  \newcommand{\tiley}{\blacktriangle}
  	\begin{figure}
		\includegraphics[width=.47\textwidth]{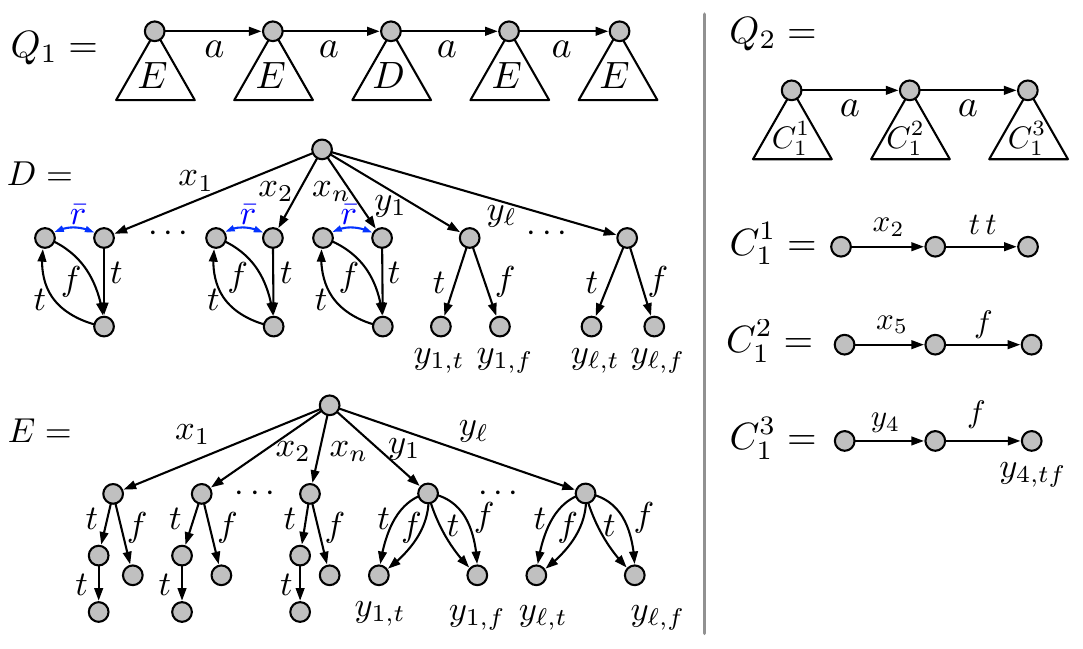}
		\caption{Left: \textnormal{Query $Q_1$ used in Theorem~\ref{thm:crpqfin-crpqfin-lower} and the gadgets $D$, and $E$ used in its definition. The $\bar r$ blue edges depict the edges of the complement of $r$ (\ie, the edges which are \emph{not} in relation $r$) for clarity.} Right: \textnormal{Example of $Q_2$ for 
		$\phi = (x_2 \lor \lnot x_5 \lor \lnot y_4)$.}}
		\label{fig:pi-p-3-hard}
	\end{figure}
\end{proof}

\paragraph{Other results}
The remaining results are summarized on Figure~\ref{fig:summary} 
and the following theorem, 
whose proofs can be found in \ifarxiv Appendix~\ref{app:other-results}\else the \fullversion\fi.
\begin{theorem}
\label{thm:other}
    \hfill
    \begin{enumerate}
        \item The \textup{\CQ/\CRPQ} and \textup{\CQ/\CQ} containment problems are \np-complete under query-injective semantics. \ifarxiv(Proposition~\ref{prop:cq-crpq-cq-inj} in Appendix~\ref{app:other-results}.)\fi
        \item The \textup{\CQ/\CQ} containment problem under atom-injective semantics is \np-complete. \ifarxiv(Corollary~\ref{cor:cq-cq-sp} in Appendix~\ref{app:other-results}.)\fi
        \item The \textup{\CRPQ/\CQ} and \textup{\CRPQfin/\CQ} containment problems are \pitwo-hard, under standard and atom-injective semantics. \ifarxiv(Proposition~\ref{prop:crpq-cq-ani-hard} in Appendix~\ref{app:other-results}.)\fi
        \item The \textup{\CRPQ/\CQ} and \textup{\CRPQfin/\CQ} containment problems are in \pitwo, under all semantics. \ifarxiv(Proposition~\ref{prop:crpq-cq-all-upper} in Appendix~\ref{app:other-results}.)\fi
        \item The \textup{\CRPQ/\CRPQfin} containment problem is \pspace-hard under all semantics. \ifarxiv(Proposition~\ref{prop:crpq-crpqfin-pspace-h} in Appendix~\ref{app:other-results}.)\fi
        \item The \textup{\CRPQ/\CRPQfin} containment problem is in \pspace under standard semantics. \ifarxiv(Proposition~\ref{prop:crpq-crpqfin-pspace-st} in Appendix~\ref{app:other-results}.)\fi
        \item The \textup{\CRPQfin/\CRPQ} containment problem is in \pitwo, under all semantics. \ifarxiv(Proposition~\ref{prop:crpqfin-crpq-upper} in Appendix~\ref{app:other-results}.)\fi
    \end{enumerate}
\end{theorem}


\section{Discussion and Outlook}
\label{sec:discussion}

We have defined two possible injective semantics for CRPQs, providing two ways to extend the simple-path semantics of RPQs to the realm of CRPQs. On these semantics, we have shown that the containment problem differs drastically from the standard semantics, in some cases improving the complexity, and in some cases making the problem directly undecidable.

Both of these semantics are natural generalizations of simple-path semantics of RPQs. 
For instance, if we revert the role of edges and nodes, CRPQs under atom-injective semantics is present in the popular graph database Neo4j. While query-injective semantics is less common in practice, we still believe that this semantics, and in particular, looking for disjoint paths, may be useful for users and may provide an interesting feature for graph query languages. 
Further empirical investigation is needed to assess the practical usefulness of these two semantics. 

While the fragments of the class of CRPQs we have studied are probably the three most fundamental subclasses, there are other more fine-grained fragments based on the form of regular expressions used in the CRPQs, which are practically very relevant \cite{BonifatiMT-vdlbj20,BonifatiMT-www19}. These fragments have been studied under standard semantics \cite{FigueiraGKMNT20}, and it would be interesting to understand how they behave under injective semantics. A different direction is to consider larger classes of queries, understanding how injective semantics are extended, and the impact on the bounds for containment -- such as CRPQ with two-way navigation and union (UC2RPQ) \cite{CGLV00}, Extended CRPQ (ECRPQ) \cite{DBLP:journals/tods/BarceloLLW12}, or Regular Queries \cite{DBLP:journals/mst/ReutterRV17}.

We have limited our investigation to (extensions of) the simple-path semantics. Another possibility is to consider \emph{trail} semantics, which can be extended in a similar way to CRPQs, obtaining again two alternative semantics: query-edge-injective and atom-edge-injective, based on the notion of edge-injective homomorphisms. Many of our results can be extended to these semantics, and we suspect that complexities for query-edge-injective and query-injective coincide on all studied fragments, and likewise for atom-edge-injective and atom-injective.
In particular, while neither the undecidability nor the \pspace upper-bound seem to go through when simply reversing the role of nodes and edges, we believe that both proofs can be adapted.

One possible research direction is on another fundamental static analysis problem for CRPQs, namely the \emph{boundedness} problem, which checks whether a CRPQ is equivalent to a finite union of CQs. 
This problem is decidable for standard semantics \cite{BarceloF019}.

\bibliographystyle{ACM-Reference-Format}
\balance
\bibliography{long,bib}

\ifarxiv
  \newpage
  \appendix
  
\section{Appendix to Section~\ref{sec:preliminaries}}

\begin{proof}[Proof of Proposition~\ref{prop:semantics-st-qinj}]
The result is well-known for the case of standard semantics so we focus on query-injective semantics. 
Suppose $G=(V,E)$ and $Q$ is of the form $Q(\bar{x}) = A_1\land \dots\land A_m$. 
Recall $Q(\bar{x})$ is equivalent to a union $\mathcal{Q}_{\epsilon-free}$ of $\epsilon$-free CRPQs. 
Assume first that $\bar{v}\in Q(G)^{\qni}$. 
Then $\bar{v}\in Q'(G)^{\qni}$ for some $Q'\in \mathcal{Q}_{\epsilon-free}$ of the form $Q'(\bar{z}) = A'_1\land \dots\land A'_k$. 
Without loss of generality, assume $A_{k+1},\dots, A_m$ are precisely the atoms of $Q$ collapsed in the construction 
of $Q'$ (that is, we take the $\epsilon$-word on those atoms). 
In particular, the language of the atom $A'_i$ is the language of $A_i$ minus $\epsilon$, for $i\in \{1,\dots, k\}$. 
There exists then an injective mapping $\mu$ from $\vars(Q')$ to $V$ (satisfying $\mu(\bar{z}) = \bar{v}$), 
and for each atom $A'_i = x_i \xrightarrow{L_i'} y_i$, 
a simple path $\pi_i$ from $\mu(x_i)$ to $\mu(y_i)$ such that distinct paths do not share internal nodes. 
We can take the expansion $E\in\Exp{}(Q)$ produced by the expansion profile of $Q$ that maps $A_i$ to the word $\epsilon$, if $i\in\{k+1,\dots, m\}$, 
and maps $A_i$ to the label of $\pi_i$, for $i\in\{1,\dots, k\}$. We can define a homomorphism $h: E\to (G,\bar{v})$ 
by mapping the non-internal variables of $E$ according to $\mu$ and each atom expansion to the corresponding simple path in $G$. 
As the paths do not share internal nodes this is an injective homomorphism, and hence $E\injto (G, \bar{v})$.

For the other direction, suppose that $h: E\injto (G, \bar{v})$ for some expansion $E\in\Exp{}(Q)$. 
Let $\varphi$ be the expansion profile generating $E$. We construct an $\epsilon$-free CRPQ $Q'\in \mathcal{Q}_{\epsilon-free}$ 
as follows: if $\varphi$ assigns the word $\epsilon$ to $A_i$, then collapse $A_i$, otherwise, remove $\epsilon$ from its language. 
Suppose that $Q'$ is of the form $Q'(\bar{z}) = A'_1\land \dots\land A'_k$. 
We have that  $\bar{v}\in Q'(G)^{\qni}$. Indeed, 
define the mapping $\mu:\vars(Q')\to V$ as the restriction of $h$ to the non-internal variables of $E$, 
and define the path $\pi_i$ to be the image via $h$ of the expansion of the atom $A'_i$ in $E$. 
Since $h$ is injective, the $\pi_i$'s are simple paths and do not share internal nodes. Moreover, 
the mapping $\mu$ is injective. Hence $\bar{v}\in Q'(G)^{\qni}$, which implies that $\bar{v}\in Q(G)^{\qni}$ as required.

\end{proof}

\begin{proof}[Proof of Proposition~\ref{prop:semantics-ainj}]
The proof is analogous to the proof of Proposition~\ref{prop:semantics-st-qinj}, replacing injective homomorphisms
by atom-injective homomorphisms. 
\end{proof}

\section{Appendix to Section~\ref{sec:containment}}

\begin{proof}[Proof of Proposition~\ref{prop:cont-char-qinj}]
The proof is identical to the proof of Proposition~\ref{prop:cont-char-exp-st} replacing homomorphisms by injective homomorphisms. 
For the sake of completeness, we give the proof below.

Assume $Q_1\subseteq_{\qni} Q_2$ and take $E_1(\bar{y})\in \Exp{}(Q_1)$. 
Recall that we can see the query $E_1(\bar{y})$ as a graph database 
(of the same name) 
where each atom is interpreted as an edge. 
We have $E_1 \injto (E_1, \bar{y})$ and hence $\bar{y}\in Q_1(E_1)^{\qni}$. 
By hypothesis, $\bar{y}\in Q_2(E_1)^{\qni}$, that is, there is $E_2\in \Exp{}(Q_2)$ such that $E_2\injto (E_1, \bar{y})$, 
i.e, $E_2\injto E_1$. 
For the other direction, assume $\bar{v}\in Q_1(G)^{\qni}$ for some graph database $G$ and tuple $\bar{v}$ of nodes. There is an expansion $E_1\in \Exp{}(Q_1)$ 
such that $E_1 \injto (G, \bar{v})$. By hypothesis, there exists $E_2\in \Exp{}(Q_2)$ with $E_2\injto E_1$. By composition, 
we obtain $E_2 \injto (G, \bar{v})$, hence $\bar{v}\in Q_2(G)^{\qni}$. 
\end{proof}

\begin{proof}[Proof of Lemma~\ref{lemma:ainj-expansions}]
We start with (1)$\Rightarrow$(2). Let $h$ be a witness for $E\ainjto (G, \bar{v})$. 
Define the query $\widetilde{F}=E \land J$ where $J$ is the conjunction of all equality atoms $x=y$ with $x,y\in\vars(E)$ and $h(x)=h(y)$. 
Since $h$ is atom-injective, we have $F:=\widetilde{F}^{\collapse}\in \Exp{\ani}(Q)$. 
Moreover, there is $g:F\injto (G,\bar{v})$ as required. 
Indeed, let $\Phi:\vars(E)\to \vars(F)$ be the canonical renaming. 
For each $x\in\vars(F)$, 
we set $g(x) = h(x')$, where $x'$ is any variable in $E$ with $\Phi(x')=x$. By construction, $g$ is an injective homomorphism. 
Conversely, suppose $F=\widetilde{F}^{\collapse}$ for $\widetilde{F}=E\land J$ with expansion $E\in \Exp{}(Q)$ and let $g$ be a witness for $F\injto (G, \bar{v})$. 
Let $\Phi:\vars(E)\to \vars(F)$ be the canonical renaming. Observe that $\Phi$ is actually a homomorphism from $E$ to $F$. 
By definition of $\ani$-expansions, $\Phi$ is an atom-injective homomorphism from $E$ to $F$. 
By composing $\Phi$ with $g$ we obtain $E\ainjto (G, \bar{v})$. The case of $E'$ instead of $G$ is analogous. 
\end{proof}

\begin{proof}[Proof of Proposition~\ref{prop:cont-char-ainj}]
By Lemma~\ref{lemma:ainj-expansions} it suffices to consider the equivalence between items \eqref{it:prop:cont-char-ainj:1} and \eqref{it:prop:cont-char-ainj:2}.
Suppose first $Q_1\subseteq_{\ani} Q_2$ and take $F_1\in \Exp{\ani}(Q_1)$. 
We can see $F_1(\bar{y})$ as a graph database (of the same name) where each atom is interpreted as an edge. 
We have $F_1\injto (F_1, \bar{y})$ and, by Corollary~\ref{coro:a-inj-sem}, we obtain $\bar{y}\in Q_1(F_1)^{\ani}$. 
By hypothesis, $\bar{y}\in Q_2(F_1)^{\ani}$, that is, there is $E_2\in \Exp{}(Q_2)$ such that $E_2\ainjto (F_1, \bar{y})$, 
i.e, $E_2\ainjto F_1$. 
For the other direction, assume $\bar{v}\in Q_1(G)^{\ani}$ for some graph database $G$ and tuple $\bar{v}$ of nodes. 
By Corollary~\ref{coro:a-inj-sem}, there is $F_1\in \Exp{\ani}(Q_1)$ and $g:F_1\injto (G,\bar{v})$. 
By hypothesis, there is $E_2\in \Exp{}(Q_2)$ 
and $f: E_2\ainjto F_1$.
By composing $f$ with $g$ we obtain that $E_2\ainjto (G,\bar{v})$. 
We conclude that $\bar{v}\in Q_2(G)^{\ani}$.
\end{proof}

\section{Full proof of Theorem~\ref{thm:crpq-crpq-nodeinj-pspace}}
\label{sec:pr:lem:crpq-crpq-nodeinj-pspace}

	Let $Q_1(\bar x_1), Q_2(\bar x_2)$ be CRPQs; we want to test $Q_1 \subseteq_{\qni} Q_2$.

	We often blur the distinction between a CRPQ and an edge-labeled graph, whose edges are regular expressions. Hence, the \defstyle{degree} \resp{\defstyle{in-degree}; \defstyle{out-degree}{}} of a variable is the number of atoms containing it \resp{as a second variable; as a first variable}.

	\emph{High-level idea.}
	We first give a high-level description of the proof.
	We will work with polynomial-sized `abstractions' of expansions of $Q_1$. These abstractions contain, for each atom $A$ of $Q_1$, all the information on how the languages of $Q_2$ can be mapped into it. For example, it includes the information ``there is a partial run from state $q$ to state $q'$ of the NFA $\+A_L$ of language $L$ from $Q_2$ reading the expansion word of $A$'', or ``there is a partial run from the initial state of $\+A_L$ to $q$ reading some suffix of the expansion word of $A$''. Such an abstraction contains all the necessary information needed to retain from an expansion to check if it is a counter-example. Indeed, any expansion with the same abstraction as a counter-example will be a counter-example. 

	In order to test whether an abstraction $\alpha$ abstracts a counter-example, we need to consider all possible ways of injectively mapping an expansion of $Q_2$ to an expansion of $Q_1$. We call this \emph{morphism type}, which contains the information of where each atom expansion of $Q_2$ is mapped. For example, we can have the information that the path to which the expansion of atom $A$ of $Q_2$ is mapped starts at some internal node of the expansion of atom $A_1= x \xrightarrow{L_1} y$ of $Q_1$ then arrives to variable $y$ with state $q$ and continues reading the full expansion of atom $A_2 = y \xrightarrow{L_2} z$ arriving to variable $z$ with state $q'$, and it ends its journey by reading a prefix of the expansion of atom $A_3 = z \xrightarrow{L_3} t$ arriving to a final state at some internal node. For each morphism type, we can check if it is \emph{compatible} with an abstraction by checking, for example, that $\alpha$ indeed contains the information of having a partial run from $q'$ to a final state reading a prefix of the expansion of $A_3$. 
	
	The important property is that an abstraction $\alpha$ is compatible with a morphism type $\tau$ if{f} for every expansion $E_1$ of $Q_1$ with abstraction $\alpha$ there is an expansion $E_2$ of $Q_2$ with morphism type $\tau$ such that  $E_2 \injto E_1$. Hence, the \pspace algorithm simply guesses $\alpha$ and checks that $\alpha,\tau$ are not compatible, for every possible morphism type $\tau$. We now give some more details for these ideas. 

	\smallskip

	\begin{remark}\label{rk:one-var-per-atom}
		Any CRPQ is \qni-equivalent to one in which there is no variable $y$ incident to two atoms, with in-degree and out-degree equal to one. This is because $x \xrightarrow{L} y \land y \xrightarrow{L'}x'$ is equivalent (under \qni or standard semantics) to $x \xrightarrow{L \cdot L'} x'$ (assuming $y \not\in\set{x,x'}$).
	\end{remark}
	Due to Remark~\ref{rk:one-var-per-atom}, we can assume that the mapping from the expansion of $Q_2$ to the expansion of $Q_1$ is such that no two variables can be mapped to two internal nodes of an atom expansion.
	
	\begin{remark}\label{rk:no-epsilon}
		For every CRPQ $Q$ one can produce an equivalent union $Q'$ of CRPQs such that (i) no language of $Q'$ contains $\epsilon$ and (ii) there are no two distinct atoms $x \xrightarrow{L} y$ and $x \xrightarrow{L'} y$ in $Q'$ with some single-letter word $a \in \A$ in $L \cap L'$. Further, $Q'$ is an exponential union of polynomial-sized CRPQs, and testing whether a CRPQ is in the union is in \pspace.
	\end{remark}

	\emph{Terminology.} By \defstyle{path} (of an expansion or directed graph) we mean a directed path, that is, a sequence of edges of the form $\pi = (v_0,v_1) (v_1,v_2) \dotsb (v_{n-1},v_n)$. An \defstyle{internal node} of a path is any node excluding the initial and final ones (\ie, $v_0$ and $v_n$).
	For a path $\pi$, and a morphism $h$, we denote by $h(\pi)$ the path obtained by replacing each vertex $v$ with $h(v)$. 
	For a (directed) path $\pi$, we denote by \defstyle{$src(\pi)$ \resp{$tgt(\pi)$}} the  first \resp{last} vertex.
	A \defstyle{subpath} of a path $\pi$ as above, is a path of the form $(v_i,v_{i+1}) \dotsb (v_{j},v_{j+1})$ where $0 \leq i \leq j <n$ (in particular of length at least 1).
	An \defstyle{infix \resp{prefix, suffix}} of a path of $\pi$ is a subpath which does not contain $src(\pi)$ or $tgt(\pi)$  \resp{contains $tgt(\pi)$ and excludes $src(\pi)$, contains $src(\pi)$ and excludes $tgt(\pi)$}.
	We often blur the distinction between regular languages, regular expressions, NFA, and CRPQ atoms containing a regular language. For instance, we may write ``$q$ is a final state of atom $x \xrightarrow{L} y$'', meaning that it is a final state of the NFA representing $L$.

	\emph{Restriction of queries.}
	To simplify the proof, we will assume that $Q_1,Q_2$ have the following properties:
	\begin{itemize}
		\item there is no $\epsilon$ in any of the languages;
		\item there are no two atoms $x \xrightarrow{L} y$ and $x \xrightarrow{L'} y$ with some single-letter word $a \in \A$ in $L \cap L'$;
		\item the queries are connected.
	\end{itemize}
	We will later show how to lift these assumptions.

	\begin{remark}\label{rk:nodeinj-two-vars-in-atom}
		As a consequence of Remark~\ref{rk:one-var-per-atom}, for the \CRPQ/\CRPQ containment problem of $Q_1 \subseteq_\qni Q_2$ under \qni semantics, and assuming the properties above, we can restrict our attention to injective homomorphisms $E_1 \injto E_1$ (where $E_i \in \Exp{}(Q_i)$) such that if two distinct variables  $x,y \in \vars(Q_1)$ are mapped to distinct internal nodes of an atom expansion, then they must both be of degree 1.	That is, in view of the characterization of Proposition~\ref{prop:cont-char-qinj}, $Q_1 \subseteq_\qni Q_2$ if{f} for every $E_1 \in \Exp{}(Q_1)$ there is $E_2 \in \Exp{}(Q_2)$ such that $h:E_2 \injto E_1$, where $h$ has the property above.
	\end{remark}

	Without any loss of generality, let us assume that all the NFA of the languages of $Q_2$ have pairwise disjoint sets of states, and that they are complete and co-complete (\ie, for every letter and state there is an incoming and an outgoing transition with that letter). Let us consider $\+A_{Q_2}$ as an automaton having as transitions the (disjoint) union of all the transitions for the automata of $Q_2$. In this context we will denote by \emph{initial state} \resp{\emph{final state}} a state which is initial \resp{final} in the automaton from which it comes. 

	An \defstyle{abstraction of an expansion} $E_1$ of $Q_1$ is a mapping $\alpha$ from the atoms of $Q_1$ to subsets of $P$, where 
	\newcommand{\runST}[2]{\tup{#1\text{-}#2}}%
	\newcommand{\runSxT}[2]{\tup{#1\text{-}\!{|}\!\text{-}#2}}%
	\newcommand{\runSxxT}[2]{\tup{#1\text{-}\!{|}{\cdot}{\cdot}{|}\!\text{-}#2}}%
	\newcommand{\runxSTx}[2]{\tup{{\cdot}{\cdot}#1\text{-}#2{\cdot}{\cdot}}}%
	\begin{align*}
		P ={} 	& \set{\runST{q}{q'} : q,q' \text{ states of } \+A_{Q_2}} \cup {}
				\set{\runSxT{q}{q'} : q,q' \text{ states of } \+A_{Q_2}} \cup {}\\
				&\set{\runSxxT{q}{q'} : q,q' \text{ states of } \+A_{Q_2}} \cup {}
				\set{\runxSTx{q}{q'} : q,q' \text{ states of } \+A_{Q_2}}
	\end{align*} such that for every expansion $x \xrightarrow{w} y$ of an atom $A$ we have: 
	\begin{itemize}
		\item $\runST{q}{q'}\in \alpha(A)$ if there is a partial run of $\+A_{Q_2}$ from $q$ to $q'$ reading $w$; 
		\item $\runSxT{q}{q'}\in \alpha(A)$ if for some $w = u \cdot v$ with $u,v \neq \epsilon$ there is a partial run of $\+A_{Q_2}$ from $q$ to a final state reading $u$, and a partial run from an initial state to $q'$ reading $v$; 
		\item $\runSxxT{q}{q'}\in \alpha(A)$ if for some $w = u \cdot s \cdot v$ with $s,u,v \neq \epsilon$ there is a partial run of $\+A_{Q_2}$ from $q$ to a final state reading $u$, and a partial run of $\+A_{Q_2}$ from an initial state to $q'$ reading $v$;
		\item $\runxSTx{q}{q'}\in \alpha(A)$ if for some $w = u \cdot s \cdot v$ with $s,u,v \neq \epsilon$ there is a partial run of $\+A_{Q_2}$ from $q$ to $q'$ reading $s$.
	\end{itemize}
	Observe that the size of any abstraction is polynomial in $Q_1,Q_2$.

	The set of \defstyle{abstractions of $Q_1$}, is the set of abstractions of all its expansions. 
	\begin{claim}\label{cl:abstraction-pspace}
		Testing whether a mapping is an abstraction of $Q_1$ is in \pspace.
	\end{claim}
	Indeed, a standard pumping argument shows that if $\alpha$ is an abstraction of $Q_1$, it has a witnessing expansion of at most exponential size. Using this bound, an on-the-fly \pspace algorithm can guess the expansion for each atom  and check that each atom $A$ has abstraction $\alpha(A)$. This is done by guessing one letter at a time while keeping track of all possible partial runs it contains. The procedure also keeps a poly-sized counter to keep track of the size of the expansion being produced, and rejects the computation whenever the size exceeds the exponential bound.

	\smallskip

	\newcommand{\lGraph}{\mathbf{G}}
	Consider the directed graph $\lGraph$ consisting of replacing each atom $A = x \xrightarrow{L} y$ of $Q_1$ with a path $\pi^\lGraph_A$ of length 3 (\ie, adding two new internal vertices). 
	A \defstyle{morphism type} from $Q_2(\bar x_2)$ to $Q_1(\bar x_1)$ is a pair $(H,h)$ such that $h: H \injto \lGraph$, and $H$ is a graph resulting from replacing each atom $A =  x \xrightarrow{L} y$ of $Q_2$ with a (non-empty) path $\pi^H_A$ from $x$ to $y$. Further, we also ask that free variables are mapped accordingly, that is, $h(\bar x_2)=\bar x_1$.
	Figure~\ref{fig:ex:GH} contains an example of a morphism type $(H,h)$.
	\begin{figure}
		\includegraphics[scale=.8]{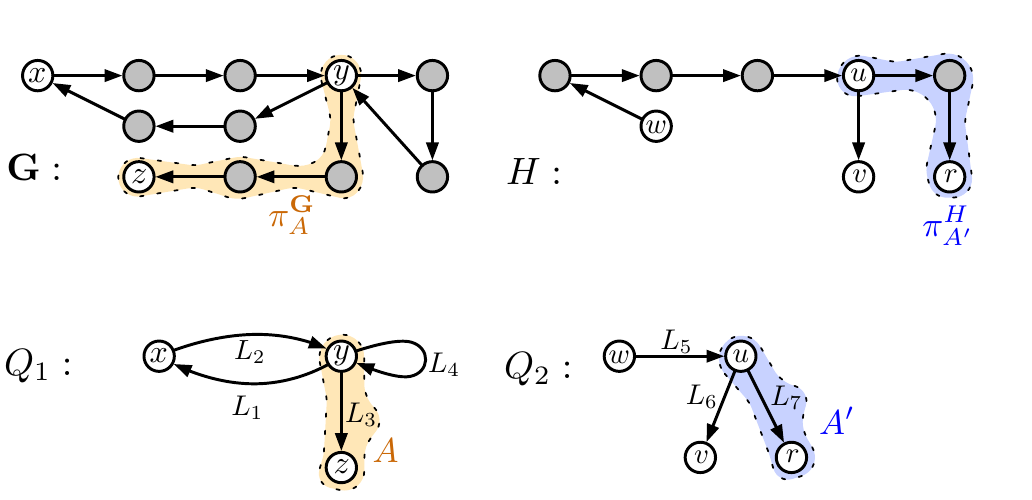}
		\caption{Example of definition of $\lGraph$ and morphism type $(H,h)$ from $Q_1,Q_2$. In this case, the injective morphism $h$ from $H$ to $\lGraph$ maps each node of $H$ to the node in the same position on $\lGraph$ (\eg, the lower-right node $r$ of $H$ maps to the lower-right node of $\lGraph$).}
		\label{fig:ex:GH}
	\end{figure}
	It follows that, by injectivity, the size of $H$ in any morphism type is linearly bounded on $Q_1$.
	\begin{claim}\label{cl:morphism-pspace}
		Testing whether a pair $(H,h)$ is a morphism type is in \pspace.
	\end{claim}

	A morphism type $(H,h)$ is \defstyle{compatible} with an abstraction $\alpha$ if there is a mapping $\lambda$ from the internal nodes of paths $\pi^H_A$ to states of $A$, such that for every atom $A=x_1 \xrightarrow{L} x_2$ of $Q_1$, 
	\begin{itemize}
		\item if there is an atom $A'$ of $Q_2$ and 
			an infix subpath $\pi$ of $\pi_{A'}^H$ 
			with 
			$h(\pi) = \pi^\lGraph_A$, then 
			$\runST{\lambda(src(\pi))}{\lambda(tgt(\pi))} \in \alpha(A)$ 
			(corresponding to case 1 in Figure~\ref{fig:cases-compatible});
		\item if there is an atom $A'$ of $Q_2$ and 
			a suffix subpath $\pi$ of $\pi_{A'}^H$ 
			with 
			$h(\pi) = \pi^\lGraph_A$, then 
			$\runST{\lambda(src(\pi))}{q_F} \in \alpha(A)$ 
			where $q_F$ is a final state of $A'$
			(case 2 in Fig.~\ref{fig:cases-compatible});
		\item if there is an atom $A'$ of $Q_2$ and 
			a suffix subpath $\pi$ of $\pi^H_{A'}$ 
			with 
			$h(\pi)$ being a prefix of $\pi^\lGraph_A$, then 
			$\runSxT{\lambda(src(\pi))}{q} \in \alpha(A)$ for some $q$ 
			(case 3 in Fig.~\ref{fig:cases-compatible});
		\item if there are atoms $A'_1,A'_2$ of $Q_2$, 
			a suffix subpath $\pi_1$ of $\pi^H_{A'_1}$ and 
			a prefix subpath $\pi_2$ of $\pi^H_{A'_2}$  
			with 
			$tgt(\pi_1)=src(\pi_2)$ and $h(\pi_1 \pi_2)=\pi^\lGraph_A$, then 
			$\runSxT{src(\pi_1)}{tgt(\pi_2)} \in \alpha(A)$ 
			(case 4 in Fig.~\ref{fig:cases-compatible});
		\item if there are atoms $A'_1,A'_2$ of $Q_2$ and a suffix subpath $\pi$ of $\pi_{A'_1}^H$ with 
			$tgt(\pi)=src(\pi_{A'_2}^H)$ and
			$h(\pi \pi_{A'_2}^H) = \pi^\lGraph_A$, then
			$\runSxT{\lambda(src(\pi))}{q_F} \in \alpha(A)$ 
			where $q_F$ is a final state of $A'_2$
			(case 5 in Fig.~\ref{fig:cases-compatible});
		\item if there are atoms $A'_1,A'_2$ of $Q_2$,
			a suffix subpath $\pi_1$ of $\pi_{A'_1}^H$, and
			a prefix subpath $\pi_2$ of $\pi_{A'_2}^H$
			with 
			$tgt(\pi_1) \neq src(\pi_2)$,
			$h(\pi_1)$ is a prefix of $\pi^\lGraph_A$ and
			$h(\pi_2)$ is a suffix of $\pi^\lGraph_A$, then
			$\runSxxT{\lambda(src(\pi_1))}{\lambda(tgt(\pi_2))} \in \alpha(A)$ 
			(case 6 in Fig.~\ref{fig:cases-compatible});
		\item  if there are atoms $A'_1, A'_2$ of $Q_2$, and 
			a suffix subpath $\pi$ of $\pi^H_{A'_1}$ 
			where 
			$h(\pi)$ is a prefix of $\pi^\lGraph_A$, 
			$h(\pi^H_{A'_2})$ is a suffix of $\pi^\lGraph_A$ and 
			$tgt(\pi) \neq src(\pi^H_{A'_2})$, then 
			$\runSxxT{\lambda(src(\pi))}{q_F} \in \alpha(A)$ 
			for some final state $q_F$ of $A'_2$ 
			(case 7 in Fig.~\ref{fig:cases-compatible});
		\item if there is an atom $A'$ of $Q_2$ and 
			a prefix subpath $\pi$ of $\pi_{A'}^H$ 
			with 
			$h(\pi) = \pi^\lGraph_A$, then 
			$\runST{q_0}{\lambda(tgt(\pi))} \in \alpha(A)$ 
			where $q_0$ is an initial state of $A'$
			(case 8 in Fig.~\ref{fig:cases-compatible});
		\item if there is an atom $A'$ of $Q_2$ 
			with 
			$h(\pi_{A'}^H) = \pi^\lGraph_A$, then 
			$\runST{q_0}{q_F} \in \alpha(A)$ 
			where $q_0$/$q_F$ is an initial/final state of $A'$
			(case 9 in Fig.~\ref{fig:cases-compatible});
		\item if there is an atom $A'$ of $Q_2$  
			with 
			$h(\pi_{A'}^H)$ a prefix of $\pi^\lGraph_A$, then
			$\runSxT{q_0}{q} \in \alpha(A)$ 
			where $q_0$ is an initial state of $A'$ and $q$ is any state.
			(case 10 in Fig.~\ref{fig:cases-compatible});
		\item if there are atoms $A'_1,A'_2$ of $Q_2$ and
			a prefix subpath $\pi_2$ of $\pi_{A'_2}^H$
			with 
			$tgt(\pi_{A'_1}^H) = src(\pi_2)$,
			$h(\pi_{A'_1}^H \pi_2) = \pi^\lGraph_A$, then
			$\runSxT{q_0}{\lambda(tgt(\pi_2))} \in \alpha(A)$ 
			where $q_0$ is an initial state of $A'_1$
			(case 11 in Fig.~\ref{fig:cases-compatible});
		\item if there are atoms $A'_1,A'_2$ of $Q_2$ 
			with 
			$tgt(\pi_{A'_1}^H) = src(\pi_{A'_2}^H)$,
			$h(\pi_{A'_1}^H \pi_{A'_2}^H) = \pi^\lGraph_A$, then
			$\runSxT{q_0}{q_F} \in \alpha(A)$ 
			where $q_0$ is an initial state of $A'_1$ and $q_F$ a final state of $A'_2$
			(case 12 in Fig.~\ref{fig:cases-compatible});
		\item if there is an atom $A'$ of $Q_2$ and
			a prefix $\pi$ of $\pi_{A'}^H$ 
			with 
			$h(\pi)$ a suffix of $\pi^\lGraph_A$, then and
			$\runSxxT{q}{\lambda(tgt(\pi))} \in \alpha(A)$ for some $q$ 
			(case 13 in Fig.~\ref{fig:cases-compatible});
		\item if there are atoms $A'_1,A'_2$ of $Q_2$ and
			a prefix subpath $\pi_2$ of $\pi_{A'_2}^H$
			with 
			$tgt(\pi_{A'_1}^H) \neq src(\pi_2)$,
			$h(\pi_{A'_1}^H)$ is a prefix of $\pi^\lGraph_A$ and
			$h(\pi_2)$ is a suffix of $\pi^\lGraph_A$, then
			$\runSxxT{q_0}{\lambda(tgt(\pi_2))} \in \alpha(A)$ 
			where $q_0$ is an initial state of $A'_1$
			(case 14 in Fig.~\ref{fig:cases-compatible});
		\item if there are atoms $A'_1,A'_2$ of $Q_2$ 
			with 
			$tgt(\pi_{A'_1}^H) \neq src(\pi_{A'_2}^H)$,
			$h(\pi_{A'_1}^H)$ is a prefix of $\pi^\lGraph_A$ and
			$h(\pi_{A'_2}^H)$ is a suffix of $\pi^\lGraph_A$, then
			$\runSxxT{q_0}{q_F} \in \alpha(A)$ 
			where $q_0$ is an initial state of $A'_1$ and $q_F$ a final state of $A'_2$
			(case 15 in Fig.~\ref{fig:cases-compatible});
		\item  if there is an atom $A'$ of $Q_2$ 
			with 
			$h(\pi^H_{A'})$ being an infix of $\pi^\lGraph_A$, then 
			$\runxSTx{q_0}{q_F} \in \alpha(A)$ for $q_0$ and $q_F$ initial and final states of $A'$ 
			(case 16 in Fig.~\ref{fig:cases-compatible});
		\item  if there is an atom $A'$ of $Q_2$
			where 
			$h(\pi^H_{A'})$ is a suffix of $\pi^\lGraph_A$, then 
			$\runSxT{q}{q_F} \in \alpha(A)$ 
			for some state $q$ and some final state $q_F$ of $A'$
			(case 17 in Fig.~\ref{fig:cases-compatible}).
	\end{itemize}

\begin{figure}
	\includegraphics[scale=.8]{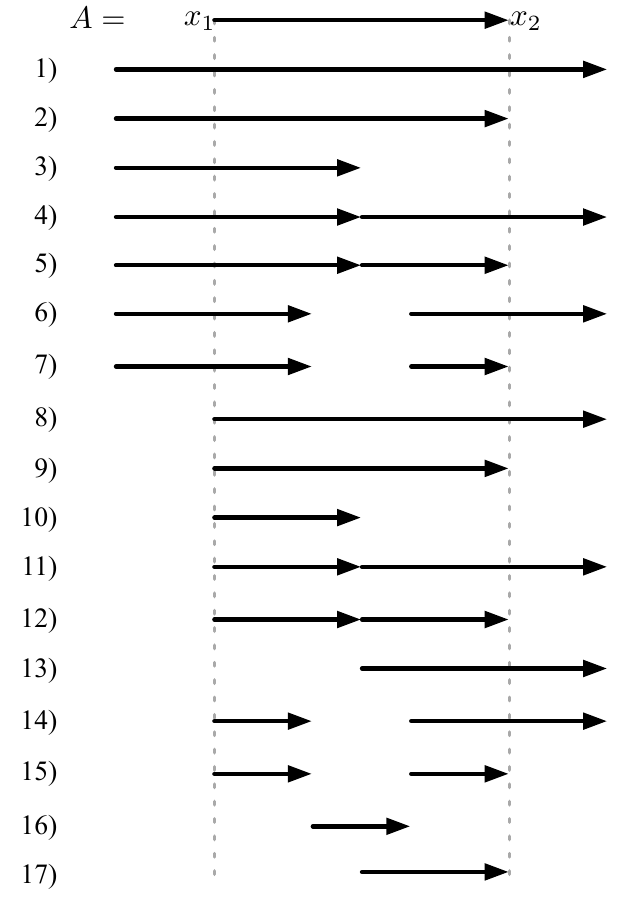}
	\caption{All the possible cases for compatibility. Observe that there are no other cases due to Remark~\ref{rk:nodeinj-two-vars-in-atom}.}
	\label{fig:cases-compatible}
\end{figure}

The following statement is a direct consequence of $H,h,\lambda$ being polynomially bounded and each of the conditions above being polynomial-time testable.
	\begin{claim}\label{cl:compatible-pspace}
		Testing whether a morphism type is compatible with an abstraction is in \pspace.
	\end{claim}

The key property of compatible abstractions, is that they allow to capture whether an expansion of $Q_1$ with a given abstraction is a counter-example for the containment problem $Q_1 \subseteq_\qni Q_2$.
	\begin{claim}\label{cl:compatible-counterexample}
		The following are equivalent:
		\begin{enumerate}
			\item There is a morphism type compatible with an abstraction $\alpha$;
			\item for every expansion $E_1 \in \Exp{}(Q_1)$ with abstraction $\alpha$ there exists some expansion $E_2 \in \Exp{}(Q_2)$ such that $E_2 \injto E_1$;
			\item there is an expansion $E_1 \in \Exp{}(Q_1)$ with abstraction $\alpha$ and an expansion $E_2 \in \Exp{}(Q_2)$ such that $E_2 \injto E_1$.
		\end{enumerate}
	\end{claim}
	\begin{proof}
		\proofcase{$1 \Rightarrow 2)$}
		Assume $E_1 \in \Exp{}(Q_1)$ has abstraction $\alpha$, and $(H,h)$ is a compatible morphism type through the mapping $\lambda$.
		Fore every atom $A$ of $Q_1$, we replace every path $\pi$ of $H$ such that $h(\pi)= \pi^\lGraph_A$ with the expansion of $A$ in $E_1$.
		The remaining edges of $H$ are all part of paths which map \emph{partially} to some $\pi^\lGraph_A$, these are replaced with paths according to the witnessing words for the elements of the form $\runSxxT{q}{q'}$ and $\runxSTx{q}{q'}$ in each $\alpha(A)$. 
		
		\begin{figure}
			\includegraphics[scale=.8]{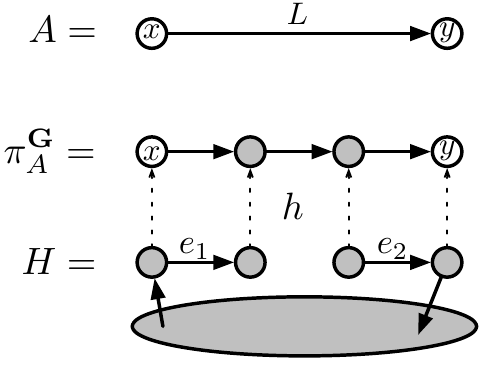}
			\caption{Example for proof of Claim~\ref{cl:compatible-counterexample}}
			\label{fig:ex:cl:compatible-counterexample}
		\end{figure}
		For example, in the case depicted in Figure~\ref{fig:ex:cl:compatible-counterexample}, we know that the expansion $w \in L$ of $A$ is of the form $w=u \cdot s \cdot v$ such that there is a partial run of $\+A_{Q_2}$ from $q$ to some final state reading $u$, and a partial run of $\+A_{Q_2}$ from an initial state to $q'$ reading $v$. Hence, we replace edge $e_1$ with a path reading $u$, and edge $e_2$ with a path reading $v$.
		It follows that by the definition of abstraction the resulting CQ (\ie, the query represented by the resulting edge-labeled graph) is an expansion of $Q_2$ that maps to $E_1$ through an injective homomorphism.

		\proofcase{$2\Rightarrow 3)$} This is trivial since an abstraction of $Q_1$ is the abstraction of an expansion thereof.

		\proofcase{$3 \Rightarrow 1)$}
		Take any $E_1 \in \Exp{}(Q_1)$ with abstraction $\alpha$ and $E_2 \in \Exp{}(Q_2)$ such that $g:E_2 \injto E_1$. We now build a graph $H$ from $E_2$ as follows.
		In the sequel, whenever we say that we replace a path $\pi$ with a path of length $n$, we mean that we (1) remove all internal nodes of $\pi$, and all edges incident to these and (2) we add $n-1$ fresh nodes and $n$ edges in such a way that there is a path of length $n$ from $src(\pi)$ to $tgt(\pi)$.
		For every path $\pi$ in $E_1$ corresponding to the expansion of atom $A$ :
		\begin{itemize}
			\item If there is a prefix \resp{suffix} of $\pi$ which has no $g$-preimage, we replace the path $g^{-1}(\pi)$ with just one edge, and we send the variable to the first internal node of $\pi^\lGraph_A$. For the remaining cases let us assume that every node of $\pi$ has a $g$-preimage.
			\item If $g^{-1}(\pi)$ is a path with no $Q_2$-variables as internal nodes, then we replace it with an unlabeled path $\pi'$  of length 3. 
			We define $h$ to map the first \resp{second} internal node of $\pi'$ to the first \resp{second} internal node of $\pi_A^\lGraph$.
			\item If $g^{-1}(\pi)$ is a path that contains one variable $x \in \vars(Q_2)$ as internal node, we replace the first half path of $g^{-1}(\pi)$ until $x$ with just an unlabeled edge, and the other half with a path $\pi'$ of length 2. We set $h$ to map the variable $x$ to the first internal node of $\pi^\lGraph_A$ and the internal node of $\pi$ to the second internal node of $\pi^\lGraph_A$.
			\item If $g^{-1}(\pi)$ contains two variables and two disjoint paths, we replace each of them with an unlabeled edge. We send the variables (\ie, the endpoints of the paths) correspondingly to the two internal vertices of $\pi^\lGraph_A$.
		\end{itemize} 
		The resulting graph $H$ and mapping $h$ is a morphism type which is compatible with $\alpha$.
	\end{proof}
	
	Finally, the \pspace algorithm guesses a mapping $\alpha$ from the atoms of $Q_1$ to subsets of $P$, checks that $\alpha$ is an abstraction of $Q_1$ (in \pspace due to Claim~\ref{cl:abstraction-pspace}), and checks that there is no morphism type $(H,h)$ which is compatible with $\alpha$ (in \pspace, due to Claims~\ref{cl:morphism-pspace} and \ref{cl:compatible-pspace}, and closure under complement of \pspace). Due to Claim~\ref{cl:compatible-counterexample}, if the algorithm succeeds, then any expansion $E_1$ of $Q_1$ is a counter-example, and thus $Q_1 \not\subseteq_\qni Q_2$; otherwise, for every expansion $E_1$ with abstraction $\alpha$ there is a compatible morphism type, which means that $E_1$ is not a counter-example and hence $Q_1 \subseteq_{\qni} Q_2$.

	First  note that if $Q_2$ is not connected, we can adapt the \pspace algorithm by testing that there are no morphism types for the connected components of $Q_2$ such that all of them are compatible with the guessed abstraction of $Q_1$.
	Further, observe that the procedure can be extended to an exponential union of polynomial-sized CRPQs: the \pspace algorithm first chooses one CRPQ $Q_1$ from the left-hand side union, guesses an abstraction of $Q_1$ and checks that no CRPQ coming from the right-hand side union has a compatible morphism type. 
	For this reason, combined with Remark~\ref{rk:no-epsilon}, the same argument extends to (unions of) arbitrary CRPQ's.
	\qed

\section{Full Proof of Theorem~\ref{theo:undec-atom-ni}}
\label{sec:app-full-undec}

    We reduce from \emph{Post Correspondence Problem} (PCP), a well-known undecidable problem. 
    An instance of the PCP is a sequence of pairs $(u_1,v_1),\dots, (u_\ell, v_\ell)$, where $u_i$ and $v_i$ are non-empty words over an alphabet $\Sigma$. 
    The goal is to decide whether there is a \emph{solution}, that is, a sequence $i_1,\dots,i_k$ of indices from $\{1,\dots, \ell\}$, with $k\geq 1$, such that  
    the words $u_{i_1}\cdots u_{i_k}$ and $v_{i_1}\cdots v_{i_k}$ coincide. 
    
    For an alphabet $\A$, we denote by $\widehat{\A}$ the alphabet $\widehat{\A}=\{\widehat{a}: a\in \A\}$. 
    Let $(u_1,v_1),\dots, (u_\ell, v_\ell)$ be a PCP instance and let $\Sigma$ be its underlying alphabet. 
    Let $\Ic$ and $\A$ be the alphabets $\Ic=\{I_1,\dots,I_\ell\}$ and $\A=\Sigma\cup \mathbb{I}\cup \{\#, \#_\infty, \square, \$, \$', \$_\infty, \blacksquare, \blacksquare'\}$. 
    We construct Boolean CRPQs $Q_1$ and $Q_2$ over alphabet $\A\cup\widehat{\A}$ 
    such that the PCP instance $(u_1,v_1),\dots, (u_\ell, v_\ell)$ has a solution if and only if $Q_1\not\subseteq_{\ani}Q_2$. 
    In particular, the PCP instance has a solution if and only if there exists a counterexample for $\ani$-semantics, \ie, an $\ani$-expansion $F$ of $Q_1$ such that there is no expansion $E\in \Exp{}(Q_2)$ with $E\ainjto F$.  
    
    The symbols in $\A$ are associated with the words $u_i$, while the symbols in $\widehat{\A}$ with the words $v_i$. 
    For each $u_i=a_1\cdots a_k$, we define the word $U_i = a_1\, \$\, \blacksquare \,a_2 \,\$\, \blacksquare \, \cdots \,a_k \,\$' \,\blacksquare'$. 
    Similarly, for each $v_i=a_1\cdots a_k$ we define $V_i=\widehat{\blacksquare}' \,\widehat{\$}' \, \widehat{a}_k \,\widehat{\blacksquare}\, \widehat{\$} \, \widehat{a}_{k-1}\, \cdots \, \widehat{\blacksquare}\, \widehat{\$}\, \widehat{a}_1$. 
    The CRPQ $Q_1$ is defined as follows (see Figure~\ref{fig:app-query-q1-q2} for an illustration):
    \begin{align*}
    &Q_1  = y_1\xrightarrow{L_{I}} x \land y_2\xrightarrow{\widehat{L}_{a}} x \land x\xrightarrow{\widehat{L}_{I}} z_1 \land x \xrightarrow{L_{a}} z_2 \\
    &  \qquad \qquad  \land x \xrightarrow{\square} x' \land x \xrightarrow{\widehat{\blacksquare}} x'  \land x' \xrightarrow{\widehat{\square}} x \land x' \xrightarrow{\blacksquare} x \\
    & \qquad \qquad \land y'_1\xrightarrow{\#_\infty} y_1\land y'_2\xrightarrow{\widehat{\$}_\infty} y_2 \land z_1\xrightarrow{\widehat{\#}_\infty \widehat{\#}} z_1' \land z_2\xrightarrow{\$_\infty \$} z_2'
    \end{align*}
    where:
    \begin{align*}
    & L_I=(\square\, \#\, \Ic)^+\qquad  \widehat{L}_I=(\widehat{\Ic}\, \widehat{\#}\, \widehat{\square})^+ \\
    & L_a=(U_1+\cdots+U_\ell)^+ \qquad \widehat{L}_a=(V_1+\cdots+V_\ell)^+ 
     \end{align*}
        \begin{figure}
            \begin{center}
                \includegraphics[width=.5\textwidth]{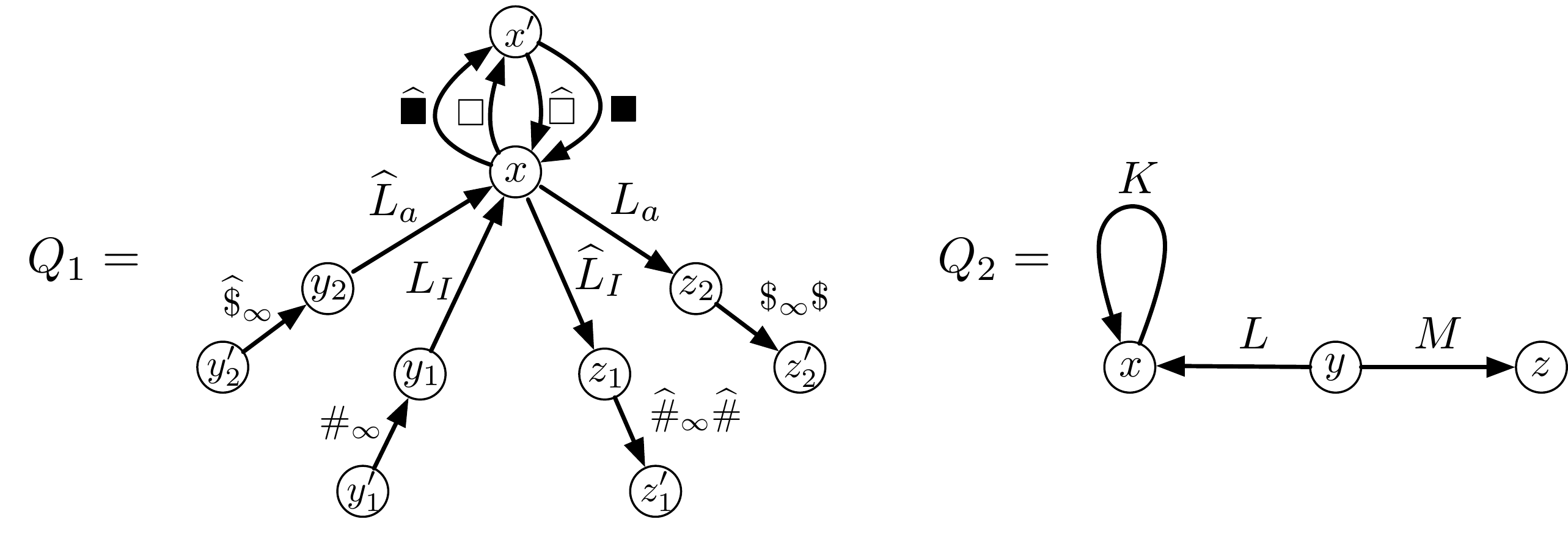}
            \end{center}
            \caption{The Boolean CRPQs $Q_1$ and $Q_2$ from the reduction.}
            \label{fig:app-query-q1-q2}
        \end{figure}
        Intuitively, a word from $L_I$ corresponds to a choice of indices from $\{1,\dots, \ell\}$, similarly for $\widehat{L}_I$. 
        On the other hand, a word from $L_a$ \resp{$\widehat{L}_{a}$} corresponds to a choice of words from $\{u_1,\dots, u_\ell\}$ \resp{$\{v_1,\dots, v_\ell\}$}. 
        
        We are interested in a particular type of $\ani$-expansions of $Q_1$ that we call \emph{well-formed} and define below.
        The idea is that well-formed $\ani$-expansions correspond to solutions of the PCP instance. In particular, 
        if there is a well-formed $\ani$-expansion of $Q_1$ then there is a solution to the PCP instance and vice versa. 
       We then show how to construct $Q_2$ such that an $\ani$-expansion of $Q_1$ is well-formed if and only if 
       it is a counterexample for $Q_1\subseteq_{\ani}Q_2$. 
       
             \begin{figure}
            \begin{center}
                \includegraphics[width=.37\textwidth]{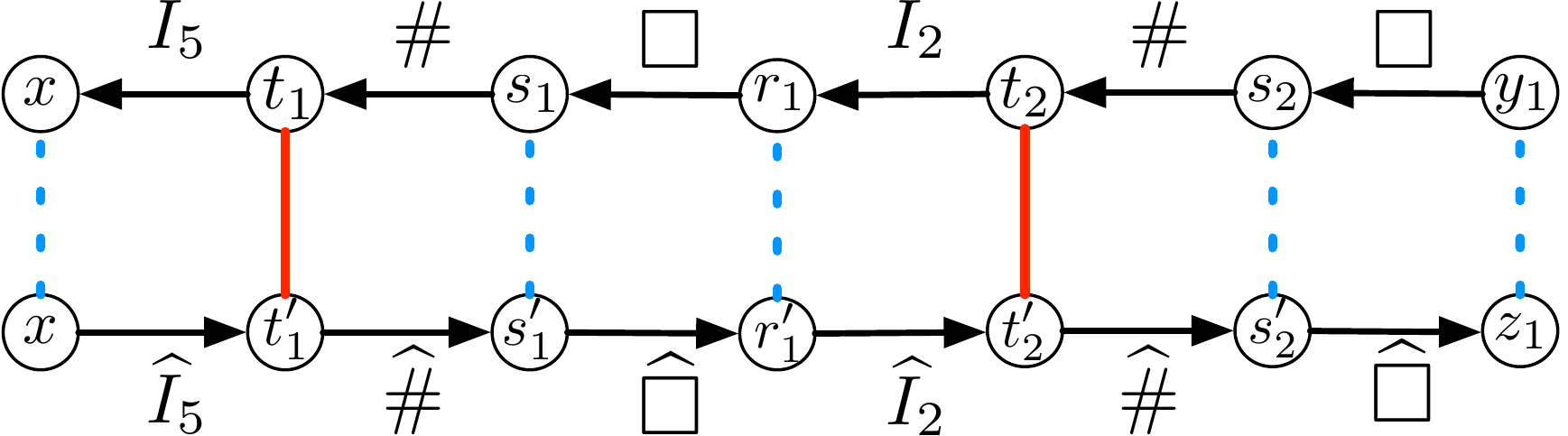}
            \end{center}
            \caption{Example of the $I$-$\widehat{I}$-condition of well-formed expansions of $Q_1$. 
            We show the expansions $y_1\xrightarrow{w_{I}} x$ and 
             $x\xrightarrow{\widehat{w}_{I}} z_1$ of the atoms $y_1\xrightarrow{L_{I}} x$ and 
             $x\xrightarrow{\widehat{L}_{I}} z_1$, respectively. The words $w_{I} = \square\, \#\, I_2 \, \square\, \#\, I_5$ and 
             $\widehat{w}_{I} =  \widehat{I}_5\, \widehat{\#}\, \widehat{\square} \, \widehat{I}_2 \, \widehat{\#}\, \widehat{\square}$ encode the sequence of indices $5,2$. 
             Dotted blue lines indicate pairs of equal variables while red lines indicate distinct variables. We have some extra symbols $\#, \widehat{\#}, \square, \widehat{\square}$.  
             }
            \label{fig:app-well-formed}
        \end{figure}
        
       Let $F$ be an $\ani$-expansion of $Q_1$ such that  $F=\widetilde{F}^{\collapse}$ for $\widetilde{F} = E\land J$ (here $E\in \Exp{}(Q_1)$ and $J$ are the equality atoms).  We say that $F$ is \emph{well-formed} if it satisfies the following four conditions:
       
       \begin{enumerate}
       \item \emph{$I$-$\widehat{I}$-condition :} 
       This condition applies to the atoms $y_1\xrightarrow{L_{I}} x$  and 
       $x\xrightarrow{\widehat{L}_{I}} z_1$ of $Q_1$. Let $y_1\xrightarrow{w_{I}} x$ and $x\xrightarrow{\widehat{w}_{I}} z_1$ be the 
       expansions associated to the atoms  $y_1\xrightarrow{L_{I}} x$  and 
       $x\xrightarrow{\widehat{L}_{I}} z_1$ in the expansion $E$.  
       The condition requires the words $w_{I}$ and 
       $\widehat{w}_{I}$ to be of the form 
    $w_I=\square\, \#\, I_{i_k}\, \cdots \,\square \,\# \,I_{i_1}$ and $\widehat{w}_I= \widehat{I}_{i_1}\, \widehat{\#} \, \widehat{\square} \, \cdots \, \widehat{I}_{i_k} \, \widehat{\#}\, \widehat{\square}$, 
    for a sequence of indices $i_1,\dots, i_k\in\{1,\dots, \ell\}$. 
In other words, the expansions of the atoms $y_1\xrightarrow{L_{I}} x$  and 
       $x\xrightarrow{\widehat{L}_{I}} z_1$ correspond to the same sequence on indices. 
       The $I$-$\widehat{I}$-condition also requires a particular behavior on the equality atoms $J$ (see Figure~\ref{fig:app-well-formed}). 
   Suppose the expansion $y_1\xrightarrow{w_{I}} x$ is of the form:
    \begin{align*} 
    \qquad \quad y_1 \xrightarrow{\square} s_{k} \xrightarrow{\#} t_{k} \xrightarrow{I_{i_k}} r_{{k-1}}  \xrightarrow{\square} & s_{{k-1}}  \xrightarrow{\#} t_{{k-1}} \xrightarrow{I_{i_{k-1}}} r_{{k-2}} \cdots \\
    & \cdots r_{1}\xrightarrow{\square} s_{1} \xrightarrow{\#} t_{1} \xrightarrow{I_{i_1}} x
    \end{align*}
     and $x \xrightarrow{\widehat{w}_{I}} z_1$ is of the form:
    \begin{align*}
       \qquad \quad x \xrightarrow{\widehat{I}_{i_1}} t'_{1} \xrightarrow{\widehat{\#}} s'_{1} \xrightarrow{\widehat{\square}} r'_{{1}} \xrightarrow{\widehat{I}_{i_2}} & t'_{{2}} \xrightarrow{\widehat{\#}} s'_{{2}}  \xrightarrow{\widehat{\square}}  r'_{{2}} \cdots \\
    & \cdots r'_{{k-1}} \xrightarrow{\widehat{I}_{i_k}}  t'_{k} \xrightarrow{\widehat{\#}} s'_{k} \xrightarrow{\widehat{\square}} z_1
    \end{align*}
    Then the relation $=_{\widetilde{F}}$ produced by the equality atoms $J$ satisfies:
    \begin{enumerate}
   \item $t_{1}\neq_{\widetilde{F}} t'_{1}, \cdots, t_{k}\neq_{\widetilde{F}} t'_{k}$
   \item $s_{1} =_{\widetilde{F}} s'_{1}, \cdots, s_{k} =_{\widetilde{F}} s'_{k}$
   \item $r_{1} =_{\widetilde{F}} r'_{1}, \cdots, r_{{k}} =_{\widetilde{F}} r'_{{k}}$
     \end{enumerate}
     where $r_k := y_1$ and $r'_k := z_1$.
     
      \item \emph{$I$-$a$-condition :} 
       This condition applies to the atoms $y_1\xrightarrow{L_{I}} x$  and 
       $x \xrightarrow{L_{a}} z_2$ of $Q_1$. Let $y_1\xrightarrow{w_{I}} x$ and $x\xrightarrow{w_{a}} z_2$ be the 
       expansions associated to the atoms  $y_1\xrightarrow{L_{I}} x$  and 
       $x \xrightarrow{L_{a}} z_2$ in the expansion $E$.  
       The condition requires the words $w_{I}$ and 
       $w_{a}$ to be of the form 
    $w_I=\square\, \#\, I_{i_k}\, \cdots \,\square \,\# \,I_{i_1}$ and $w_a=U_{i_1}\cdots U_{i_k}$, 
    for a sequence of indices $i_1,\dots, i_k\in\{1,\dots, \ell\}$. 
Intuitively, the word $w_a$ chooses words from $\{u_1,\dots, u_\ell\}$ according to the sequence $i_1,\dots, i_k$.
       The $I$-$a$-condition also requires a particular behavior on the equality atoms $J$. 
   Suppose the expansion $y_1\xrightarrow{w_{I}} x$ is of the form:
    \begin{align*}
    \qquad \quad y_1 \xrightarrow{\square} s_{k} \xrightarrow{\#} t_{k} \xrightarrow{I_{i_k}} r_{{k-1}}  \xrightarrow{\square} & s_{{k-1}}  \xrightarrow{\#} t_{{k-1}} \xrightarrow{I_{i_{k-1}}} r_{{k-2}} \cdots \\
    & \cdots r_{1}\xrightarrow{\square} s_{1} \xrightarrow{\#} t_{1} \xrightarrow{I_{i_1}} x
    \end{align*}
    and the expansion $x \xrightarrow{w_{a}} z_2$ is of the form:
    \begin{align*}
     \qquad \quad x \xrightarrow{\widetilde{U}_{i_1}} s'_{1} \xrightarrow{\blacksquare'} r'_{{1}} \xrightarrow{\widetilde{U}_{i_2}} s'_{{2}}  \xrightarrow{\blacksquare'}  r'_{{2}} \cdots r'_{{k-1}} \xrightarrow{\widetilde{U}_{i_k}} s'_{k} \xrightarrow{\blacksquare'} z_2
    \end{align*}
   where $\widetilde{U}_i$ is the word obtained from $U_i$ by removing the last symbol $\blacksquare'$.
   Then the relation $=_{\widetilde{F}}$ produced by the equality atoms $J$ satisfies:
     \begin{enumerate}
   \item $t_{j} \neq_{\widetilde{F}} t$, for every internal variable $t$ of the expansion  $r'_{{j-1}} \xrightarrow{\widetilde{U}_{i_j}} s'_{{j}}$ (here $r'_{0}:= x$)
\item $s_{1} =_{\widetilde{F}} s'_{1}, \cdots, s_{k} =_{\widetilde{F}} s'_{k}$
\item $r_{1} =_{\widetilde{F}} r'_{1}, \cdots, r_{{k}} =_{\widetilde{F}} r'_{{k}}$
     \end{enumerate}
where $r_k := y_1$ and $r'_k := z_2$.

          \item \emph{$\widehat{a}$-$\widehat{I}$-condition :} 
          This is analogous to the $I$-$a$-condition and applies to the atoms $y_2\xrightarrow{\widehat{L}_{a}} x$  and 
       $x\xrightarrow{\widehat{L}_{I}} z_1$ of $Q_1$. Let $y_2\xrightarrow{\widehat{w}_{a}} x$  and 
       $x\xrightarrow{\widehat{w}_{I}} z_1$ be the 
       expansions associated to the atoms  $y_2\xrightarrow{\widehat{L}_{a}} x$  and 
       $x\xrightarrow{\widehat{L}_{I}} z_1$ in the expansion $E$.  
       The condition requires the words $\widehat{w}_{a}$ and 
       $\widehat{w}_{I}$ to be of the form 
    $\widehat{w}_I= \widehat{I}_{i_1}\, \widehat{\#} \, \widehat{\square} \, \cdots \, \widehat{I}_{i_k} \, \widehat{\#}\, \widehat{\square}$ and $\widehat{w}_a=V_{i_k}\cdots V_{i_1}$, 
    for a sequence of indices $i_1,\dots, i_k\in\{1,\dots, \ell\}$. 
That is, the word $\widehat{w}_a$ chooses words from $\{v_1,\dots, v_\ell\}$ according to the sequence $i_1,\dots, i_k$.
       We also require some conditions on the equality atoms $J$. 
   Suppose the expansion $y_2\xrightarrow{\widehat{w}_{a}} x$ is of the form:
    \begin{align*}
   \qquad y_2 \xrightarrow{\widehat{\blacksquare}'} s_{k} \xrightarrow{\widetilde{V}_{i_{k}}} r_{{k-1}} \xrightarrow{\widehat{\blacksquare}'}  s_{{k-1}} \xrightarrow{\widetilde{V}_{i_{k-1}}}  r_{{k-2}} \cdots 
    r_{1} \xrightarrow{\widehat{\blacksquare}'} s_{1} \xrightarrow{\widetilde{V}_{i_{1}}} x
    \end{align*}
    where $\widetilde{V}_i$ is the word obtained from $V_i$ by removing the first symbol $\widehat{\blacksquare}'$. 
    Suppose also the expansion $x\xrightarrow{\widehat{w}_{I}} z_1$ is of the form:
   \begin{align*}
       \qquad \quad x \xrightarrow{\widehat{I}_{i_1}} t'_{1} \xrightarrow{\widehat{\#}} s'_{1} \xrightarrow{\widehat{\square}} r'_{{1}} \xrightarrow{\widehat{I}_{i_2}} & t'_{{2}} \xrightarrow{\widehat{\#}} s'_{{2}}  \xrightarrow{\widehat{\square}}  r'_{{2}} \cdots \\
    & \cdots r'_{{k-1}} \xrightarrow{\widehat{I}_{i_k}}  t'_{k} \xrightarrow{\widehat{\#}} s'_{k} \xrightarrow{\widehat{\square}} z_1
    \end{align*}
   Then the relation $=_{\widetilde{F}}$ produced by the equality atoms $J$ satisfies:
     \begin{enumerate}
\item $t'_{j} \neq_{\widetilde{F}} t$, for every internal variable $t$ of the expansion $s_{j} \xrightarrow{\widetilde{V}_{i_{j}}} r_{{j-1}}$ (here $r_{0}:= x$)
\item $s_{1} =_{\widetilde{F}} s'_{1}, \cdots, s_{k} =_{\widetilde{F}} s'_{k}$
\item $r_{1} =_{\widetilde{F}} r'_{1}, \cdots, r_{{k}} =_{\widetilde{F}} r'_{{k}}$
     \end{enumerate}
     where $r_k := y_2$ and $r'_k := z_1$.
     
          \item \emph{$\widehat{a}$-$a$-condition :} 
       This condition applies to the atoms  $y_2\xrightarrow{\widehat{L}_{a}} x$  and 
       $x \xrightarrow{L_{a}} z_2$ of $Q_1$. Let $y_2\xrightarrow{\widehat{w}_{a}} x$ and $x\xrightarrow{w_{a}} z_2$ be the 
       expansions associated to the atoms  $y_2\xrightarrow{\widehat{L}_{a}} x$  and 
       $x \xrightarrow{L_{a}} z_2$ in the expansion $E$.  
       The condition requires the words $\widehat{w}_{a}$ and 
       $w_{a}$ to be of the form 
    $\widehat{w}_{a}=\clubsuit \, \clubsuit \, \widehat{a}_n \, \cdots \,\clubsuit \,\clubsuit \, \widehat{a}_1$ and 
    $w_a = a_1 \, \clubsuit \, \clubsuit \, \cdots \, a_n \, \clubsuit \, \clubsuit$, for a word $a_1\cdots a_n \in\Sigma^*$ (recall $\Sigma$ is the alphabet of the PCP instance). 
    Here, $\clubsuit$ is a placeholder representing some symbol.     
Intuitively, the word $\widehat{w}_{a}$ and $w_a$ represent the same word from $\Sigma^*$. 
The $\widehat{a}$-$a$-condition also requires some conditions on the equality atoms $J$. 
     Assume the expansion $y_2\xrightarrow{\widehat{w}_{a}} x$ is of the form:
\begin{align*}
    \qquad \quad y_2 \xrightarrow{\clubsuit} s_{n} \xrightarrow{\clubsuit} t_{n} \xrightarrow{\widehat{a}_n} r_{n-1} \xrightarrow{\clubsuit} & s_{n-1} \xrightarrow{\clubsuit} t_{n-1} \xrightarrow{\widehat{a}_{n-1}} r_{n-2} \cdots \\
   & \cdots r_{1} \xrightarrow{\clubsuit} s_{1} \xrightarrow{\clubsuit} t_{1} \xrightarrow{\widehat{a}_1} x
\end{align*}
and the expansion $x \xrightarrow{w_{a}} z_2$ is of the form:
\begin{align*}
     \qquad \quad x \xrightarrow{a_1} t'_{1} \xrightarrow{\clubsuit} s'_{1} \xrightarrow{\clubsuit} r'_{1} \xrightarrow{a_2} & t'_{2} \xrightarrow{\clubsuit} s'_{2}  \xrightarrow{\clubsuit}  r'_{2} \cdots \\
 &   \cdots r'_{n-1} \xrightarrow{a_n}  t'_{n} \xrightarrow{\clubsuit} s'_{n} \xrightarrow{\clubsuit} z_2
\end{align*}
   Then the relation $=_{\widetilde{F}}$ produced by the equality atoms $J$ satisfies:
 \begin{enumerate}
\item  $t_{1}\neq_{\widetilde{F}} t'_{1}, \cdots, t_{n}\neq_{\widetilde{F}} t'_{n}$
 \item  $s_{1} =_{\widetilde{F}} s'_{1}, \cdots, s_{n} =_{\widetilde{F}} s'_{n}$
 \item $r_{1} =_{\widetilde{F}} r'_{1}, \cdots, r_{n} =_{\widetilde{F}} r'_{n}$
\end{enumerate}
where $r_n := y_2$ and $r'_n := z_2$.

\end{enumerate}

     \medskip
     
The key property of well-formedness is that it can be characterized in terms of the non-existence of a finite number of simple cycles and simple paths having certain labels. In order to do this, we need to define some finite languages. 
    Recall $\widetilde{U}_i$ is obtained from $U_i$ by removing the last symbol $\blacksquare'$ and 
    $\widetilde{V}_i$ is obtained from $V_i$ by removing the first symbol $\widehat{\blacksquare}'$. 
    We define $N$ to be the maximum length of the words $U_i$. 
      We denote by $e^{i,j}$, for $i\leq j$, the regular expression $(e^i + e^{i+1}+\cdots+e^j)$. We have:
    \begin{align*}
 & K_{I \,\widehat{I}} = \Ic\, \widehat{\Ic} + \#_\infty\, \widehat{\Ic} + \Ic\, \widehat{\#}_\infty \\
 & M_{I \,\widehat{I}} = \sum_{i\neq j} I_i \widehat{I}_j + \widehat{\Ic} \, \# + \widehat{\#}\, \Ic + \#\, \Ic \, \widehat{\Ic} \, \widehat{\#} 
  + \square\, \widehat{\square} + \#_\infty \, \widehat{\Ic} + \Ic\, \widehat{\#}_\infty \\
 & K_{I a} = \Ic\, \Sigma + \#_{\infty}\, \Sigma + \Ic\, \$_{\infty} \\
 & M_{I  a} = (\Sigma + \$ + \$' + \blacksquare)\, \Ic + (\Sigma + \$ + \blacksquare)^{1,N}\, \#
 + \sum_{i}\sum_{j\neq i} I_i \widetilde{U}_j  \, + \\ 
 &  \qquad + \# \, \Ic \, (\widetilde{U}_1 + \cdots + \widetilde{U}_\ell) + \square\, \blacksquare' + \#_{\infty}\, \Sigma + \Ic\, \$_{\infty} \\
 & K_{\widehat{a} \, \widehat{I}} = \widehat{\Sigma}\, \widehat{\Ic} + \widehat{\$}_\infty\, \widehat{\Ic} + \widehat{\Sigma}\, \widehat{\#}_\infty \\
 & M_{\widehat{a}\,  \widehat{I}} = \widehat{\Ic} \, (\widehat{\Sigma} + \widehat{\$} + \widehat{\$}' + \widehat{\blacksquare}) 
 + \widehat{\#}\, \widehat{\Sigma} + \widehat{\Ic} \, \widehat{\#} \,(\widehat{\Sigma} + \widehat{\$} + \widehat{\blacksquare})
 +  \sum_{i}\sum_{j\neq i} \widetilde{V}_j \widehat{I}_i \, + \\
 & \qquad + (\widetilde{V}_1 + \cdots + \widetilde{V}_\ell)\,  \widehat{\Ic} \, \widehat{\#} \,
 + \widehat{\blacksquare}' \widehat{\square} + \widehat{\$}_\infty\, \widehat{\Ic} + \widehat{\Sigma}\, \widehat{\#}_\infty \\
  & K_{\widehat{a} a} = \widehat{\Sigma}\, \Sigma + \widehat{\$}_\infty\, \Sigma + \widehat{\Sigma}\, \$_{\infty}\\
 & M_{\widehat{a} a} =  \sum_{a\neq b} \widehat{a} b + \Sigma\, (\widehat{\$} + \widehat{\$}') + (\$ + \$') \,\widehat{\Sigma}
  + (\widehat{\$} + \widehat{\$}') \,\widehat{\Sigma}\, \Sigma \,(\$ + \$') + \\
  & \qquad + (\widehat{\blacksquare} + \widehat{\blacksquare}')(\blacksquare + \blacksquare')
  + \widehat{\$}_\infty\, \Sigma + \widehat{\Sigma}\, \$_{\infty}
   \end{align*}

  \begin{claim}
  \label{claim:well-formed}
  Let $F$ be an $\ani$-expansion of $Q_1$. Then:
  \begin{enumerate}
  \item $F$ satisfies the $I$-$\widehat{I}$-condition iff $F$ does not contain 
a simple cycle with label in $K_{I \widehat{I}}$ nor a simple path with label in $M_{I \widehat{I}}$.
  \item $F$ satisfies the $I$-$a$-condition iff $F$ does not contain 
a simple cycle with label in $K_{I a}$ nor a simple path with label in $M_{I a}$.
  \item $F$ satisfies the $\widehat{a}$-$\widehat{I}$-condition iff $F$ does not contain 
a simple cycle with label in $K_{\widehat{a} \, \widehat{I}}$ nor a simple path with label in $M_{\widehat{a} \, \widehat{I}}$.
\item $F$ satisfies the $\widehat{a}$-$a$-condition iff $F$ does not contain 
a simple cycle with label in $K_{\widehat{a} a}$ nor a simple path with label in $M_{\widehat{a} a}$.
  \end{enumerate}
  \end{claim}
     
     \begin{proof}
 Suppose $F=\widetilde{F}^{\collapse}$ for $\widetilde{F} = E\land J$, where $E\in \Exp{}(Q_1)$ and $J$ are equality atoms. 
 We start with item (1). 
 Let $y_1\xrightarrow{w_{I}} x$ and $x\xrightarrow{\widehat{w}_{I}} z_1$ be the 
expansions associated to the atoms  $y_1\xrightarrow{L_{I}} x$  and 
 $x\xrightarrow{\widehat{L}_{I}} z_1$ in the expansion $E$. 
 Suppose that  $w_I=\square\, \#\, I_{i_k}\, \cdots \,\square \,\# \,I_{i_1}$ and $\widehat{w}_I= \widehat{I}_{j_1}\, \widehat{\#} \, \widehat{\square} \, \cdots \, \widehat{I}_{j_p} \, \widehat{\#}\, \widehat{\square}$,
 for indices    $i_1,\dots, i_k,$ $j_1,\dots,j_p$ $\in\{1,\dots, \ell\}$.
 Assume also that $y_1\xrightarrow{w_{I}} x$ is of the form:
  \begin{align*} 
    y_1 \xrightarrow{\square} s_{k} \xrightarrow{\#} t_{k} \xrightarrow{I_{i_k}} r_{{k-1}}  \xrightarrow{\square} & s_{{k-1}}  \xrightarrow{\#} t_{{k-1}} \xrightarrow{I_{i_{k-1}}} r_{{k-2}} \cdots \\
    & \cdots r_{1}\xrightarrow{\square} s_{1} \xrightarrow{\#} t_{1} \xrightarrow{I_{i_1}} x
    \end{align*}
     and $x \xrightarrow{\widehat{w}_{I}} z_1$ is of the form:
    \begin{align*}
  x \xrightarrow{\widehat{I}_{j_1}} t'_{1} \xrightarrow{\widehat{\#}} s'_{1} \xrightarrow{\widehat{\square}} r'_{{1}} \xrightarrow{\widehat{I}_{i_2}} & t'_{{2}} \xrightarrow{\widehat{\#}} s'_{{2}}  \xrightarrow{\widehat{\square}}  r'_{{2}} \cdots \\
    & \cdots r'_{{p-1}} \xrightarrow{\widehat{I}_{j_p}}  t'_{p} \xrightarrow{\widehat{\#}} s'_{p} \xrightarrow{\widehat{\square}} z_1
    \end{align*}
    
    We consider first the backward direction of item (1). 
    We have $t_1 \neq_{\widetilde{F}} t_1'$, otherwise there would be a simple cycle 
    $t_1 \xrightarrow{I_{i_1}} x \xrightarrow{\widehat{I}_{j_1}} t_1'$ 
 with label in  $\Ic\, \widehat{\Ic} \subseteq K_{I \widehat{I}}$. 
 This implies that $i_1=j_1$. If this is not the case, we would have a simple path $t_1 \xrightarrow{I_{i_1}} x \xrightarrow{\widehat{I}_{j_1}} t_1'$ with a label in $\sum_{i\neq j} I_i \widehat{I}_j \subseteq M_{I \widehat{I}}$. 
    We claim that $s_1 =_{\widetilde{F}} s_1'$. Note first that 
    $t_1 \neq_{\widetilde{F}} s_1'$ and $t_1' \neq_{\widetilde{F}} s_1$. 
    Indeed, if $t_1 =_{\widetilde{F}} s_1'$, then 
    there would be a simple path $t_1' \xrightarrow{\widehat{\#}} s_1' \xrightarrow{I_{i_1}} x$
    with label in $\widehat{\#}\, \Ic\subseteq M_{I \widehat{I}}$. 
    If $t_1' =_{\widetilde{F}} s_1$, then 
    there would be a simple path $x \xrightarrow{\widehat{I}_{j_1}} t_1' \xrightarrow{\#} t_1$
    with label in $\widehat{\Ic}\, \# \subseteq M_{I \widehat{I}}$. 
    It follows that $s_1 =_{\widetilde{F}} s_1'$ as otherwise the path
     $s_{1} \xrightarrow{\#} t_{1} \xrightarrow{I_{i_1}} x \xrightarrow{\widehat{I}_{j_1}} t'_{1} 
     \xrightarrow{\widehat{\#}} s'_{1} $
     would be a simple path with a label in  
   $ \#\, \Ic \, \widehat{\Ic} \, \widehat{\#} \subseteq M_{I \widehat{I}}$. 
   Finally, we have that $r_1 =_{\widetilde{F}} r_1'$. If this is not true, 
   then $r_{1}\xrightarrow{\square} s_{1} \xrightarrow{\widehat{\square}} r'_{{1}}$ 
   would be a simple path with a label in $\square\, \widehat{\square} \subseteq M_{I \widehat{I}}$. 
   
   We can iterate this argument, replacing in each step the ``middle'' variable $x$ 
   by the corresponding new ``middle'' variable $r_i$ (see Figure~\ref{fig:app-well-formed}). 
   We obtain the following ($\alpha = \min\{k,p\}$):
   \begin{itemize}
   \item $i_1=j_1$, $i_2=j_2$, $\dots$, $i_\alpha = j_\alpha$
   \item $t_{1}\neq_{\widetilde{F}} t'_{1}, \cdots, t_{\alpha}\neq_{\widetilde{F}} t'_{\alpha}$
   \item $s_{1} =_{\widetilde{F}} s'_{1}, \cdots, s_{\alpha} =_{\widetilde{F}} s'_{\alpha}$
   \item $r_{1} =_{\widetilde{F}} r'_{1}, \cdots, r_{\alpha} =_{\widetilde{F}} r'_{\alpha}$
   \end{itemize}
where $r_{k}:= y_1$ and $r'_p:=z_1$. 
Note that to conclude the $I$-$\widehat{I}$-condition, it suffices to show that $\alpha = k = p$. 
Towards a contradiction, suppose first that $k < p$. We know that $y_1 =_{\widetilde{F}} r'_k$, 
and $r'_k\neq z_1$. In particular, we have at least the following atoms:
$$ y_1' \xrightarrow{\#_\infty} y_1 \qquad r'_k \xrightarrow{\widehat{I}_j} t'_{k+1}$$
We have two cases. If $y_1' =_{\widetilde{F}} t'_{k+1}$ then we have a 
simple cycle $y_1' \xrightarrow{\#_\infty} y_1\xrightarrow{\widehat{I}_j} t'_{k+1}$ 
with a label in $\#_\infty \, \widehat{\Ic} \subseteq K_{I \widehat{I}}$. 
On the other hand,  if $y_1' \neq_{\widetilde{F}} t'_{k+1}$ then we have a 
simple path $y_1' \xrightarrow{\#_\infty} y_1\xrightarrow{\widehat{I}_j} t'_{k+1}$ 
with a label in $\#_\infty \, \widehat{\Ic} \subseteq M_{I \widehat{I}}$. In either case, 
we obtain a contradiction.
Suppose now that $k > p$. We know that $r_p =_{\widetilde{F}} z_1$, 
and $r_p\neq y_1$. We have at least the following atoms:
$$ t_{p+1} \xrightarrow{I_i} r_p \qquad z_1 \xrightarrow{\widehat{\#}_\infty} z_1''$$
Again we have two cases. If $t_{p+1} =_{\widetilde{F}} z_1''$ then we have a 
simple cycle $t_{p+1} \xrightarrow{I_i} r_p \xrightarrow{\widehat{\#}_\infty} z_1''$ 
with a label in $\Ic\, \widehat{\#}_\infty \subseteq K_{I \widehat{I}}$. 
On the other hand,  if $t_{p+1} \neq_{\widetilde{F}} z_1''$ then we have a 
simple path $t_{p+1} \xrightarrow{I_i} r_p \xrightarrow{\widehat{\#}_\infty} z_1''$ 
with a label in $\Ic\, \widehat{\#}_\infty \subseteq M_{I \widehat{I}}$. 
We obtain a contradiction in either case. We conclude that $k=p$ and hence the $I$-$\widehat{I}$-condition holds. 

The forward direction of item (1) follows directly from the definition of the $I$-$\widehat{I}$-condition 
and inspection of the languages $K_{I \widehat{I}}$ and $M_{I \widehat{I}}$. 

Now we turn to item (2). 
Let $y_1\xrightarrow{w_{I}} x$ and $x\xrightarrow{w_{a}} z_2$ be the 
 expansions associated to the atoms  $y_1\xrightarrow{L_{I}} x$  and 
$x \xrightarrow{L_{a}} z_2$ in the expansion $E$. 
Suppose that $w_I=\square\, \#\, I_{i_k}\, \cdots \,\square \,\# \,I_{i_1}$ and $w_a=U_{j_1}\cdots U_{j_p}$
 for indices  $i_1,\dots, i_k,$ $j_1,\dots,j_p$ $\in\{1,\dots, \ell\}$.
Assume that the expansion $y_1\xrightarrow{w_{I}} x$ is of the form:
    \begin{align*}
 y_1 \xrightarrow{\square} s_{k} \xrightarrow{\#} t_{k} \xrightarrow{I_{i_k}} r_{{k-1}}  \xrightarrow{\square} & s_{{k-1}}  \xrightarrow{\#} t_{{k-1}} \xrightarrow{I_{i_{k-1}}} r_{{k-2}} \cdots \\
    & \cdots r_{1}\xrightarrow{\square} s_{1} \xrightarrow{\#} t_{1} \xrightarrow{I_{i_1}} x
    \end{align*}
    and the expansion $x \xrightarrow{w_{a}} z_2$ is of the form:
    \begin{align*}
 x \xrightarrow{\widetilde{U}_{j_1}} s'_{1} \xrightarrow{\blacksquare'} r'_{{1}} \xrightarrow{\widetilde{U}_{i_2}} s'_{{2}}  \xrightarrow{\blacksquare'}  r'_{{2}} \cdots r'_{{p-1}} \xrightarrow{\widetilde{U}_{j_p}} s'_{p} \xrightarrow{\blacksquare'} z_2
    \end{align*}
   where $\widetilde{U}_j$ is the word obtained from $U_j$ by removing the last symbol $\blacksquare'$.
   
       We consider first the backward direction of item (2). 
       Suppose the expansion $x \xrightarrow{\widetilde{U}_{j_1}} s'_{1}$ has the form:
  \begin{align*}  
        x \xrightarrow{\clubsuit} o_1 \xrightarrow{\clubsuit} o_2 \cdots o_m \xrightarrow{\clubsuit} s'_{1}  
  \end{align*}
  where $\clubsuit$ is a placeholder representing some symbol. 
  We claim that $t_1 \neq_{\widetilde{F}} t$, for all $t\in \{o_1,\dots, o_m, o_{m+1}\}$, where $o_{m+1}:=s'_1$. 
  Note first that $t_1 \neq_{\widetilde{F}} o_1$, otherwise there would be a 
  simple cycle $t_1 \xrightarrow{I_{i_1}} x \xrightarrow{\clubsuit} o_1$, where $\clubsuit\in \Sigma$. 
  In particular, the label would belong to $\Ic\, \Sigma\subseteq K_{I a}$; a contradiction. 
  Now we argue by induction. Suppose $t_1 \neq_{\widetilde{F}} o_h$, for some $h\in\{1,\dots, m\}$. 
  By contradiction, assume $t_1 =_{\widetilde{F}} o_{h+1}$. 
  We have a simple path $o_h \xrightarrow{\clubsuit} o_{h+1} \xrightarrow{I_{i_1}} x$, 
  where $\clubsuit \in \Sigma \cup \{\$, \blacksquare, \$'\}$ (note that $\clubsuit$ cannot be $\blacksquare'$). 
  Then the label belongs to $(\Sigma + \$ + \$' + \blacksquare)\, \Ic \subseteq M_{I a}$; a contradiction. 
  
  We now claim that $s_1 \neq_{\widetilde{F}} t$, for all $t\in\{o_1,\dots, o_m\}$. 
  By contradiction, suppose that $s_1 =_{\widetilde{F}} t$, for some $t\in\{o_1,\dots, o_m\}$. 
  Then there is a simple path $x \xrightarrow{U} t \xrightarrow{\#} t_1$, where $U\in (\Sigma + \$ + \blacksquare)^{1,N}$
  (note how we use the fact that $t_1\neq_{\widetilde{F}} t$, for all $t\in\{o_1,\dots, o_m\}$; 
  otherwise the path would not be necessarily simple). The label of the simple path 
  belongs to $(\Sigma + \$ + \blacksquare)^{1,N}\, \# \subseteq M_{I a}$; a contradiction. 
  
We have that $i_1=j_1$. If this is not the case, then 
we would have the simple path $t_1 \xrightarrow{I_{i_1}} x \xrightarrow{\widetilde{U}_{j_1}} s'_{1}$
with a label in $\sum_{i}\sum_{j\neq i} I_i \widetilde{U}_j  \subseteq M_{I a}$. 
Moreover, we have $s_1 =_{\widetilde{F}} s_1'$, otherwise 
we would have the simple path $s_1 \xrightarrow{\#} t_1 \xrightarrow{I_{i_1}} x \xrightarrow{\widetilde{U}_{j_1}} s'_{1}$
with a label in $\# \, \Ic \, (\widetilde{U}_1 + \cdots + \widetilde{U}_\ell)  \subseteq M_{I a}$. 
Finally, we have $r_1 =_{\widetilde{F}} r_1'$. If this is not the case, then 
we have a simple path $r_1 \xrightarrow{\square} s_1 \xrightarrow{\blacksquare'} r_1'$ 
with a label in $\square\, \blacksquare' \subseteq M_{I a}$. 

As in the case of item (1),  we can iterate this argument, replacing in each step the ``middle'' variable $x$ 
   by the corresponding new ``middle'' variable $r_i$. 
   We obtain the following ($\alpha = \min\{k,p\}$):
   \begin{itemize}
   \item $i_1=j_1$, $i_2=j_2$, $\dots$, $i_\alpha = j_\alpha$
   \item For every $j\in\{1,\dots, \alpha\}$, 
we have $t_{j} \neq_{\widetilde{F}} t$, for every internal variable $t$ of the expansion  $r'_{{j-1}} \xrightarrow{\widetilde{U}_{i_j}} s'_{{j}}$ (here $r'_{0}:= x$)
\item $s_{1} =_{\widetilde{F}} s'_{1}, \cdots, s_{\alpha} =_{\widetilde{F}} s'_{\alpha}$
\item $r_{1} =_{\widetilde{F}} r'_{1}, \cdots, r_{{\alpha}} =_{\widetilde{F}} r'_{{\alpha}}$
   \end{itemize}
   where $r_{k}:= y_1$ and $r'_p:=z_2$. 
   By using the same arguments as in the case of item (1) we obtain that $k=p$, 
   and hence the $I$-$a$-condition holds. 
   
   The forward direction of item (2) follows directly from the definition of the $I$-$a$-condition
and inspection of the languages $K_{I a}$ and $M_{I a}$. 

The cases of item (3) and (4) are analogous to cases (1) and (2).
\end{proof}

\medskip

Let $Q_2^{\circlearrowright}$ and $Q_2^{\rightarrow}$ be the following Boolean CRPQs in $\CRPQfin$:

\begin{align*}
Q_2^{\circlearrowright} & = x \xrightarrow{K^{\circlearrowright}} x \qquad Q_2^{\rightarrow} = y \xrightarrow{M^{\rightarrow}} z
\end{align*}
where $K^{\circlearrowright} := K_{I \widehat{I}} \, + K_{I a} \,+ K_{\widehat{a} \widehat{I}} 
\, + K_{\widehat{a} a}$ and $M^{\rightarrow} := M_{I \widehat{I}} \, + M_{I a} \, + M_{\widehat{a} \widehat{I}} 
\, + M_{\widehat{a} a}$. 

From Claim~\ref{claim:well-formed}, we obtain the reduction for the case when 
the right-hand side query is the \emph{union} of the CRPQs $Q_2^{\circlearrowright}$ and $Q_2^{\rightarrow}$, 
which we denote by $Q_2^{\circlearrowright} \lor Q_2^{\rightarrow}$. 

\begin{claim}
\label{claim:well-formed2}
Let $F$ be an $\ani$-expansion of $Q_1$. Then $F$ is well-formed if and only if 
$Q_2^{\circlearrowright} \lor Q_2^{\rightarrow}(F)^{\ani} = \emptyset$. Moreover, 
there is a solution to the PCP instance if and only if $Q_1\not\subseteq_{\ani}Q_2^{\circlearrowright} \lor Q_2^{\rightarrow}$. 
\end{claim}

From Claim~\ref{claim:well-formed2}, we obtain the undecidability of containment under atom-injective semantics
of a CRPQ in a union of two CRPQs from $\CRPQfin$. We conclude our proof explaining how to simulate the union $Q_2^{\circlearrowright} \lor Q_2^{\rightarrow}$ with a single query $Q_2\in \CRPQfin$ as in Figure~\ref{fig:app-query-q1-q2}. 

We define the following languages:
\begin{align*}
K_{dummy} & = (\square + \widehat{\blacksquare} + \widehat{\blacksquare}')(\widehat{\square} + \blacksquare +\blacksquare')\\ 
M_{dummy} & = \widehat{\#} + \$ + \$' \\
L & = \epsilon + \Ic + \#\, \Ic + \widehat{\#}\, \Ic + \square\, \#\, \Ic + \#_\infty +  (\Sigma + \$ + \$' + \blacksquare)\, \Ic + \\
& \widehat{\Sigma} + \widehat{\#}\, \widehat{\Sigma} + (\widetilde{V}_1 + \cdots + \widetilde{V}_\ell) + \widehat{\blacksquare}'\, (\widetilde{V}_1 + \cdots + \widetilde{V}_\ell) + \widehat{\$}_\infty + \\
& + (\$ + \$')\, \widehat{\Sigma} + (\widehat{\$} + \widehat{\$}')\, \widehat{\Sigma}
  + (\widehat{\blacksquare} + \widehat{\blacksquare}')  (\widehat{\$} + \widehat{\$}')\, \widehat{\Sigma}
 \end{align*}
 Let $Q_2$ be the CRPQ defined as:
 $$Q_2  = x \xrightarrow{K} x \land y \xrightarrow{L} x \land y \xrightarrow{M} z$$
 where $K:= K^{\circlearrowright} + K_{dummy}$ and $M:= M^{\rightarrow} + M_{dummy}$. 
 We conclude with the following claim:
 \begin{claim}
 Let $F$ be an $\ani$-expansion of $Q_1$. Then $Q_2^{\circlearrowright} \lor Q_2^{\rightarrow}(F)^{\ani} \neq \emptyset$
 if and only if $Q_2(F)^{\ani} \neq \emptyset$. 
 \end{claim}
 \begin{proof}
 For the forward direction, 
 suppose $Q_2^{\circlearrowright} \lor Q_2^{\rightarrow}(F)^{\ani} \neq \emptyset$. 
 Assume first that $Q_2^{\circlearrowright}(F)^{\ani} \neq \emptyset$. 
 We consider two cases for the label $w$ of the simple cycle mapping to $F$ and provide 
an expansion of $Q_2$ that maps to $F$:
 \begin{itemize}
 \item $w\in K_{I \,\widehat{I}} \cup K_{\widehat{a}\, \widehat{I}}$ : 
 take expansion 
 $x \xrightarrow{w} x \land y \xrightarrow{\epsilon} x \land y \xrightarrow{\widehat{\#}} z$. 
 \item $w\in K_{I a} \cup K_{\widehat{a} a}$ : take either expansion 
 $x \xrightarrow{w} x \land y \xrightarrow{\epsilon} x \land y \xrightarrow{\$} z$ or  
 $x \xrightarrow{w} x \land y \xrightarrow{\epsilon} x \land y \xrightarrow{\$'} z$. 
 \end{itemize}
 Note above that $\epsilon\in L$ and $\widehat{\#}, \$, \$' \in M_{dummy}\subseteq M$.
 
 Suppose now that $Q_2^{\rightarrow}(F)^{\ani} \neq \emptyset$. Again, 
 We consider all the possible cases for the label $w$ of the simple path mapping to $F$ and provide 
an expansion of $Q_2$ that maps to $F$. We start with the case $w\in M_{I \,\widehat{I}}$:
 \begin{itemize}
 \item $w=I_p \widehat{I_q} \in \sum_{i\neq j} I_i \widehat{I}_j $ : 
 take expansion 
 $x \xrightarrow{\square\widehat{\square}} x \land y \xrightarrow{I_p} x \land y \xrightarrow{w} z$. 
  \item $w\in \widehat{\Ic} \, \# $ :
 take expansion 
 $x \xrightarrow{\square\widehat{\square}} x \land y \xrightarrow{\epsilon} x \land y \xrightarrow{w} z$. 
  \item $w = \widehat{\#} I_p \in  \widehat{\#}\, \Ic $ :
 take expansion 
 $x \xrightarrow{\square\widehat{\square}} x \land y \xrightarrow{\widehat{\#} I_p} x \land y \xrightarrow{w} z$. 
  \item $w = \# I_p \widehat{I_q} \widehat{\#} \in \#\, \Ic \, \widehat{\Ic} \, \widehat{\#}$ :
 take expansion 
 $x \xrightarrow{\square\widehat{\square}} x \land y \xrightarrow{\# I_p} x \land y \xrightarrow{w} z$. 
   \item $w\in \square\, \widehat{\square} $ :
 take expansion 
 $x \xrightarrow{\square\widehat{\square}} x \land y \xrightarrow{\square \# I_i} x \land y \xrightarrow{w} z$, for a suitable $I_i \in \Ic$. 
    \item $w\in \#_\infty \, \widehat{\Ic} $ :
 take expansion 
 $x \xrightarrow{\square\widehat{\square}} x \land y \xrightarrow{\#_\infty} x \land y \xrightarrow{w} z$. 
   \item $w = I_p \widehat{\#}_\infty \in \Ic\, \widehat{\#}_\infty  $ :
 take expansion 
 $x \xrightarrow{\square\widehat{\square}} x \land y \xrightarrow{I_p} x \land y \xrightarrow{w} z$. 
 \end{itemize}
 Note that $\square\widehat{\square} \in K_{dummy}\subseteq K$ and 
 $I_p, \epsilon, \widehat{\#} I_p, \# I_p, \square \# I_i, \#_\infty\in L$. For the case 
 $w\in M_{I a}$ we have the following:
  \begin{itemize}
 \item $w=\clubsuit I_p \in (\Sigma + \$ + \$' + \blacksquare)\, \Ic$ : 
 take expansion 
 $x \xrightarrow{\square \blacksquare'} x \land y \xrightarrow{\clubsuit I_p} x \land y \xrightarrow{w} z$. 
  \item $w\in (\Sigma + \$ + \blacksquare)^{1,N}\, \#$ :
 take expansion 
 $x \xrightarrow{\square \blacksquare'} x \land y \xrightarrow{\epsilon} x \land y \xrightarrow{w} z$. 
  \item $w = I_p \widetilde{U}_q \in  \sum_{i}\sum_{j\neq i} I_i \widetilde{U}_j $ :
 take expansion 
 $x \xrightarrow{\square \blacksquare'} x \land y \xrightarrow{I_p} x \land y \xrightarrow{w} z$. 
  \item $w = \# I_p \widetilde{U}_q \in \# \, \Ic \, (\widetilde{U}_1 + \cdots + \widetilde{U}_\ell)$ :
 take expansion 
 $x \xrightarrow{\square \blacksquare'} x \land y \xrightarrow{\# I_p} x \land y \xrightarrow{w} z$. 
   \item $w\in \square\, \blacksquare'$ :
 take expansion 
 $x \xrightarrow{\square \blacksquare'} x \land y \xrightarrow{\square \# I_i} x \land y \xrightarrow{w} z$, for a suitable $I_i \in \Ic$. 
    \item $w\in \#_{\infty}\, \Sigma $ :
 take expansion 
 $x \xrightarrow{\square \blacksquare'} x \land y \xrightarrow{\#_\infty} x \land y \xrightarrow{w} z$. 
   \item $w = I_p \$_{\infty} \in \Ic\, \$_{\infty} $ :
 take expansion 
 $x \xrightarrow{\square \blacksquare'} x \land y \xrightarrow{I_p} x \land y \xrightarrow{w} z$. 
 \end{itemize}
Observe that $\square \blacksquare' \in K_{dummy}\subseteq K$ and 
 $\clubsuit I_p \in L$, for $\clubsuit\in \Sigma + \$ + \$' + \blacksquare$, and $\epsilon, I_p, \# I_p, \square \# I_i, \#_\infty \in L$. For the case 
 $w\in M_{\widehat{a} \, \widehat{I}}$ we have:
   \begin{itemize}
 \item $w \in \widehat{\Ic} \, (\widehat{\Sigma} + \widehat{\$} + \widehat{\$}' + \widehat{\blacksquare}) $ : 
 take expansion 
 $x \xrightarrow{\widehat{\blacksquare}' \widehat{\square}} x \land y \xrightarrow{\epsilon} x \land y \xrightarrow{w} z$. 
  \item $w= \widehat{\#} \widehat{a} \in \widehat{\#}\, \widehat{\Sigma}$ :
 take expansion 
 $x \xrightarrow{\widehat{\blacksquare}' \widehat{\square}} x \land y \xrightarrow{\widehat{\#} \widehat{a}} x \land y \xrightarrow{w} z$. 
   \item $w \in \widehat{\Ic} \, \widehat{\#} \,(\widehat{\Sigma} + \widehat{\$} + \widehat{\blacksquare})$ :
 take expansion 
 $x \xrightarrow{\widehat{\blacksquare}' \widehat{\square}} x \land y \xrightarrow{\epsilon} x \land y \xrightarrow{w} z$. 
  \item $w = \widetilde{V}_p \widehat{I}_q \in  \sum_{i}\sum_{j\neq i} \widetilde{V}_j \widehat{I}_i $ :
 take expansion 
 $x \xrightarrow{\widehat{\blacksquare}' \widehat{\square}} x \land y \xrightarrow{\widetilde{V}_p} x \land y \xrightarrow{w} z$. 
  \item $w = \widetilde{V}_p \widehat{I}_q \widehat{\#} \in (\widetilde{V}_1 + \cdots + \widetilde{V}_\ell)\,  \widehat{\Ic} \, \widehat{\#} $ :
 take expansion 
 $x \xrightarrow{\widehat{\blacksquare}' \widehat{\square}} x \land y \xrightarrow{\widetilde{V}_p} x \land y \xrightarrow{w} z$. 
   \item $w\in \widehat{\blacksquare}' \widehat{\square} $ :
 take expansion 
 $x \xrightarrow{\widehat{\blacksquare}' \widehat{\square}} x \land y \xrightarrow{\widehat{\blacksquare}' \widetilde{V}_i} x \land y \xrightarrow{w} z$, for a suitable $\widetilde{V}_i \in (\widetilde{V}_1 + \cdots + \widetilde{V}_\ell)$. 
    \item $w\in \widehat{\$}_\infty\, \widehat{\Ic} $ :
 take expansion 
 $x \xrightarrow{\widehat{\blacksquare}' \widehat{\square}} x \land y \xrightarrow{\widehat{\$}_\infty} x \land y \xrightarrow{w} z$. 
   \item $w = \widehat{a}\, \widehat{\#}_\infty \in \widehat{\Sigma}\, \widehat{\#}_\infty $ :
 take expansion 
 $x \xrightarrow{\widehat{\blacksquare}' \widehat{\square}} x \land y \xrightarrow{\widehat{a}} x \land y \xrightarrow{w} z$. 
 \end{itemize}
 Note that $\widehat{\blacksquare}' \widehat{\square} \in K_{dummy}\subseteq K$ and  
 $\epsilon, \widehat{\#} \widehat{a}, \widetilde{V}_p, \widehat{\blacksquare}' \widetilde{V}_i, \widehat{\$}_\infty, \widehat{a} \in L$. For the case 
 $w\in M_{\widehat{a} a}$ we have ($\clubsuit\in \{\widehat{\blacksquare}, \widehat{\blacksquare}'\}$ and $\spadesuit\in \{\blacksquare, \blacksquare'\}$ 
 are suitable symbols in each case):
   \begin{itemize}
 \item $w= \widehat{p} q \in \sum_{a\neq b} \widehat{a} b$ : 
 take expansion 
 $x \xrightarrow{\clubsuit \spadesuit} x \land y \xrightarrow{\widehat{p}} x \land y \xrightarrow{w} z$. 
  \item $w\in \Sigma\, (\widehat{\$} + \widehat{\$}')$ :
 take expansion 
 $x \xrightarrow{\clubsuit \spadesuit} x \land y \xrightarrow{\epsilon} x \land y \xrightarrow{w} z$. 
  \item $w = \diamond \widehat{p} \in (\$ + \$') \,\widehat{\Sigma} $ :
 take expansion 
 $x \xrightarrow{\clubsuit \spadesuit} x \land y \xrightarrow{\diamond \widehat{p}} x \land y \xrightarrow{w} z$. 
  \item $w = \widehat{\diamond} \widehat{p} q \diamond \in (\widehat{\$} + \widehat{\$}') \,\widehat{\Sigma}\, \Sigma \,(\$ + \$')$ :
 take expansion 
 $x \xrightarrow{\clubsuit \spadesuit} x \land y \xrightarrow{\widehat{\diamond}\widehat{p}} x \land y \xrightarrow{w} z$. 
   \item $w = \widehat{\bigstar} \bigstar \in (\widehat{\blacksquare} + \widehat{\blacksquare}')(\blacksquare + \blacksquare')$ :
 take expansion 
 $x \xrightarrow{\clubsuit \spadesuit} x \land y \xrightarrow{\widehat{\bigstar} \widehat{\diamond} \widehat{a}} x \land y \xrightarrow{w} z$, for suitable $\widehat{\diamond}\in (\widehat{\$} + \widehat{\$}')$ and $\widehat{a}\in\widehat{\Sigma}$. 
    \item $w\in \widehat{\$}_\infty\, \Sigma$ :
 take expansion 
 $x \xrightarrow{\clubsuit \spadesuit} x \land y \xrightarrow{\widehat{\$}_\infty} x \land y \xrightarrow{w} z$. 
   \item $w = \widehat{p} \$_{\infty} \in \widehat{\Sigma}\, \$_{\infty}$ :
 take expansion 
 $x \xrightarrow{\clubsuit \spadesuit} x \land y \xrightarrow{\widehat{p}} x \land y \xrightarrow{w} z$. 
 \end{itemize}
  Note that $\clubsuit \spadesuit \in K_{dummy}\subseteq K$,  and  
 $\widehat{p}, \epsilon, \widehat{\$}_\infty \in L$, and $\diamond \widehat{p},  \widehat{\diamond}\widehat{p}, 
 \widehat{\bigstar} \widehat{\diamond} \widehat{a} \in L$, for 
 $\diamond\in (\$ + \$')$, $\widehat{\diamond}\in (\widehat{\$} + \widehat{\$}')$, 
 and $\widehat{\bigstar} \in (\widehat{\blacksquare} + \widehat{\blacksquare}')$.

For the backward direction, suppose that $Q_2(F)^{\ani} \neq \emptyset$. 
 If suffices to show that the expansion of $Q_2$ mapping to $F$ cannot use simultaneously words in $K_{dummy}$ 
 and $M_{dummy}$. It is possible to check that any mapping of an expansion $x \xrightarrow{w}  x  \land x \xleftarrow{w'} y \land y \xrightarrow{w''} z$, 
where $w\in K_{dummy}$, $w'\in L$ and $w''\in M_{dummy}$, maps $y$ to a variable $\bullet$ such that the labels of all outgoing edges of $\bullet$ belongs to the set 
$$\{\square, \Ic, \#, \widehat{\blacksquare}, \widehat{\blacksquare}', \widehat{\Sigma}, \widehat{\$}, \widehat{\$}', \widehat{\square}, \widehat{\Ic}, \blacksquare, \blacksquare', \Sigma, \#_\infty, 
 \widehat{\#}_\infty, \$_\infty, \widehat{\$}_\infty\}$$ 
However, this set of symbols is disjoint from $M_{dummy}$ and hence $z$ cannot be mapped to any variable. 
 \end{proof}

\section{Full proof of Theorem~\ref{thm:crpqfin-crpqfin-lower}}
\label{sec:app-cq-crpq-sf-atom-full}

	We show that even when all languages on the right-hand side are of the form $\set{ w }$ with $|w| \leq 2$ we have \pitwo-hardness for containment. For this, we show how to adapt the proof of \pitwo-hardness of \cite[Theorem~4.3]{FigueiraGKMNT20}, which shows \pitwo-hardness for \CRPQfin/\CQ containment for the standard semantics.\footnote{Actually, it shows hardness for the fragment where the left-hand side can only have regular expressions of the form $a_1 + \dotsb + a_n$. In some sense, we simulate disjunction with the choice of an atom-injective expansion for a CQ.}
	
	We use a reduction from
	$\forall \exists$-QBF. The main idea is to use
	sets $\{t,f\}$ in $Q_1$ to encode true or false. 
	
	More precisely,  let
	\[
	\Phi\quad =\quad \forall x_1, \ldots, x_n\; \exists y_1, \ldots, y_\ell\;
	\varphi(x_1, \ldots, x_n,y_1, \ldots, y_\ell)
	\]
	be an instance of
	$\forall \exists$-QBF such that $\varphi$ is quantifier free and in
	3-CNF. We construct boolean queries $Q_1$ and $Q_2$ such that
	$Q_1 \subseteq_\ani Q_2$ if and only if $\Phi$ is satisfiable.

	\begin{figure*}
		\includegraphics[width=.65\textwidth]{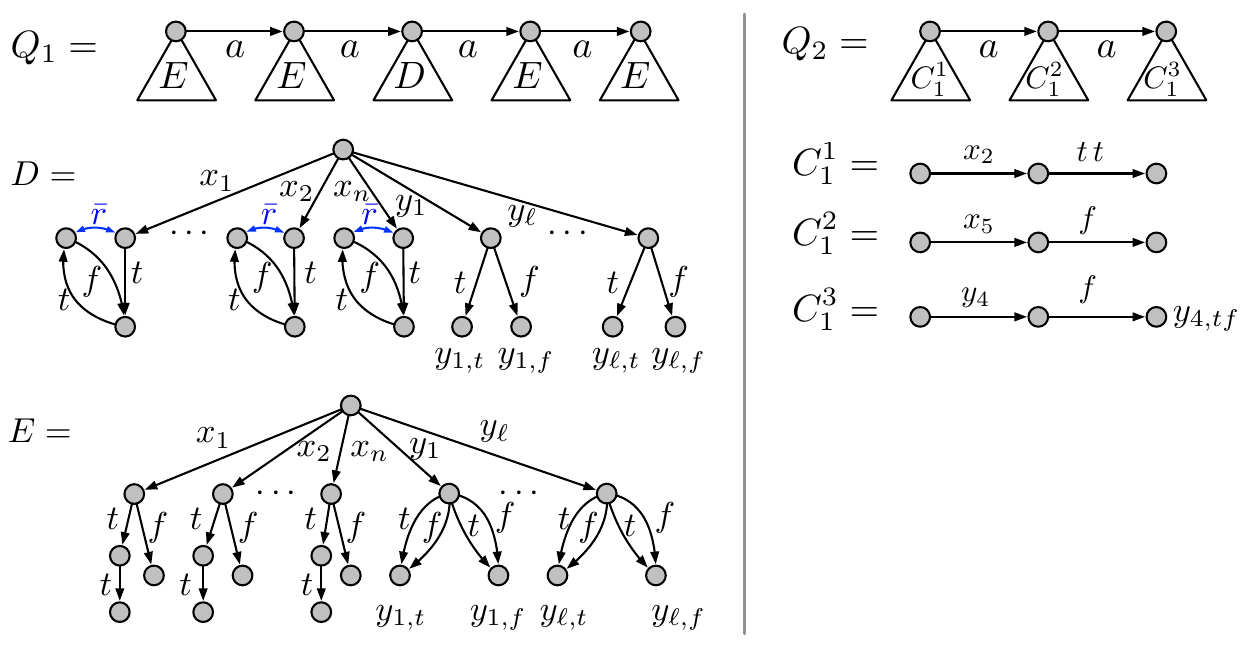}
		\caption{Query $Q_1$ used in Theorem~\ref{thm:crpqfin-crpqfin-lower} and the gadgets $D$, and $E$ used in its definition. For the $r$ edge relation (in blue) we depict the edges of its complement (\ie, the edges which are \emph{not} in relation $r$). Example of $Q_2$ for the formula $\phi = (x_2 \lor \lnot x_5 \lor \lnot y_4)$.}
		\label{fig:pi-p-3-hard-app}
	\end{figure*}

	The \textbf{query $Q_1$} is sketched in
	Figure~\ref{fig:pi-p-3-hard-app} and built as follows: The basis is an $a$-path of
	length 4. We add 4 gadgets $E$ to the outer
	nodes of the path and one gadget $D$ at the innermost.  The choice of 4 $E$ gadgets 
	surrounding the $D$ gadget will be made clear once we discuss $Q_2$. 
		Basically,
	the $E$-gadgets will accept everything while the $D$-gadget will
	ensure that the chosen literal evaluates to true. The gadgets are also depicted 
	in Figure~\ref{fig:pi-p-3-hard-app}. The gadgets are constructed as follows.

	\newcommand{\yit}{\ensuremath{y_{i,t}}\xspace}
	\newcommand{\yif}{\ensuremath{y_{i,f}}\xspace}
	\newcommand{\yitf}{\ensuremath{y_{i,tf}}\xspace}
	\newcommand{\yktf}{\ensuremath{y_{k,tf}}\xspace}

	The \textbf{gadget $D$} is constructed such that the root node has
	one outgoing edge for each variable in $\Phi$, that is, $n+\ell$
	many. Each edge is labeled differently, that is,
	$x_1, \ldots, x_n, y_1, \ldots, y_\ell$. After each $x_i$-edge we
	add a $t$-edges which
	leads to a different node.  From each of these nodes we have  cycle of length 2 reading $t \, f$.
	For each $i \in \{1,\ldots, \ell\}$ we do
	the following. We add a $t$-edge to a node we name \yit after the
	$y_i$-edge and an edge labeled $f$ that leads to a node we name
	\yif.  We named these nodes because we need those nodes also in the
	$E$-gadgets. Nodes with the same names across gadgets are actually the same node. 
	
	Each \textbf{gadget $E$} is constructed similar to the $D$
	gadget. The root node has one outgoing
	edge for each variable in $\Phi$, that is $n+\ell$ many. Each edge
	is labeled differently, that is
	$x_1, \ldots, x_n, y_1, \ldots, y_\ell$. After each $x_i$-edge we
	add path of length 2 reading $t\, t$- and an $f$-edge. Each of those edges leads to a different node.  After each $y_i$-edge we add a $t$-edge and an $f$-edge to both \yit and to \yif.
	
	We now explain the construction of $Q_2$. An example is given in Figure~\ref{fig:pi-p-3-hard-app}.
	For each clause $i$, \textbf{query $Q_2$} has a small DAG, 
	which might share nodes ($\yktf$) with the DAGs constructed for the other clauses.  
	For clause $i$, we construct
	$C_i^1$, with an  $a$-edge to the gadget $C_i^2$, and from there again a 
	$a$-edge to the gadget $C_i^3$.
	
	The gadget $C_i^j$ represents the $j$th literal in the $i$th
	clause. Since the QBF is in 3-CNF, we have $j \in \{1,2,3\}$.
	If the literal is the positive variable $x_k$,
	$C_i^j$ is a path labeled $x_k\, t\, t$. If it is the
	negative variable $\neg x_k$, $C_i^j$ is a path labeled
	$x_k f$. If the
	literal is the positive variable $y_k$, $C_i^j$ is a path
	labeled $y_k t$ and it ends in a node we call $\yktf$ and, if it
	is the negative variable $\neg y_k$, $C_i^j$ is a path
	labeled $y_k f$ and it ends in $\yktf$, too. 
	
	This completes the construction. We will now give some intuition.
	The gadget $D$ controls via the $\{t,f\}$ (simple) paths, which variables
	$x_i$ are set to true and which to false. We will consider it \emph{false} 
	whenever there is a $(x_i \, f)$-path, meaning that the two nodes non-related via $r$ are equal in the \ani-expansion of $Q_1$.
	Otherwise, $x_i$ is set to true. Observe that whenever $x_i$ is false, it is not possible to map in an injective way any path $v \xrightarrow{x_i} v' \xrightarrow{t \, t} v''$ coming from a clause encoded from $Q_2$. And vice-versa, whenever $x_i$ is true there is no way to map a $(x_i  \, f)$-path.
	Hence, depending on this, we can either map
	$C_i^j$ into it or not. 
	The $E$ gadgets are constructed such
	that every $C_i^j$ can be mapped into it. The choice of the \ani-expansion of $Q_2$ determines which 
	path should be mapped into $D$ and, therefore, which literal should be verified. 
	The structure of $Q_1$ where two $E$ gadgets each surround the $D$ gadget aids in embedding the clauses
	$C_i^1, C_i^2, C_i^3$ for each $i$ in \ani-expansion $E_1$. If the $i$th clause is $(x_2 \vee \neg y_1 \vee \neg x_3)$, we have the assignment of $f$ to $x_2$, $t$ to $x_3$, then we can embed 
	$C_i^1, C_i^3$ in the second and third $E$'s, and $\neg y_1$ can be embedded in $y_{1f}$ in $D$. Embedding 
	$\neg y_1$  in $y_{1f}$ fixes the assignment $f$ to $y_1$ across all 
	gadgets $E, D$, and all clauses in $Q_2$. Likewise, for a clause   
	 $(x_1 \vee \neg x_4 \vee y_5)$ in $\Phi$, and an assignment $f$ to $x_1$, 
	 $t$ to $x_4$ in the canonical model $G$, we can embed $x_1, \neg x_4$ in the first and second $E$'s and 
	 $y_5$ to the node $y_{5t}$.

	 We will now show
	 correctness, that is: $Q_1 \subseteq_\ani Q_2$ if and only if $\Phi$ is
	 satisfiable.  Let $Q_1 \subseteq_\ani Q_2$. Then there exists
	 an injective homomorphism from $Q_2$ to each \ani-expansion of $Q_1$. The
	 \ani-expansions of $Q_1$ look exactly like $Q_1$ except that each
	 some pairs of vertices non-related via $r$ may have been identified together.
	 
	Let $B$ be an arbitrary \ani-expansion of $Q_1$ and
	$D_B$ the gadget $D$ in $B$. 
	We define $\theta_B(x_i)=1$ if there are two distinct vertices non-related via $r$ accessible via $x_i$, and $\theta_B(x_i)=0$ otherwise.
	Let $h$ be an injective homomorphism mapping an \ani-expansion of $Q_2$ to $B$.
	We furthermore define $\theta_B(y_i)=1$ if $h$ maps
	\yitf to \yit and $\theta_B(y_i)=0$ otherwise, \ie, if
	\yitf is mapped to \yif.  We now show that $\theta_B$ is well-defined and
	satisfies $\varphi$.  It is obvious that each $C_i^j$ will be mapped either to
	the gadget $D_B$ or to $E$ and that for each $i \in \{1,\ldots, m\}$
	exactly one $C_i^j$ is mapped to $D_B$.  
	If $C_i^j$ corresponds to $x_k$, i.e.,
	it is a path labeled $x_k t$, then it can only be mapped into $D_B$
	if $\theta_B(x_k)=1$. Analogously, if $C_i^j$ corresponds to
	$\neg x_k$, it is a path labeled $x_k f$, and can
	therefore only be mapped into $D_B$ if $\theta_B(x_k)=0$.  If
	$C_i^j$ corresponds to $y_k$ or $\neg y_k$, it can always be
	mapped into $D_B$, but since \yktf can only be mapped either to $y_{k,t}$
	or $y_{k,f}$, we can either map positive $y_k$ into $D_B$ or negative
	ones, but not both. Therefore, the definition of $\theta_B(y_k)$ is
	unambiguous, and it indeed satisfies $\varphi$.
	
	Since $B$ is arbitrary, we obtain a choice $y_1, \ldots, y_\ell$ for
	all possible truth-assignments to $x_1,\ldots, x_n$ this
	way. Therefore, $\Phi$ is satisfiable.

	For the only if direction let $\Phi$ be satisfiable. Then we find
	for each truth-assignment to $x_1, \ldots, x_n$ an assignment to
	$y_1, \ldots, y_\ell$ such that
	$\varphi(x_1, \ldots, x_n, y_1, \ldots, y_\ell)$ is true. Let
	$\theta$ be a function that, given the $x_i$, returns an assignment
	for all $y_i$ such that the formula evaluates to true. We will show
	how to map $Q_2$ into an arbitrary \ani-expansion $B$ of $Q_1$.

	Let $B$ and $\theta$ be given. Let $D_B$ be again the gadget $D$ in $B$.
	We use $\theta$ to obtain truth-values for
	$y_1, \ldots, y_\ell$ as follows.  Since this assignment is satisfiable, there
	is a literal that evaluates to \emph{true} in each clause. We map this
	literal to $D_B$ and the others in this clause to gadgets $E$.  If this
	literal is $x_i$, then we can map to the $(x_i\, t\, t)$-path in $D_B$. If it is $\neg x_i$, then we can map to the
	$(x_i\, f)$-path in $D_B$.  If the
	literal is $y_i$, we can map the $(y_i \, t)$-path ending in $\yitf$ to
	$D_B$. This also implies that each $\yitf$ in $Q_2$ is mapped to \yit,
	which is no problem since each path mapped to $E$ can choose freely
	between \yit and \yif and, since $\theta$ is a function, we only
	have either $\theta(y_i)=1$ or $\theta(y_i)=0$.  Analogously, if the
	literal is $\neg y_i$, we can map the $(y_i \, f)$-path ending in
	$\yitf$ to $D_B$, which implies that each $\yitf$ in $Q_2$ is mapped to
	\yif.
\qed

\medskip

Observe that under standard semantics, the right-hand query $Q_2$ in the reduction of Theorem~\ref{thm:crpqfin-crpqfin-lower} above would be in fact equivalent to a CQ, but under \ani semantics this is not the case.

\section{Other results}
\label{app:other-results}
The following theorem summarizes all other complexity results which complete the picture of complexity results of Figure~\ref{fig:summary}, whose proofs can be found below.

\begin{theorem}[Restatement of Theorem~\ref{thm:other}]
    \hfill
    \begin{enumerate}
        \item The \textup{\CQ/\CRPQ} and \textup{\CQ/\CQ} containment problems are \np-complete under query-injective semantics. (Proposition~\ref{prop:cq-crpq-cq-inj})
        \item The \textup{\CQ/\CQ} containment problem under atom-injective semantics is \np-complete. (Corollary~\ref{cor:cq-cq-sp})
        \item The \textup{\CRPQ/\CQ} and \textup{\CRPQfin/\CQ} containment problems are \pitwo-hard, under standard and atom-injective semantics. (Proposition~\ref{prop:crpq-cq-ani-hard})
        \item The \textup{\CRPQ/\CQ} and \textup{\CRPQfin/\CQ} containment problems are in \pitwo, under all semantics. (Proposition~\ref{prop:crpq-cq-all-upper})
        \item The \textup{\CRPQ/\CRPQfin} containment problem is \pspace-hard under all semantics. (Proposition~\ref{prop:crpq-crpqfin-pspace-h})
        \item The \textup{\CRPQ/\CRPQfin} containment problem is in \pspace under standard semantics. (Proposition~\ref{prop:crpq-crpqfin-pspace-st})
        \item The \textup{\CRPQfin/\CRPQ} containment problem is in \pitwo, under all semantics. (Proposition~\ref{prop:crpqfin-crpq-upper})
    \end{enumerate}
\end{theorem}


\begin{proposition}\label{prop:cq-crpq-cq-inj}
	The \textup{\CQ/\CRPQ} and \textup{\CQ/\CQ} containment problems are \np-complete under query-injective semantics.
\end{proposition}
\begin{proof}
For the upper bound, let $Q_1(\bar x)$ be a CQ and $Q_2(\bar y)$ a CRPQ.
Remember that for any $\star \in \set{\ani, \qni}$, we have that $Q_1 \subseteq_\star Q_2$ if{f} $\bar x \in Q_2(Q_1)^\star$, where $Q_1$ is seen as a graph database. By Proposition~\ref{prop:eval-np-c} we then have \np-membership.

The lower bound follows by a direct reduction from the respective evaluation problem for CQ together with Proposition~\ref{prop:eval-np-c}.
\end{proof}


Let us call a homomorphism $h: A \to B$ \defstyle{contracting} if for some $x \xrightarrow{a} y$ in $A$ such that $x \neq y$, we have $h(x)=h(y)$. Observe that the composition of two non-contracting homomorphism is non-contracting.
\begin{lemma}\label{lem:char-CQ-CQ-sp}
For any two \CQ $Q_1, Q_2$, the following are equivalent:
\begin{enumerate}
	\item there is a non-contracting homomorphism $Q_2 \to Q_1$,
	\item $Q_1 \subseteq_\ani Q_2$.
\end{enumerate}
\end{lemma} 
\begin{proof}
	From top to bottom, let $h: Q_2 \to Q_1$ be a non-contracting homomorphism, that is, such that for every atom $x \xrightarrow{a} y $ of $Q_2$ we have $h(x) \neq h(y)$.
	Let $E_1 \in \Exp{\ani} (Q_1)$. Observe that there exists $g : Q_1 \to E_1$ which is not contracting, and hence their composition $g(h) : Q_2 \to E_1$ is non-contracting either. Let $U$ be the conjunction of all atoms $x=y$ such that  $g(h(x)) = g(h(y))$, and observe that $E_2 = (Q_2 \land U)^\collapse \in \Exp{\ani}(Q_2)$. Further, we have $g(h): E_2 \injto E_1$. Summing up, for every $E_1 \in \Exp{\ani} (Q_1)$ there is $E_2 \in \Exp{\ani}(Q_2)$ such that $E_2 \injto E_1$, which by the characterization of Proposition~\ref{prop:cont-char-ainj} proves that $Q_1 \subseteq_\ani Q_2$.
	
	From bottom to top, observe that $Q_1 \in \Exp{\ani}(Q_1)$, and hence by Proposition~\ref{prop:cont-char-ainj} we have that there is some $E_2 \in \Exp{\ani}(Q_2)$ such that $h:E_2 \injto Q_1$, which in particular means that $h$ is non-contracting. On the other hand, as argued before, there must be a homomorphism $g: Q_2 \to E_2$ which is non-contracting. Since the composition of non-contracting homomorphisms yields a non-contracting homomorphism, we obtain that $h(g) : Q_2 \to Q_1$ is non-contracting.
\end{proof}
\begin{corollary}\label{cor:cq-cq-sp}
	The \textup{\CQ/\CQ} containment problem under atom-injective semantics is \np-complete.
\end{corollary}
\begin{proof}
	The upper bound follows from Lemma~\ref{lem:char-CQ-CQ-sp} above. For the lower bound, it is easy to see that the standard reduction from 3-colorability for \CQ under standard semantics \cite{DBLP:conf/stoc/ChandraM77} still applies in this setting.
\end{proof}
\begin{corollary}
	The \textup{$\CQ$/$\CRPQ(A)$} containment problem under simple-path semantics is \np-complete.\footnote{In the jargon of \cite{FigueiraGKMNT20}, $\CRPQ(A)$ are CRPQ whose regular expressions are of the form $a_1 + \dotsb + a_n$.}
\end{corollary}


\begin{proposition}\label{prop:crpq-cq-ani-hard}
    The \textup{\CRPQ/\CQ} and \textup{\CRPQfin/\CQ} containment problems are \pitwo-hard, under standard and atom-injective semantics.
\end{proposition}
\begin{proof}
	The lower bound for standard semantics follows from \cite[Theorem~4.3]{FigueiraGKMNT20}.
	Further, it is easy to see that the \pitwo-hardness proof of \cite[Theorem~4.3]{FigueiraGKMNT20} goes through for atom-injective semantics. This is because, in the reduction, all atoms of queries contain languages of words of length 1, and they have no self-loops. Indeed, as a consequence of Lemma~\ref{lem:char-CQ-CQ-sp}, under such restrictive conditions the \ani and standard containment problems coincide.
\end{proof}


\begin{proposition}\label{prop:crpq-cq-all-upper}
    The \textup{\CRPQ/\CQ} and \textup{\CRPQfin/\CQ} containment problems are in \pitwo, under all semantics.
\end{proposition}
\begin{proof}
	 Let $\star \in \set{st,\qni,\ani}$.
	Given $Q_1, Q_2$, let $N$ be the number of atoms of $Q_2$. Consider the set $S$ of all expansions of $Q_1$ where every atom expansion $x \xrightarrow{w} y$ of size greater than $2N$ is replaced with $x \xrightarrow{u \# v} y$, where $\#$ is a fresh symbol, and $u$ \resp{$v$} is the $N$-prefix \resp{$N$-suffix} of $w$. Observe that every element of $S$ is of polynomial size, and that we can check in polynomial time whether any given polysized CQ is in $S$. Consider the following \sigmatwo algorithm for non-containment. 
	We check that there exists some connected component $\hat Q_2$ of $Q_2$ and element $E_1^{\#} \in S$ such that the following two conditions hold:
	\begin{enumerate}[(i)]
		\item  $E_1^{\#} \not\subseteq_\star \hat Q_2$ (which is in co-\np due to Proposition~\ref{prop:cq-crpq-cq-inj} for \qni 
		and \cite[Theorem~4.2]{FigueiraGKMNT20} for standard);
		\item \label{eq:item:lem:crpq-cq-pitwo}  for each atom $x \xrightarrow{u \# v} y$ of $E_1^{\#}$ associated to an atom $x \xrightarrow{L} y$ of $Q_1$, there is no $w$ such that (a) $u w v \in L$ and (b) $E_1^w \subseteq_\star \hat Q_2$, where $E_1^w() = x \xrightarrow{ u w v} y$.
	\end{enumerate}
	Since $\hat Q_2$ is connected it has to be mapped through a homomorphism \resp{injective homomorphism, \ani homomorphism} either to the $N$-neighbourhood of a variable of $Q_1$, or entirely inside an atom expansion. The two items above ensure that none of these cases can occur, and hence that there exists a counter-example for the containment $Q_1 \subseteq_\star Q_2$.
	Observe that in item \eqref{eq:item:lem:crpq-cq-pitwo}, if a $\star$-expansion of $\hat Q_2$ maps into a directed path, it means that $\hat Q_2$ is $\star$-equivalent to a directed path, and hence that $w$ can be taken of polynomial size. This, in turn, means that \eqref{eq:item:lem:crpq-cq-pitwo} can be done in co-\np.
\end{proof}


\begin{proposition}\label{prop:crpq-crpqfin-pspace-h}
	The \textup{\CRPQ/\CRPQfin} containment problem is \pspace-hard under all semantics.
\end{proposition}
\begin{proof}
	First observe that for Boolean queries $Q_1,Q_2$ of the form $Q_i() =  x \xrightarrow{L_i} y$, we have $Q_1 \subseteq_{st} Q_2$ if{f} $Q_1 \subseteq_{\ani} Q_2$ if{f} $Q_1 \subseteq_{\qni} Q_2$. In \cite[Theorem~4.5]{FigueiraGKMNT20} it was shown that the containment problem (under standard semantics) for this kind of queries is \pspace-hard, even when $L_2$ is a star-free expression and the alphabet is of fixed size.
\end{proof}


\begin{proposition}\label{prop:crpq-crpqfin-pspace-st}
	The \textup{\CRPQ/\CRPQfin} containment problem is in \pspace under standard semantics.
\end{proposition}
\begin{proof}
	Given $Q_1, Q_2$, let $N$ be the maximum number of atoms of an expansion of $Q_2$. 
	Consider the set $S$ of all expansions of $Q_1$, where every expansion  $x \xrightarrow{w} y$ of an atom $x \xrightarrow{L} y$ thereof such that $|w| > 2N$  is replaced with $u \cdot \# \cdot v$, where $u$ \resp{$v$} is the $N$-prefix \resp{$N$-suffix} of $w$, and $\#$ is a fresh symbol.
	The \pspace algorithm then guesses an element $E_1^{\#}$ of $S$ and checks whether there exists some expansion $E_1$ of $Q_1$ from which $E_1^{\#}$ could be obtained such that no expansion $E_2$ of $Q_2$ can be homomorphically mapped to $E_1$. For this, we check that there exists some connected component $\hat Q_2$ of $Q_2$ such that the following two conditions hold:
	\begin{enumerate}[(i)]
		\item \label{eq:item:lem:crpq-crpqfin-pspace-st:a} $E_1^{\#} \not\subseteq \hat Q_2$ (which is in co-\np \cite[Theorem~4.2]{FigueiraGKMNT20});
		\item \label{eq:item:lem:crpq-crpqfin-pspace-st}  for each path of the form $x \xrightarrow{u \# v} y$ in $E_1^{\#}$ associated with the expansion of an atom $x \xrightarrow{L} y$ of $Q_1$, there is no $w$ such that:  (a) $u w v \in L$ and  (b) $E_1^w \subseteq \hat Q_2$, where $E_1^w() = x \xrightarrow{ u w v} y$.
	\end{enumerate}
	Since $\hat Q_2$ is connected it needs to be mapped either to the $N$-neighbourhood of a variable of $Q_1$ (ruled out by item \ref{eq:item:lem:crpq-crpqfin-pspace-st:a}), or entirely inside an atom expansion (ruled out by item \ref{eq:item:lem:crpq-crpqfin-pspace-st}). On the other hand, $E_1 \not\subseteq Q_2$ if{f} $E_1 \not\subseteq \hat Q_2$ for some component $\hat Q_2$.
	These two items hence ensure that none of these cases can occur, and that there exists a counter-example for the containment $Q_1 \subseteq Q_2$.
	Observe that item \eqref{eq:item:lem:crpq-crpqfin-pspace-st} can be seen as an instance of the intersection emptiness problem for regular languages, that is, the problem of whether $\bigcap_{i \in I} L_i = \emptyset$ for a given set $\set{L_i}_{i \in I}$ of regular expressions, which is a \pspace-complete problem \cite{Kozen77}.
\end{proof}



\begin{proposition}\label{prop:crpqfin-crpq-upper}
	The \textup{\CRPQfin/\CRPQ} containment problem is in \pitwo, under all semantics.
\end{proposition}
\begin{proof}
	Let $\star \in \set{st,\ani,\qni}$ and let $Q_1(\bar x),Q_2(\bar x)$ be an instance. One can test non-containment by guessing a $\star$-expansion $E_1(\bar x) \in \Exp{\star}(Q_1(\bar x))$ (of linear size) and test that $\bar x \not\in Q_2(E_1)$ under $\star$ semantics, which is in co-\np by Proposition~\ref{prop:eval-np-c}.
\end{proof}
\fi
\end{document}